\documentclass{article}
\usepackage{graphicx} 
\usepackage{amsthm}
\usepackage{amsfonts}
\usepackage{amssymb}
\usepackage{amsmath}
\usepackage{dsfont}
\usepackage{fullpage}
\usepackage[english]{babel}
\usepackage[T1]{fontenc}
\usepackage{multirow}
\usepackage{multicol}
\usepackage{blkarray}
\usepackage{float}
\usepackage{biblatex}
\usepackage{placeins}
\bibliography{main.bib}

\usepackage{lscape}
\usepackage{mathtools}
\usepackage{enumitem}
\usepackage{subcaption}
\usepackage[ruled,vlined,linesnumbered, onelanguage]{algorithm2e}
\DontPrintSemicolon

\usepackage{tikz}
\usepackage{xcolor}
\usepackage{colortbl}
\usepackage{arydshln}
\usepackage{etoolbox}
\usepackage{enumitem}
\setlist[itemize]{label=\textbullet}

\usetikzlibrary{patterns.meta,calc}
\usetikzlibrary{shapes,arrows,shadows,positioning}
\usetikzlibrary {shapes.multipart}
\usetikzlibrary{decorations.pathreplacing}
\usetikzlibrary{patterns}
\usetikzlibrary{matrix}
\usetikzlibrary{chains}
\usetikzlibrary{shapes.misc}

\usepackage{euler} 
\usepackage{palatino} 
\usepackage{fourier-orns} 

\usepackage{pifont}
\newcommand{\xmark}{\ding{55}}%

\definecolor{shadethmcolor}{rgb}{0.96,0.93,1.0} 
\definecolor{ourpurple}{rgb}{0.60,0.50,0.71}

\usepackage{mdframed}

\newmdtheoremenv[backgroundcolor=shadethmcolor, %
innertopmargin =0pt , %
innerleftmargin=0pt,%
topline=false,bottomline=false,leftline=false,rightline=false,]{theorem}{Theorem}

\newmdtheoremenv[backgroundcolor=shadethmcolor, %
innertopmargin =0pt , %
innerleftmargin=0pt,%
topline=false,bottomline=false,leftline=false,rightline=false,]{example}{Example}

\newmdtheoremenv[backgroundcolor=shadethmcolor, %
innertopmargin =0pt , %
innerleftmargin=0pt,%
topline=false,bottomline=false,leftline=false,rightline=false,]{proposition}{Proposition}

\newmdtheoremenv[backgroundcolor=shadethmcolor, %
innertopmargin =0pt , %
innerleftmargin=0pt,%
topline=false,bottomline=false,leftline=false,rightline=false,]{lemma}{Lemma}

\newmdtheoremenv[backgroundcolor=shadethmcolor, %
innertopmargin =0pt , %
innerleftmargin=0pt,%
topline=false,bottomline=false,leftline=false,rightline=false,]{corollary}{Corollary}

\newmdtheoremenv[backgroundcolor=shadethmcolor, %
innertopmargin =0pt , %
innerleftmargin=0pt,%
topline=false,bottomline=false,leftline=false,rightline=false,]{remark}{Remark}

\newmdtheoremenv[backgroundcolor=shadethmcolor, %
innertopmargin =0pt , %
innerleftmargin=0pt,%
topline=false,bottomline=false,leftline=false,rightline=false,]{definition}{Definition}

\newmdtheoremenv[backgroundcolor=shadethmcolor, %
innertopmargin =0pt , %
innerleftmargin=0pt,%
topline=false,bottomline=false,leftline=false,rightline=false,]{conjecture}{Conjecture}

\newcommand{\ch}[1]{\textnormal{\texttt{#1}}}

 \renewenvironment{proof}{{\noindent\bfseries Proof.}}{\hfill{\color{ourpurple}{\decoone}}}

\newcommand{\textnumbering}[2][1]{%
  \let\nodecontents\empty%
  \let\counters\empty%
  \begin{tikzpicture}%
  \foreach \x[count=\currentX from 0] in {#2} {
    \pgfmathtruncatemacro\number{\currentX+#1}
    \expandafter\gappto\expandafter\nodecontents\expandafter{\x\&}
    \expandafter\gappto\expandafter\counters\expandafter{\number\&}
  }
  \matrix[inner sep=0pt, matrix of nodes, ampersand replacement=\&, anchor=base, nodes={inner sep=0pt},
  row 1/.style={font=\tiny}]
  {\counters\\[2pt]
  \nodecontents\\};
\end{tikzpicture}%
}

\usepackage{hyperref}
\hypersetup{
colorlinks=true,
breaklinks=true,
urlcolor= ourpurple,
linkcolor= ourpurple,
citecolor= ourpurple,
}

\colorlet{lblue}{blue!50!white}
\colorlet{lred}{red!50!white}
\colorlet{lgreen}{green!50!white}
\colorlet{lpurple}{purple!50!white}
\colorlet{lorange}{orange!50!white}
\colorlet{lpink}{pink!50!white}
\colorlet{lbrown}{brown!50!white}
\colorlet{lyellow}{yellow!50!white}
\colorlet{lolive}{olive!50!white}

\DeclareMathOperator{\rc}{\text{rc}}

\title{On the number of $k$-mers admitting\\ a given lexicographical minimizer}
\author{Florian Ingels\\\url{florian.ingels@univ-lille.fr} 
\and 
Camille Marchet\\ \url{camille.marchet@univ-lille.fr}
\and
Mikaël Salson\\ \url{mikael.salson@univ-lille.fr}}
\date{Univ. Lille, CNRS, Centrale Lille, UMR 9189 CRIStAL, F-59000 Lille, France}

\begin{document}
\maketitle

\begin{abstract}
\noindent
The minimizer of a word of size $k$ (a $k$-mer) is defined as its smallest substring of size $m$ (with $m\leq k$), according to some ordering on $m$-mers. minimizers have been used in bioinformatics --- notably --- to partition sequencing datasets, binning together $k$-mers that share the same minimizer. It is folklore that using the lexicographical order lead to very unbalanced partitions, resulting in an abundant literature devoted to devising alternative orders for achieving better balanced partitions. To the best of our knowledge, the unbalanced-ness of lexicographical-based minimizer partitions has never been investigated from a theoretical point of view. In this article, we aim to fill this gap and determine, for a given minimizer, how many $k$-mers would admit the chosen minimizer --- i.e. what would be the size of the bucket associated to the chosen minimizer in the worst case, where all $k$-mers would be seen in the data.
We show that this number can be computed in $O(km)$ space and $O(km^2)$ time. We further introduce approximations that can be computed in $O(k)$ space and $O(km)$ time. We also show on genomic datasets that the practical number of $k$-mers associated to a minimizer are closely correlated to the theoretical expected number. We introduce two conjectures that could help closely approximating the total number of $k$-mers sharing a minimizer. We believe that characterising the distribution of the number of $k$-mers per minimizer will help devise efficient lexicographic-based minimizer bucketting.
\end{abstract}


\section{Introduction}

\subsection{Context \& motivations}

$k$-mers, fragments of sequences of fixed size $k$, are a key tool in bioinformatics for processing and analysing data from DNA or RNA sequencing, and have been used successfully to tackle problems such as genome assembly, sequence alignment, species assignment, gene quantification, and so on \cite{marcais_sketching_2019,jenike2024guide}. Technological advances are resulting in increasingly larger datasets to process, which are growing faster than the computing capacity of computers \cite{schatz2010cloud,stephens2015big,katz2022sequence} --- requiring the implementation of ever more powerful methods and algorithms. In particular, even storing data on disk or loading it into memory --- before any processing --- is a challenge in itself. Consider for instance the Sequence Read Archive (SRA) \footnote{\url{https://www.ncbi.nlm.nih.gov/sra/}, accessed in December 2024.} and the European Nucleotide Archive (ENA)\footnote{\url{https://www.ebi.ac.uk/ena/browser/home}, accessed in December 2024.}, two examples of public databases that index tens of petabytes of sequencing data \cite{sra,enapeta}.

A common strategy for addressing this problem is to partition the $k$-mers into bins, and process/store each bin separately. A popular approach to partitioning $k$-mers is to group together those that share a common \emph{minimizer}. Introduced first in 2003 in \cite{schleimer2003winnowing}, this technique is based on the following principle: for a fixed $m\leq k$, each $k$-mer is subdivided into its constitutive $m$-mers. Provided a total order over the set of all possibles $m$-mers, e.g. the lexicographical order, the minimizer of a $k$-mer is defined as its smallest $m$-mer. Consecutive $k$-mers\footnote{I.e. $k$-mers that overlap over $k-1$ characters.} in the data often share the same minimizer, guaranteeing some form of locality preservation in partitioning. In addition, we can define super-$k$-mers, i.e. sequences obtained by concatenating all consecutive $k$-mers sharing the same minimizer. These super-$k$-mers are stored in bins corresponding to their minimizer, for a fraction of the space needed to store the corresponding $k$-mers directly \cite{li2015mspkmercounter,deorowicz2015kmc}. Very recent work has even taken this principle all the way to defining hyper-$k$-mers, see \cite{martayan2024hyper}.

In the context of partitioning sequences using minimizers, many approaches in the literature rely on a common heuristic framework~\cite{pibiri2023locality,marchet2023scalable,marchet2021blight}. 
The process begins by selecting a minimizer using a primary hash-based method. 
Once the minimizer is chosen, its sequence is hashed using a reversible hash function, and the result is taken modulo $b$, the total number of partitions, to assign it to a specific partition.
This process is used to address the observed imbalance happening in practice when $k$-mers are partitioned with super-$k$-mers. By allowing different minimizers to be mapped to the same partition, it distributes $k$-mers more evenly across all partitions.

\begin{figure}[ht!]
    \centering
\begin{subfigure}{0.49\textwidth}
\centering
\includegraphics[width=\textwidth]{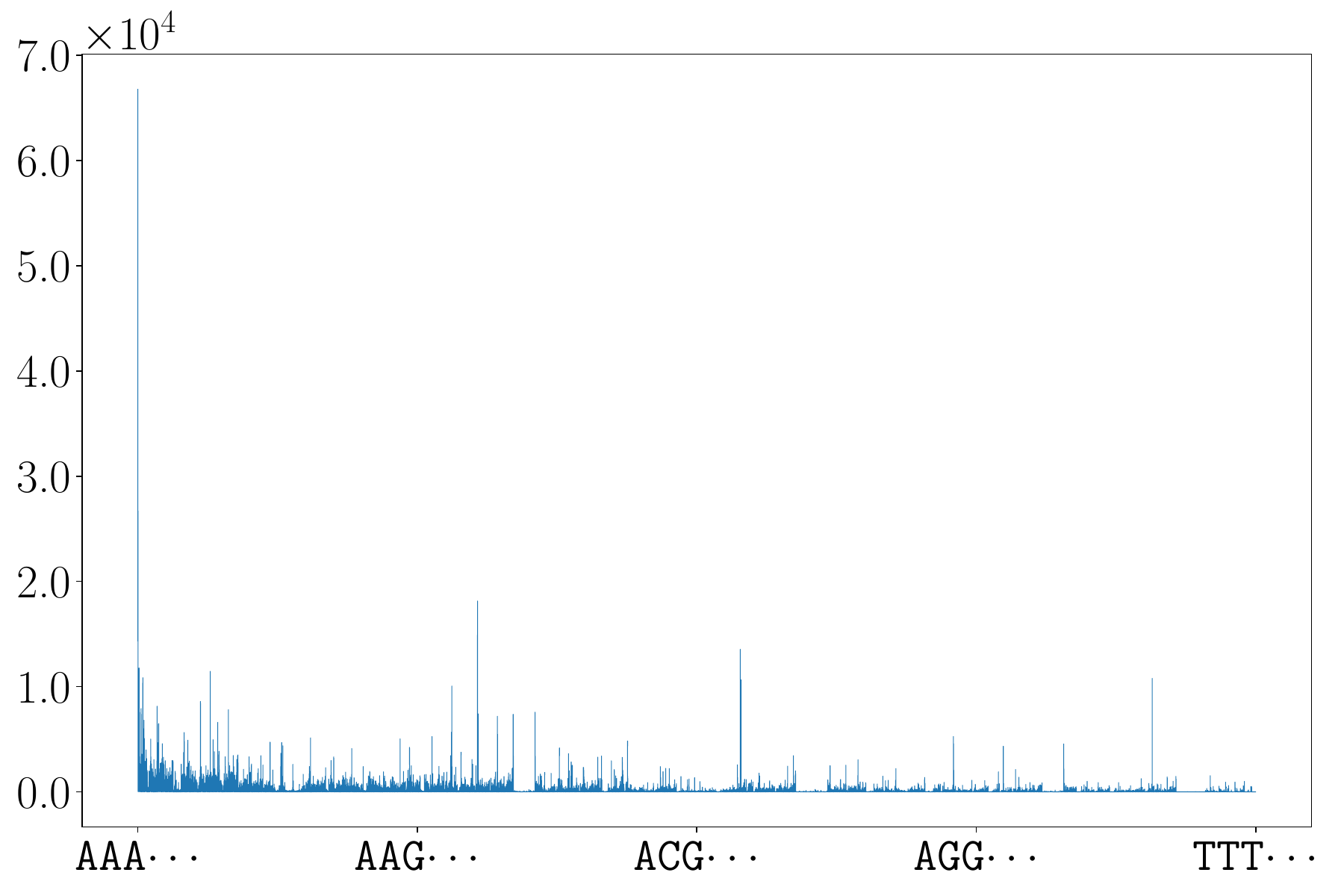}
    \caption{Chromosome Y, $k=31, m=10$}
    \label{fig:empirical_partition:a} 
\end{subfigure}\hfill
\begin{subfigure}{0.49\textwidth}
\centering
\includegraphics[width=\textwidth]{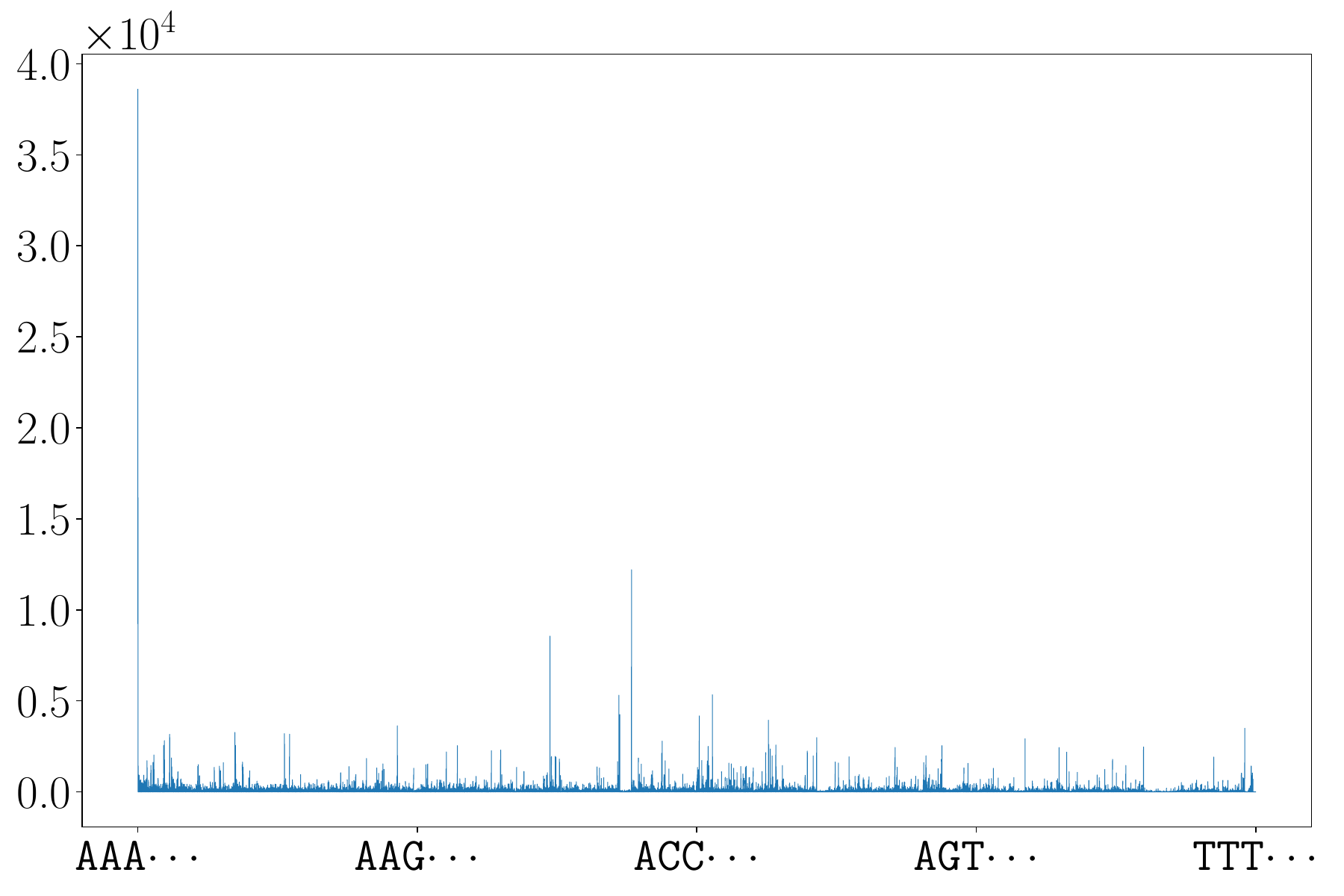}
    \caption{Chromosome 1, $k=31, m=21$}
    \label{fig:empirical_partition:b} 
\end{subfigure}

\begin{subfigure}{0.49\textwidth}
\centering
\includegraphics[width=\textwidth]{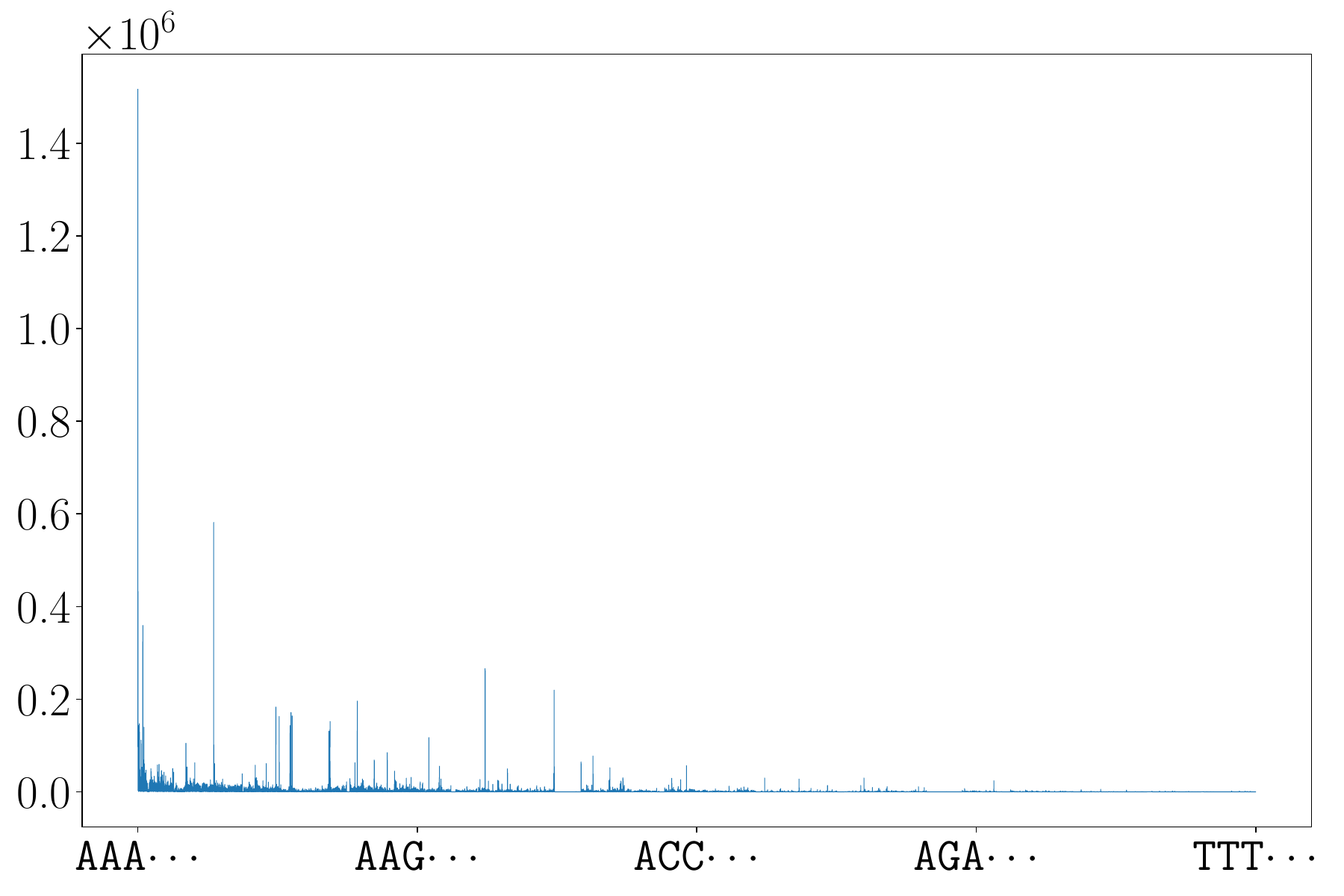}
    \caption{RNA fusion, $k=61, m=10$}
    \label{fig:empirical_partition:c} 
\end{subfigure}\hfill
\begin{subfigure}{0.49\textwidth}
\centering
\includegraphics[width=\textwidth]{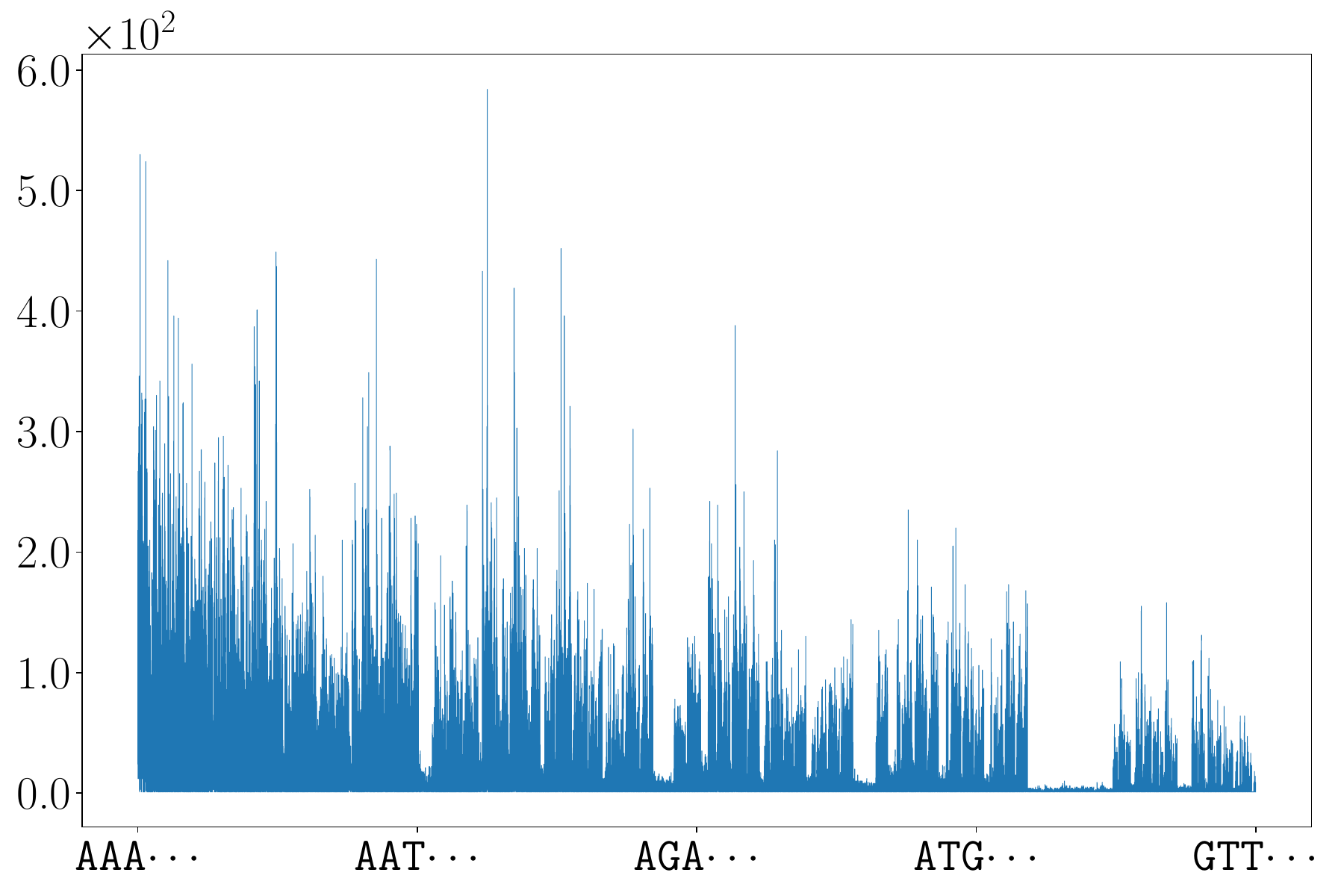}
    \caption{\emph{Escherichia Coli}, $k=21, m=10$}
    \label{fig:empirical_partition:d} 
\end{subfigure}
    \caption{Examples of empirical partitions obtained on different datasets --- see Section~\ref{ss:imabalance} for details. The $x$-axis depicts the observed minimizers, sorted by lexicographical order, and the $y$-axis displays the number of $k$-mers found in the partition associated with each minimizer for that particular dataset.}
    \label{fig:empirical_partition}
\end{figure}

As noted in \cite{roberts2004reducing}, using the lexicographic order for $m$-mers generates undesirable effects as the partitions obtained are empirically highly skewed, where one would rather wish for a fairly balanced subdivision of the data space. Several examples of empirical imbalance are showed in Figure~\ref{fig:empirical_partition}. This problem has been widely addressed in the literature, and several approaches have been proposed. Some suggest modifying the lexicographical order, for example by reversing the order of certain letters according to their observed frequency in the data \cite{roberts2004reducing}, or by giving lower priority to minimizers that start with a certain combination of letters or have a particular pattern \cite{deorowicz2015kmc}. Others use random orders, obtained with hash functions, which counter the effect of lexicographic order, but without any guarantee that we will not find, nonetheless, highly repeated $m$-mers with low order, which would lead to partitions that are again very unbalanced, as remarked in \cite{chikhi2014representation}. There are also methods that order $m$-mers according to their frequency of appearance in the data \cite{chikhi2014representation} --- or in a sample of the data \cite{nystrom2021compact,flomin2022data}. Finally, compact universal hitting sets --- small sets of $m$-mers for which it is guaranteed that any $k$-mer contains at least one of them --- have recently been introduced \cite{orenstein2016compact,ekim2020randomized}. A comparative study of most of these methods can be found in \cite{marccais2017improving}. For practical applications of minimizers, we refer the interested reader to the recent review in \cite{ndiaye2024less}.

Density, which measures the expected fraction of distinct minimizers sampled over a sequence, plays an important role when minimizers are used for sampling, alignment, or creating efficient $k$-mer set representations.  The density of random orders is still being explored, with first results in 2003 in \cite{schleimer2003winnowing} until very recently \cite{golan2024expected}. To the best of our knowledge, density considerations are the only theoretical works addressing the behaviour of lexicographic minimizers~\cite{zheng2020improved}. However, when minimizers are applied for partitioning $k$-mers, an independent question arises: how challenging is it to achieve balanced partitioned buckets? In practice, heuristics often exclude certain minimizers, particularly those with low complexity~\cite{deorowicz2015kmc} or ``randomize'' lexicographic minimizers in a reversible way~\cite{wood2019improved}, as they tend to attract an excessive number of $k$-mers. Consequently, most methods doing partitions have moved away from using lexicographic minimizers and instead rely on random minimizers~\cite{chikhi2016compacting,marchet2023scalable,fan2024fulgor}.

In this article, we open a new perspective to better understand lexicographic partitions. We are interested in precisely quantifying the size of partitions in the worst possible case, i.e. if all $k$-mers admitting the same minimizer were to be observed in the data. More precisely, provided a certain $m$-mer $w$, we are interested in evaluating the number of $k$-mers admitting $w$ as their minimizer.

We believe that a better theoretical understanding of the distribution of lexicographical minimizers would make it possible to design new partition heuristics to compensate for the imbalances observed, and give lexicographical minimizers a new opportunity for practical use. Indeed, using lexicographic minimizers offers at least two key advantages. First, since the minimizer sequence is known, it eliminates the need for explicit storage, saving memory and reducing the number of bits written. Second, when minimizers are used for tasks like sampling, their sequences can be directly compared instead of relying on hash values. Comparing lexicographic signals leads to more precise comparisons, avoiding hash collisions and preserving the true relationships between sequences, with recent works showing it improves accuracy in measuring sequence distances~\cite{rouze2023fractional,greenberg2023lexichash}.

\subsection{Preliminaries}

Let $\Sigma$ be a totally ordered set, called the alphabet, whose elements are called letters. As a convention, indeterminate letters will be noted $a,b,c,\dots$, while determined letters will be noted in small capitals $\ch{A},\ch{B},\ch{C},\dots$. A word over $\Sigma$ is a finite sequence of letters $w=a_1\cdots a_n$, with $a_i\in\Sigma$. As a convention, words will be denoted $w,x,\dots$. The size of $w$, denoted by $|w|$, is equal to $n$. A word of size $k>0$ is called a $k$-mer, and the set of all $k$-mers over $\Sigma$ is denoted by $\Sigma^k$. We recall that the lexicographical order over $\Sigma^k$ is defined as follows. Let $x=a_1\cdots a_k$ and $y=b_1\cdots b_k$ be two $k$-mers; then $x>y$ if and only if either (i) $a_1>b_1$ or (ii) there exists $1\leq i\leq k-1$ so that $a_1\cdots a_i = b_1\cdots b_i$ and $a_{i+1}>b_{i+1}$. We denote by $\varepsilon$ the empty word, so that $\varepsilon < x$ for any $k$-mer $x$; in particular, $\varepsilon < a$ for any letter $a\in \Sigma$. Finally, for $a\in \Sigma$, we define $\varphi_>(a)$ the number of letter strictly greater than $a$; i.e. $\varphi_>(a) = |\lbrace a' \in \Sigma : a'>a\rbrace|$. By convention, $\varphi_>(\varepsilon)=|\Sigma|$. Besides, $\varphi_> : \Sigma \cup \lbrace \varepsilon \rbrace \to [\![ 0, |\Sigma| ]\!]$ is a bijection.

Since the motivation for this work comes from bioinformatics, we shall use the DNA alphabet $\Sigma = \lbrace \ch{A}, \ch{C},\ch{G},\ch{T} \rbrace$ for our examples. However, we stress that the results presented here are valid whatever the alphabet chosen. In particular, these results are also valid if the order associated with the alphabet is not the standard lexicographical order, but an arbitrary order on the letters. In particular, from the perspective of bioinformatics applications, this means that using the order $\ch{C} < \ch{A} < \ch{T} < \ch{G}$ as proposed by \cite{roberts2004reducing} in no way prevents the observation of unbalanced partitions in practice, as discussed in the introduction.


In this paper, we focus on the notion of lexicographical minimizers of $k$-mers. Let $1\leq m \leq k$ and let $x=a_1\cdots a_k$ be a $k$-mer.

\begin{definition}
The minimizer of $x$, denoted by $\min(x)$, is the smallest $m$-mer contained in $x$.
\end{definition}

In other words, for any $1\leq i \leq k-m+1$, $\min(x) \leq a_i \cdots a_{i+m-1}$. We denote $f_{\min} : \Sigma^k \to [\![ 1,k-m+1]\!]$ the function that returns the smallest index $i$ so that $\min(x) = a_i\cdots a_{i+m-1}$ --- i.e.

$$f_{\min}(x) = \min \lbrace 1\leq i\leq k-m+1 : a_i\cdots a_{i+m-1} \leq a_{i'}\cdots a_{i'+m-1}, \forall 1\leq i'\leq k-m+1\rbrace$$ 

and $\min(x) = a_{f_{\min}(x)}\cdots a_{f_{\min}(x)+m-1}$. In the literature, the function $f_{\min}$ can be defined for any total order over $\Sigma^k$, not necessarily the lexicographical order, and is called a minimizer scheme \cite{roberts2004reducing,schleimer2003winnowing}. 

Provided $1\leq m \leq k$ and a $m$-mer $w$, we denote by $\Sigma^k_w$ the set of $k$-mers that admit $w$ as their minimizer; i.e. $\Sigma^k_w = \lbrace x\in \Sigma^k : \min(x)=w\rbrace$. Note that $\Sigma^k_w$ is never empty, since it contains at least the $k$-mer formed by $k-m$ copies of $\max(\Sigma)$ followed by $w$. Therefore, we can partition $\Sigma^k$ into the following disjoint union
$$\Sigma^k = \bigsqcup_{w\in \Sigma^m} \Sigma^k_w.$$

If $k$-mers were to be equally distributed over the $\Sigma^k_w$'s, we would have $|\Sigma_w^k| = |\Sigma|^{k-m}$. It is folklore that this is not the case and that this partition is very unevenly distributed, as stated in the introduction and illustrated in Figure~\ref{fig:empirical_partition}. However, to the best of our knowledge, nobody seems to have actually determined theoretically to what extent. We introduce the following counting function, aimed to fill this gap. 

\begin{definition}\label{def:minimizer_counting_function}
The minimizer counting function $\pi_k : \Sigma^m \to [\![1,|\Sigma|^k]\!]$ is defined, for $w\in\Sigma^m$, by $\pi_k(w) = |\Sigma^k_w|$.
\end{definition}

The rest of the paper is organised as follows:
\begin{itemize}
    \item Section~\ref{sec:computing},~\ref{sec:antemers} and~\ref{sec:postmers} are devoted to the computation of $\pi_k(w)$, by establishing systems of recursive equations that can be calculated using dynamic programming;
    \item Section~\ref{sec:numerical} offers a number of numerical illustrations, in particular comparing theoretical values with those observed in practice;
    \item a conclusion outlining a number of future research directions.
\end{itemize}

In this paper, we use Iverson brackets \cite{iverson1962programming} to denote indicator functions; that is, the Iverson bracket of property $P$, denoted by $[P]$, is defined as
$$[P] = \begin{cases} 1 & \text{if } P \text{ is true};\\
0 & \text{otherwise}.
\end{cases}$$

\section{Computing $\pi_k(w)$}\label{sec:computing}

In this section, we assume that $1\leq m\leq k$ are fixed, as well as the $m$-mer $w  = a_1\cdots a_m$. Our goal is to compute $\pi_k(w)$, that is, as per Definition~\ref{def:minimizer_counting_function}, the number of $k$-mers that admits $w$ as their minimizer. In other words,
$$\pi_k(w) = \left| \Sigma_w^k\right| = \left| \lbrace x\in \Sigma^k : \min(x) = w \rbrace\right|.$$

\subsection{General strategy}

\begin{minipage}[c]{0.475\textwidth}
Let $x$ be a $k$-mer, so that $\min(x)=w$. There exists (possibly empty) words $y$ and $z$, so that $x = ywz$, and with $f_{\min}(x)= |y|+1$. By definition of $f_{\min}$, all $m$-mers of $x$ preceding $w$ are necessarily $>w$; whereas all $m$-mers succeeding $w$ are necessarily $\geq w$ --- as illustrated in Figure~\ref{fig:anatomy_kmer}. To respect these constraints, it should be noted that the letters in $y$ are dependent on each other, as are those in $z$; however, the letters in $y$ and $z$ are independent of each other. Finally, $|y|+|z| = k-m$.
\end{minipage}\hfill
\begin{minipage}[c]{0.475\textwidth}
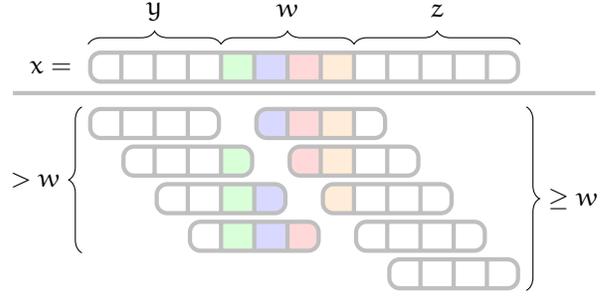
\begin{figure}[H]
    \centering
\begin{tikzpicture}
\tikzstyle{kmer}=[rectangle,rectangle split, rectangle split horizontal,rectangle split parts=13,ultra thick,draw=lightgray,rounded corners]

\tikzstyle{mmer}=[rectangle,rectangle split, rectangle split horizontal,rectangle split parts=4,ultra thick,draw=lightgray,rounded corners]

\tikzstyle{c4}=[rectangle split part fill={white,white,white,white,lgreen!30,lblue!30,lred!30,lorange!30,white}]

\node[kmer,c4] (s) at (0,0) {\nodepart{one} \nodepart{two} \nodepart{three}  \nodepart{four}  \nodepart{five} \nodepart{six} \nodepart{seven} \nodepart{eight} \nodepart{nine} \nodepart{ten} };

\node at ([xshift=-0.5cm]s.west) {$x = $};

\draw [decorate,decoration={brace,amplitude=5pt,raise=0.5ex}]
  (s.north west) -- (s.four split north) node[midway,yshift=1.5em]{$y$};

\draw [decorate,decoration={brace,amplitude=5pt,raise=0.5ex}]
  (s.eight split north) -- (s.north east) node[midway,yshift=1.5em]{$z$};

  \draw [decorate,decoration={brace,amplitude=5pt,raise=0.5ex}]
  (s.four split north) -- (s.eight split north) node[midway,yshift=1.5em]{$w$};

\draw[ultra thick,lightgray] ([shift={(-1cm,-0.125cm)}]s.south west)--([shift={(1cm,-0.125cm)}]s.south east);

\node[mmer] (m1) at ([yshift=-0.75cm]s.two split) {};
\node[mmer,rectangle split part fill={white,white,white,lgreen!30}] (m2) at ([yshift=-1.25cm]s.three split) {};
\node[mmer,rectangle split part fill={white,white,lgreen!30,lblue!30}] (m3) at ([yshift=-1.75cm]s.four split) {};
\node[mmer,rectangle split part fill={white,lgreen!30,lblue!30,lred!30}] (m4) at ([yshift=-2.25cm]s.five split) {};
\node[mmer,rectangle split part fill={lblue!30,lred!30,lorange!30,white}] (m5) at ([yshift=-0.75cm]s.seven split) {};
\node[mmer,rectangle split part fill={lred!30,lorange!30,white}] (m6) at ([yshift=-1.25cm]s.eight split) {};
\node[mmer,rectangle split part fill={lorange!30,white}] (m7) at ([yshift=-1.75cm]s.nine split) {};
\node[mmer] (m8) at ([yshift=-2.25cm]s.ten split) {};
\node[mmer] (m9) at ([yshift=-2.75cm]s.eleven split) {};

\draw [decorate,decoration={brace,amplitude=5pt,raise=0.5ex}]
  ([yshift=2cm]m9.north east) -- (m9.south east) node[midway,xshift=2em]{$\geq w$};

  \draw [decorate,decoration={brace,amplitude=5pt,mirror,raise=0.5ex}]
  (m1.north west) -- ([yshift=-1.5cm]m1.south west) node[midway,xshift=-2em]{$> w$};
\end{tikzpicture}
    \caption{Anatomy of a $k$-mer $x$ so that $\min(x)=w$.}
    \label{fig:anatomy_kmer}
\end{figure}
\end{minipage}

So, it appears that determining $\pi_k(w)$ is a matter of counting how many words $y$ and words $z$ there are for each of the possible starting positions of $w$ in a $k$-mer. We introduce the concepts of antemers (corresponding to $y$, placed before $w$) and postmers (corresponding to $z$, placed after $w$) to this end. It is clear that both the possible antemers and postmers depend strongly on $w$; in the sequel, this dependency on $w$ will be implicitly assumed, and therefore omitted from the notations in order to simplify them.

\begin{definition}\label{def:antemer}
An $\alpha$-antemer is a word $y\in \Sigma^\alpha$ so that, for any $m$-mer $w'$ of the word $yw$ but the last one (which is $w$ itself), we have $w'>w$.
\end{definition}

We denote by $A(\alpha)$ the counting function of $\alpha$-antemers --- i.e. the number of $\alpha$-mers that are antemers. Also, for $0\leq i < m$, we denote by $A_i(\alpha)$ the number of $\alpha$-antemers that share with $w$ a prefix of size $i$. Note that $i=m$ is forbidden, otherwise there would be a $m$-mer equal to $w$ in the antemer. Naturally, we have
\begin{equation}\label{eq:antemer_def}
A(\alpha) = \sum_{i=0}^{m-1} A_i(\alpha)
\end{equation}
and $A(0)=1$. Note also that $A_i(\alpha)=0$ for $i > \alpha$.

\begin{definition}\label{def:postmer}
A $\beta$-postmer is a word $z\in\Sigma^\beta$ so that, for any $m$-mer $w'$ of $z$, we have $w'\geq w$.
\end{definition}

\begin{remark}\label{rmk:postmer_def}
If we take $w=\ch{ACA}$ and $z=\ch{ACC}$, for example, then $z$ is a postmer according to this definition, despite the word $wz=\ch{ACAACC}$ containing the $m$-mer \ch{AAC} which is $<w$. If we have not defined postmers by requiring that all $m$-mers of the word $wz$ --- instead of $z$ as written --- be $\geq w$, this is to simplify certain notations in the sequel, when recursively computing the number of postmers. To reconcile this concern for simplification with the issue mentioned above, in practice we will consider $(\beta+m)$-postmers of the form $wz$, i.e. whose prefix is fixed and equal to $w$ --- while $z$ remains to be determined.
\end{remark}

We denote by $P(\beta)$ the counting function of $\beta$-postmers --- \textit{i.e.} the number of $\beta$-mers that are postmers. For $0\leq i \leq m$, we also denote by $P_i(\beta)$ the number of $\beta$-postmers that share with $w$ a prefix of size $i$. Again, we have
\begin{equation}\label{eq:postmer_def}
P(\beta) = \sum_{i=0}^m P_i(\beta)
\end{equation}
and $P(0)=1$. Note also that $P_i(\beta)=0$ for $i > \beta$.

An example of antemers and postmers is provided in Figure~\ref{fig:antemer_example}.

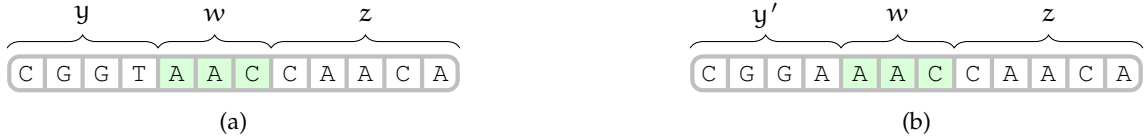
\begin{figure}[H]
    \centering
\begin{subfigure}{0.45\textwidth}
\centering
    \begin{tikzpicture}
    \tikzstyle{kmer}=[rectangle,rectangle split, rectangle split horizontal,rectangle split parts=12,ultra thick,draw=lightgray,rounded corners,rectangle split part fill={white,white,white,white,lgreen!30,lgreen!30,lgreen!30,white}]
    
    \node[kmer] (s) at (0,0) {\nodepart{one} \ch{C} \nodepart{two} \ch{G}\nodepart{three}\ch{G} \nodepart{four} \ch{T}\nodepart{five} \ch{A}\nodepart{six} \ch{A}\nodepart{seven} \ch{C} \nodepart{eight} \ch{C}\nodepart{nine} \ch{A} \nodepart{ten} \ch{A} \nodepart{eleven} \ch{C} \nodepart{twelve} \ch{A}};
    
    \draw [decorate,decoration={brace,amplitude=5pt,raise=0.5ex}]
  (s.north west) -- (s.four split north) node[midway,yshift=1.5em]{$y$};

\draw [decorate,decoration={brace,amplitude=5pt,raise=0.5ex}]
  (s.seven split north) -- (s.north east) node[midway,yshift=1.5em]{$z$};

  \draw [decorate,decoration={brace,amplitude=5pt,raise=0.5ex}]
  (s.four split north) -- (s.seven split north) node[midway,yshift=1.5em]{$w$};

    \end{tikzpicture}
    \caption{\label{fig:antemer_example:a}}
\end{subfigure}\hfill
\begin{subfigure}{0.45\textwidth}
\centering
    \begin{tikzpicture}
    \tikzstyle{kmer}=[rectangle,rectangle split, rectangle split horizontal,rectangle split parts=12,ultra thick,draw=lightgray,rounded corners,rectangle split part fill={white,white,white,white,lgreen!30,lgreen!30,lgreen!30,white}]
    
    \node[kmer] (s) at (0,0) {\nodepart{one} \ch{C} \nodepart{two} \ch{G}\nodepart{three}\ch{G} \nodepart{four} \ch{A}\nodepart{five} \ch{A}\nodepart{six} \ch{A}\nodepart{seven} \ch{C} \nodepart{eight} \ch{C}\nodepart{nine} \ch{A} \nodepart{ten} \ch{A} \nodepart{eleven} \ch{C} \nodepart{twelve} \ch{A}};
    
    \draw [decorate,decoration={brace,amplitude=5pt,raise=0.5ex}]
  (s.north west) -- (s.four split north) node[midway,yshift=1.5em]{$y'$};

\draw [decorate,decoration={brace,amplitude=5pt,raise=0.5ex}]
  (s.seven split north) -- (s.north east) node[midway,yshift=1.5em]{$z$};

  \draw [decorate,decoration={brace,amplitude=5pt,raise=0.5ex}]
  (s.four split north) -- (s.seven split north) node[midway,yshift=1.5em]{$w$};

    \end{tikzpicture}
    \caption{\label{fig:antemer_example:b}}
\end{subfigure}
    \caption{Examples of antemers and postmers, with $w=\ch{AAC}$ and $m=3$. In Figure~\ref{fig:antemer_example:a}, $y$ is an $4$-antemer because \ch{CGG}, \ch{GGT}, \ch{GTA} and \ch{TAA} are strictly greater than \ch{AAC}. $wz$ is a $8$-postmer --- cf. Remark~\ref{rmk:postmer_def} --- because \ch{AAC}, \ch{ACC}, \ch{CCA}, \ch{CAA}, \ch{AAC} and \ch{ACA} are greater or equal than \ch{AAC}. In Figure~\ref{fig:antemer_example:b}, $y'$ is not an  $4$-antemer because $\ch{AAA} \leq \ch{AAC}$. Thus, $w$ is not a minimizer in Figure~\ref{fig:antemer_example:b}.}
    \label{fig:antemer_example}
\end{figure}

As a result of the previous discussion and considering Remark~\ref{rmk:postmer_def}, we have the following result.
\begin{theorem}
\begin{equation}\label{eq:pi_calcul}
  \pi_k(w) = \sum_{\alpha+\beta = k-m} A(\alpha) \cdot P_{m}(\beta+m).
\end{equation}
\end{theorem}
\begin{proof}
It is easy to see that it is not possible to generate the same $k$-mer via two distinct pairs $(\alpha,\beta)$ and $(\alpha',\beta')$, by virtue of the fact that for $x=ywz$ with $|y|=\alpha$ and $|z|=\beta$, $f_{\min}(x) = \alpha+1$.
\end{proof}

{
\color{blue}
\begin{remark}
Words of the form $yw$ (corresponding to antemers) or $wz$ (corresponding to postmers we are interested in practice, in light of Remark~\ref{rmk:postmer_def}) are designed in the literature by the term \emph{closed syncmers} \cite{edgar2021syncmers}.
\end{remark}
}

It has already been established that $A$ and $P$ are independent of each other; consequently in the following we will detail their calculation in a dedicated section for each. The strategy will be the same for both: find recurrence relations between the values of $A$ and $A_0,\dots,A_{m-1}$ (resp. $P$ and $P_0,\dots, P_m$); in particular our approach is based on the study of words sharing a common prefix of size $i$ with $w$, and the different ways of choosing the $(i+1)$-th letter of this word. Before proceeding with this analysis, we introduce our main tool in the next section.

\subsection{Autocorrelation matrix}

Classically, the autocorrelation vector of a word $w=a_1\cdots a_m$ is defined as a binary vector $(b_1,\dots,b_m)$ where $b_i =1 \iff a_{i}\cdots a_m = a_1\cdots a_{m-i+1}$ --- \textit{i.e.} $b_i=1$ if and only if the suffix of size $m-i+1$ of $w$ is equal to its prefix of size $m-i+1$ \cite{rivals2003combinatorics}. In a similar way, we are interested in this section in how the various substrings of $w$ compare with its prefixes --- the reason of which will emerge later. More precisely, for any $1\leq j\leq i\leq m$, we are interested in the lexicographic order between the substring $a_j\cdots a_i$ and the prefix $a_1\cdots a_{i-j+1}$. Unlike for the autocorrelation vector, we do not just want to know whether these two substrings are equal; if they are different, we need to know whichever is greater than the other. Since our notion generalizes the autocorrelation vector, we propose to use the name autocorrelation matrix to designate these comparisons.

\begin{definition}\label{def:autocorrelation_matrix}
The autocorrelation matrix of $w$ is defined as the lower triangular matrix denoted by $\mathbf{R}$, with $\mathbf{R}_{i,j} \in \lbrace <,=,>\rbrace$ for $1\leq j \leq i \leq m$, so that $a_j \cdots a_i \mathrel{\mathbf{R}_{i,j}} a_1\cdots a_{i-j+1}$.
\end{definition}

The practical computation of this matrix is covered in Appendix~\ref{app:computation_autocorrelation_matrix}. An illustration is provided in Figure~\ref{fig:autocorrelation_matrix_example}.

\begin{figure}[H]
    \centering

\def\scale{0.75}
\begin{tikzpicture}[scale=\scale]

\tikzstyle{mmer}=[rectangle,rectangle split, rectangle split horizontal,rectangle split parts=6,ultra thick,draw=lightgray,rounded corners,scale=\scale]

\tikzstyle{fleche}=[->,>=latex,thick]

\def\x{3.25}
\def\y{1.5}

\node[anchor=east] at (0.5*\x,0) {$\overset{\downarrow}{a_j}\cdots\underset{\uparrow}{a_i}$};

\node (i1) at (\x/3,-\y) {$1$};
\node at (\x/3,-2*\y) {$2$};
\node at (\x/3,-3*\y) {$3$};
\node  at (\x/3,-4*\y) {$4$};
\node at (\x/3,-5*\y) {$5$};
\node (i6) at (\x/3,-6*\y) {$6$};
\draw [decorate,decoration={brace,amplitude=5pt,raise=1ex,mirror}]
  (i1.north) -- (i6.south) node[midway,xshift=-2em]{$i$};


\node[anchor=east] at (0.5*\x,-8*\y) {$a_1\cdots a_{i-j+1}$};

\node[mmer,rectangle split part fill={lgreen!30,white}] (t1) at (6*\x,-8*\y) {\nodepart{one} \ch{A}};

\node[mmer,rectangle split part fill={lgreen!30,lblue!30,white}] (t2) at (5*\x,-8*\y) {\nodepart{one} \ch{A} \nodepart{two} \ch{C}};

\node[mmer,rectangle split part fill={lgreen!30,lblue!30,lred!30,white}] (t3) at (4*\x,-8*\y) {\nodepart{one} \ch{A}\nodepart{two} \ch{C} \nodepart{three}\ch{A}};

\node[mmer,rectangle split part fill={lgreen!30,lblue!30,lred!30,lorange!30,white}] (t4) at (3*\x,-8*\y) {\nodepart{one} \ch{A}\nodepart{two} \ch{C} \nodepart{three}\ch{A} \nodepart{four}\ch{C}};

\node[mmer,rectangle split part fill={lgreen!30,lblue!30,lred!30,lorange!30,lolive!50,white}] (t5) at (2*\x,-8*\y) {\nodepart{one} \ch{A}\nodepart{two} \ch{C} \nodepart{three}\ch{A} \nodepart{four}\ch{C} \nodepart{five}\ch{A}};

\node[mmer,rectangle split part fill={lgreen!30,lblue!30,lred!30,lorange!30,lolive!50,lpurple!50}] (t6) at (\x,-8*\y) {\nodepart{one} \ch{A}\nodepart{two} \ch{C} \nodepart{three}\ch{A} \nodepart{four}\ch{C} \nodepart{five}\ch{A} \nodepart{six}\ch{A}};


\draw[ultra thick,red!25,->,>=latex] (\x,-\y)--(6*\x,-6*\y) to node[draw,fill=white] {\color{black}$\overset{?}{\leq}$} (t1);
\draw[ultra thick,red!25,->,>=latex] (\x,-2*\y)--(5*\x,-6*\y)to node[draw,fill=white] {\color{black}$\overset{?}{\leq}$}(t2);
\draw[ultra thick,red!25,->,>=latex] (\x,-3*\y)--(4*\x,-6*\y)to node[draw,fill=white] {\color{black}$\overset{?}{\leq}$}(t3);
\draw[ultra thick,red!25,->,>=latex] (\x,-4*\y)--(3*\x,-6*\y)to node[draw,fill=white] {\color{black}$\overset{?}{\leq}$}(t4);
\draw[ultra thick,red!25,->,>=latex] (\x,-5*\y)--(2*\x,-6*\y)to node[draw,fill=white] {\color{black}$\overset{?}{\leq}$}(t5);
\draw[ultra thick,red!25,->,>=latex] (\x,-6*\y)--(\x,-6*\y)to node[draw,fill=white] {\color{black}$\overset{?}{\leq}$}(t6);


\node[mmer,rectangle split part fill={lgreen!30,white}] (s11) at (\x,-\y) {\nodepart{one} \ch{A}};
\draw[fleche] ([shift={(0,-\y/4)}]s11.one south)--(s11.one south);
\draw[fleche] ([shift={(0,\y/4)}]s11.one north)--(s11.one north);

\node[mmer,rectangle split part fill={lgreen!30,lblue!30,white}] (s21) at (\x,-2*\y) {\nodepart{one} \ch{A} \nodepart{two} \ch{C}};
\draw[fleche] ([shift={(0,-\y/4)}]s21.two south)--(s21.two south);
\draw[fleche] ([shift={(0,\y/4)}]s21.one north)--(s21.one north);

\node[mmer,rectangle split part fill={white,lblue!30,white}] (s22) at (2*\x,-2*\y) {\nodepart{two} \ch{C}};
\draw[fleche] ([shift={(0,\y/4)}]s22.two north)--(s22.two north);
\draw[fleche] ([shift={(0,-\y/4)}]s22.two south)--(s22.two south);


\node[mmer,rectangle split part fill={lgreen!30,lblue!30,lred!30,white}] (s31) at (\x,-3*\y) {\nodepart{one} \ch{A}\nodepart{two} \ch{C} \nodepart{three}\ch{A}};
\draw[fleche] ([shift={(0,-\y/4)}]s31.three south)--(s31.three south);
\draw[fleche] ([shift={(0,\y/4)}]s31.one north)--(s31.one north);

\node[mmer,rectangle split part fill={white,lblue!30,lred!30,white}] (s32) at (2*\x,-3*\y) {\nodepart{two} \ch{C} \nodepart{three}\ch{A}};
\draw[fleche] ([shift={(0,-\y/4)}]s32.three south)--(s32.three south);
\draw[fleche] ([shift={(0,\y/4)}]s32.two north)--(s32.two north);

\node[mmer,rectangle split part fill={white,white,lred!30,white}] (s33) at (3*\x,-3*\y) { \nodepart{three}\ch{A}};
\draw[fleche] ([shift={(0,-\y/4)}]s33.three south)--(s33.three south);
\draw[fleche] ([shift={(0,\y/4)}]s33.three north)--(s33.three north);


\node[mmer,rectangle split part fill={lgreen!30,lblue!30,lred!30,lorange!30,white}] (s41) at (\x,-4*\y) {\nodepart{one} \ch{A}\nodepart{two} \ch{C} \nodepart{three}\ch{A} \nodepart{four}\ch{C}};
\draw[fleche] ([shift={(0,-\y/4)}]s41.four south)--(s41.four south);
\draw[fleche] ([shift={(0,\y/4)}]s41.one north)--(s41.one north);

\node[mmer,rectangle split part fill={white,lblue!30,lred!30,lorange!30,white}] (s42) at (2*\x,-4*\y) {\nodepart{two} \ch{C} \nodepart{three}\ch{A}\nodepart{four}\ch{C}};
\draw[fleche] ([shift={(0,-\y/4)}]s42.four south)--(s42.four south);
\draw[fleche] ([shift={(0,\y/4)}]s42.two north)--(s42.two north);

\node[mmer,rectangle split part fill={white,white,lred!30,lorange!30,white}] (s43) at (3*\x,-4*\y) { \nodepart{three}\ch{A}\nodepart{four}\ch{C}};
\draw[fleche] ([shift={(0,-\y/4)}]s43.four south)--(s43.four south);
\draw[fleche] ([shift={(0,\y/4)}]s43.three north)--(s43.three north);

\node[mmer,rectangle split part fill={white,white,white,lorange!30,white}] (s44) at (4*\x,-4*\y) {\nodepart{four}\ch{C}};
\draw[fleche] ([shift={(0,\y/4)}]s44.four north)--(s44.four north);
\draw[fleche] ([shift={(0,-\y/4)}]s44.four south)--(s44.four south);


\node[mmer,rectangle split part fill={lgreen!30,lblue!30,lred!30,lorange!30,lolive!50,white}] (s51) at (\x,-5*\y) {\nodepart{one} \ch{A}\nodepart{two} \ch{C} \nodepart{three}\ch{A} \nodepart{four}\ch{C} \nodepart{five}\ch{A}};
\draw[fleche] ([shift={(0,-\y/4)}]s51.five south)--(s51.five south);
\draw[fleche] ([shift={(0,\y/4)}]s51.one north)--(s51.one north);

\node[mmer,rectangle split part fill={white,lblue!30,lred!30,lorange!30,lolive!50,white}] (s52) at (2*\x,-5*\y) {\nodepart{two} \ch{C} \nodepart{three}\ch{A}\nodepart{four}\ch{C}\nodepart{five}\ch{A}};
\draw[fleche] ([shift={(0,-\y/4)}]s52.five south)--(s52.five south);
\draw[fleche] ([shift={(0,\y/4)}]s52.two north)--(s52.two north);

\node[mmer,rectangle split part fill={white,white,lred!30,lorange!30,lolive!50,white}] (s53) at (3*\x,-5*\y) { \nodepart{three}\ch{A}\nodepart{four}\ch{C}\nodepart{five}\ch{A}};
\draw[fleche] ([shift={(0,-\y/4)}]s53.five south)--(s53.five south);
\draw[fleche] ([shift={(0,\y/4)}]s53.three north)--(s53.three north);

\node[mmer,rectangle split part fill={white,white,white,lorange!30,lolive!50,white}] (s54) at (4*\x,-5*\y) {\nodepart{four}\ch{C}\nodepart{five}\ch{A}};
\draw[fleche] ([shift={(0,-\y/4)}]s54.five south)--(s54.five south);
\draw[fleche] ([shift={(0,\y/4)}]s54.four north)--(s54.four north);

\node[mmer,rectangle split part fill={white,white,white,white,lolive!50,white}] (s55) at (5*\x,-5*\y) {\nodepart{five}\ch{A}};
\draw[fleche] ([shift={(0,\y/4)}]s55.five north)--(s55.five north);
\draw[fleche] ([shift={(0,-\y/4)}]s55.five south)--(s55.five south);


\node[mmer,rectangle split part fill={lgreen!30,lblue!30,lred!30,lorange!30,lolive!50,lpurple!50}] (s61) at (\x,-6*\y) {\nodepart{one} \ch{A}\nodepart{two} \ch{C} \nodepart{three}\ch{A} \nodepart{four}\ch{C} \nodepart{five}\ch{A} \nodepart{six}\ch{A}};
\draw[fleche] ([shift={(0,-\y/4)}]s61.six south)--(s61.six south);
\draw[fleche] ([shift={(0,\y/4)}]s61.one north)--(s61.one north);

\node[mmer,rectangle split part fill={white,lblue!30,lred!30,lorange!30,lolive!50,lpurple!50}] (s62) at (2*\x,-6*\y) {\nodepart{two} \ch{C} \nodepart{three}\ch{A}\nodepart{four}\ch{C}\nodepart{five}\ch{A}\nodepart{six}\ch{A}};
\draw[fleche] ([shift={(0,-\y/4)}]s62.six south)--(s62.six south);
\draw[fleche] ([shift={(0,\y/4)}]s62.two north)--(s62.two north);

\node[mmer,rectangle split part fill={white,white,lred!30,lorange!30,lolive!50,lpurple!50}] (s63) at (3*\x,-6*\y) { \nodepart{three}\ch{A}\nodepart{four}\ch{C}\nodepart{five}\ch{A}\nodepart{six}\ch{A}};
\draw[fleche] ([shift={(0,-\y/4)}]s63.six south)--(s63.six south);
\draw[fleche] ([shift={(0,\y/4)}]s63.three north)--(s63.three north);

\node[mmer,rectangle split part fill={white,white,white,lorange!30,lolive!50,lpurple!50}] (s64) at (4*\x,-6*\y) {\nodepart{four}\ch{C}\nodepart{five}\ch{A}\nodepart{six}\ch{A}};
\draw[fleche] ([shift={(0,-\y/4)}]s64.six south)--(s64.six south);
\draw[fleche] ([shift={(0,\y/4)}]s64.four north)--(s64.four north);

\node[mmer,rectangle split part fill={white,white,white,white,lolive!50,lpurple!50}] (s65) at (5*\x,-6*\y) {\nodepart{five}\ch{A}\nodepart{six}\ch{A}};
\draw[fleche] ([shift={(0,-\y/4)}]s65.six south)--(s65.six south);
\draw[fleche] ([shift={(0,\y/4)}]s65.five north)--(s65.five north);

\node[mmer,rectangle split part fill={white,white,white,white,white,lpurple!50}] (s66) at (6*\x,-6*\y) {\nodepart{six}\ch{A}};
\draw[fleche] ([shift={(0,-\y/4)}]s66.six south)--(s66.six south);
\draw[fleche] ([shift={(0,\y/4)}]s66.six north)--(s66.six north);

\node (j1) at (\x,0) {$1$};
\node at (2*\x,0) {$2$};
\node at (3*\x,0) {$3$};
\node at (4*\x,0) {$4$};
\node at (5*\x,0) {$5$};
\node (j6) at (6*\x,0) {$6$};
\draw [decorate,decoration={brace,amplitude=5pt,raise=1ex}]
  (s11.west|-,0) -- (s66.east|-,0) node[midway,yshift=2em]{$j$};

\end{tikzpicture}

    \caption{Each prefix of the minimizer (first column, $j=1$) must be compared to each substring of the minimizer having the same length as the prefix (substrings that are on the same diagonal as the prefix) for computing the autocorrelation matrix. The figure shows an example for $m=6$, with \ch{ACACAA} as an example minimizer. For instance, the strings $\colorbox{lgreen!30}{\ch{A}}\colorbox{lblue!30}{\ch{C}}\colorbox{lred!30}{\ch{A}}\colorbox{lorange!30}{\ch{C}}\colorbox{lolive!50}{\ch{A}}$ ($(i,j)=(5,1)$) and $\colorbox{lblue!30}{\ch{C}}\colorbox{lred!30}{\ch{A}}\colorbox{lorange!30}{\ch{C}}\colorbox{lolive!50}{\ch{A}}\colorbox{lpurple!50}{\ch{A}}$ ($(i,j) = (6,2)$) are both to be compared with the string $\colorbox{lgreen!30}{\ch{A}}\colorbox{lblue!30}{\ch{C}}\colorbox{lred!30}{\ch{A}}\colorbox{lorange!30}{\ch{C}}\colorbox{lolive!50}{\ch{A}}$ ($i-j+1=5$) --- leading, respectively, to $\mathbf{R}_{5,1}=``="$ and $\mathbf{R}_{6,2}=``>"$. The final autocorrelation matrix for this example can be found in upcoming Example~\ref{ex:autocorrelation_matrix}.}
    \label{fig:autocorrelation_matrix_example}
\end{figure}
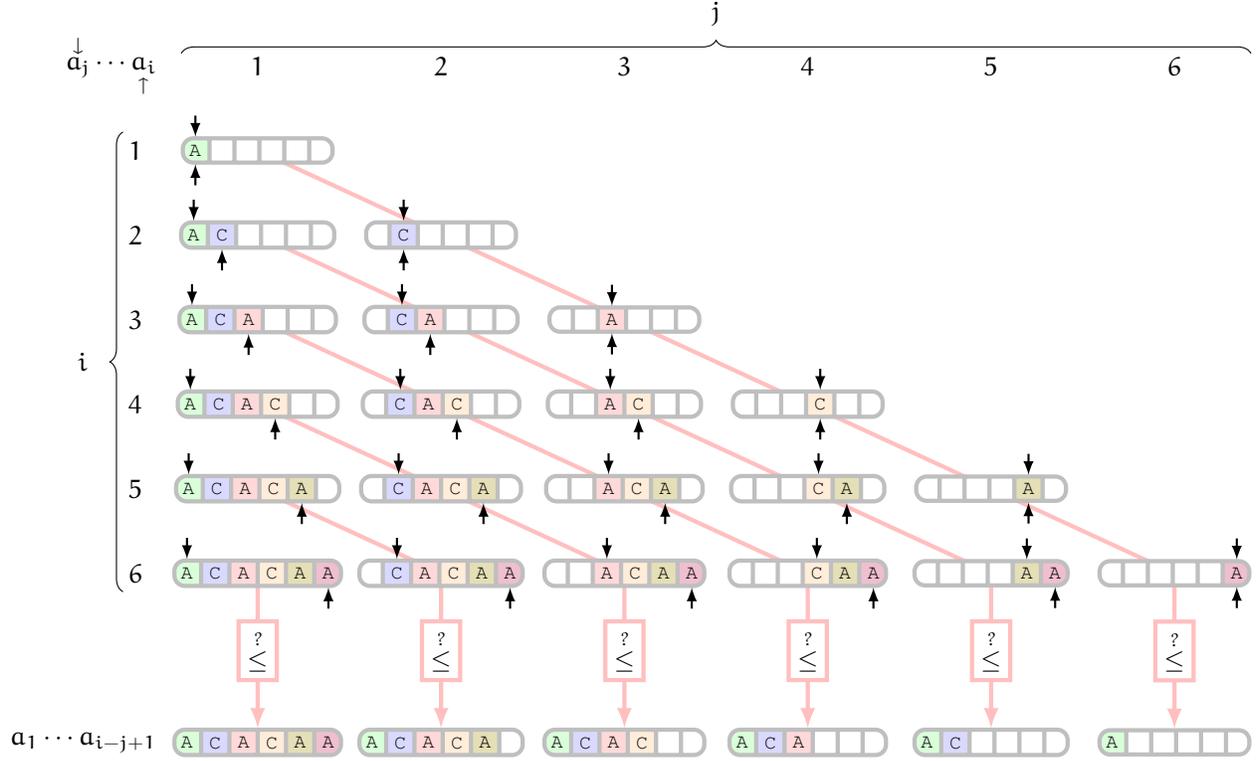

 To simplify notations, in the sequel we shall use the following binary variables.

\begin{definition}\label{def:binary_variables}
For $\star\in \lbrace <,=,>\rbrace$ and $1\leq j\leq i \leq m$, we define the binary variable $\textbf{R}_{i,j}^\star$, so that $\textbf{R}_{i,j}^\star$ is true if and only if $\textbf{R}_{i,j} = \star$.
\end{definition}

In other words, the sentence ``if $\mathbf{R}_{i,j} = \star$" will simply be written ``if $\mathbf{R}_{i,j}^\star$". Whenever useful, $\mathbf{R}_{i,j}^\star$ will also be treated as a number in $\lbrace 0,1 \rbrace$. By definition, the aforementioned classical autocorrelation vector corresponds to the binary vector $(\mathbf{R}^=_{1,m},\dots,\mathbf{R}^=_{m,m})$.

\begin{example}\label{ex:autocorrelation_matrix}
In the sequel of this paper, we shall follow two minimizers examples, $w_1 = \ch{\textnumbering{A,C,A,C,A,A}}$ and $w_2=\ch{\textnumbering{A,C,A,C,A,C}}$, whose autocorrelation matrices are given below (resp. left and right). 
\begin{center}
\begin{minipage}{0.45\textwidth}
  \centering
$\begin{blockarray}{ccccccc}
w_1&1 &2 & 3 & 4& 5& 6\\
\begin{block}{c(cccccc)}
  1 & = & . & . & . & . &.  \\
  2 & = & > & . & .  & . &. \\
  3 & = & > & = & . & . &.  \\
  4 & = & > & = & > & . &.  \\
  5 & = & > & = & > & =&.  \\
  6 & = & > &<& > & <& = \\
\end{block}
\end{blockarray}$
 \end{minipage}~
\begin{minipage}{0.45\textwidth}
 \centering
$\begin{blockarray}{ccccccc}
w_2&1 &2 & 3 & 4& 5& 6\\
\begin{block}{c(cccccc)}
  1 & = & . & . & . & . &.  \\
  2 & = & > & . & .  & . &. \\
  3 & = & > & = & . & . &.  \\
  4 & =& > & = & > & . &.  \\
  5 & = & > & = & > & = &.  \\
  6 & = & > &= &>  &= &>  \\
\end{block}
\end{blockarray}$
\end{minipage}
\end{center}
\end{example}

\subsection{Deductions from $\mathbf{R}_{i,j}^<$}\label{subsec:deduction_from_autocorrelation_matrix}

The autocorrelation matrix provides information about the constraints weighing on antemers and postmers. The $\mathbf{R}_{i,j}^>$ and $\mathbf{R}_{i,j}^=$ cases will be explored later; in this section we focus on $\mathbf{R}_{i,j}^<$. Consider, for instance, one minimizer coming from Example~\ref{ex:autocorrelation_matrix}, $w_1=\ch{ACACAA}$. For $(i,j)$ being either $(6,3)$ or $(6,5)$, we have $\mathbf{R}_{i,j}^<$. Imagine that $w_1$ is placed at the beginning of some $k$-mer, so that the situation is as follows: 
\begin{center}
\begin{tikzpicture}
\tikzstyle{mmer}=[rectangle,rectangle split, rectangle split horizontal,rectangle split parts=13,ultra thick,draw=lightgray,rounded corners,fill=white]

\node[mmer] (s) at (0,0) {\nodepart{one}\ch{A}\nodepart{two}\ch{C}\nodepart{three}\ch{A}\nodepart{four}\ch{C}\nodepart{five}\ch{A}\nodepart{six}\ch{A} \nodepart{twelve} $\cdots$};

\draw [decorate,decoration={brace,amplitude=5pt,raise=0.5ex}]
  (s.two split north) -- (s.eight split north) node[midway,yshift=1.5em]{$< w_1$};

\draw [decorate,decoration={brace,amplitude=5pt,raise=0.5ex,mirror}]
  (s.four split south) -- (s.ten split south) node[midway,yshift=-1.5em]{$< w_1$};

\end{tikzpicture}
\end{center}
We can immediately see that if the $k$-mer is large enough, no matter how we fill in the missing letters, we will obtain an $m$-mer $<w$, and so $w$ will not be the minimizer anymore of the $k$-mer. This simple example motivates the following key result.

\begin{proposition}\label{prop:key_result_matrix}
If there exist $2\leq j\leq i\leq m$ so that $\mathbf{R}_{i,j}^<$ then
\begin{itemize}
    \item $A_i(\alpha)=0$ for $\alpha \geq i$
    \item $P_i(\beta) =0$ for $\beta \geq m+j-1$
\end{itemize}
\end{proposition}
\begin{proof}
Let $i,j$ so that $\mathbf{R}_{i,j}^<$. Let $\gamma> i$ and consider the $\gamma$-mer $x$ sharing with $w$ a prefix of size $i$, i.e. $x = a_1\cdots a_i b_{i+1}\cdots b_\gamma$. We can rewrite $x$ as
$$x =\underbrace{a_1\cdots \overbrace{a_j\cdots a_i}^{i-j+1}}_{i} \underbrace{b_{i+1}\cdots b_\gamma}_{\gamma-i}.$$

Suppose $m$-mers of $x$ to be either all $>w$ or all $\geq w$. $\mathbf{R}_{i,j}^<$ translates into $a_j\cdots a_i < a_1\cdots a_{i-j+1}$.  Among the $\gamma-i$ remaining letters of $x$, if there are enough to complete $a_j\cdots a_i$ into a $m$-mer, then this $m$-mer will be $<w$, regardless of the choice of remaining letters --- thus a contradiction. Therefore we must have $(\gamma -i) + (i-j+1) < m$, i.e. $\gamma < m+j-1$.

We address the case of $\beta$-postmers by posing $\gamma=\beta$. We indeed have $P_i(\beta)=0$ as soon as $\beta\geq m+j-1$. Concerning $\alpha$-antemers, given that we add $w$ at the end when we consider the condition of $m$-mers $>w$, we have $\gamma =\alpha+m$. Remember that $A_i(\alpha)=0$ if $\alpha <i$. With $\alpha \geq i \geq j-1$, we do have $\alpha +m\geq m+j-1$, hence the result.
\end{proof}

Applying Proposition~\ref{prop:key_result_matrix} to the special case of $(\beta+m)$-postmers sharing a prefix of size $m$ with $w$, we obtain the following result.

\begin{corollary}\label{corr:postmer_beta_max}
If there exists $2\leq j\leq m$ so that $\mathbf{R}_{m,j}^<$, then $P_{m}(\beta+m)=0$ for $\beta \geq j-1$.
\end{corollary}

We derive from Corollary~\ref{corr:postmer_beta_max} the maximum value that $\beta$ can take without $P_m(\beta+m)$ being zero.

\begin{definition}\label{def:postmer_beta_max}
$\beta_{\max}(w)=\min\left(k-m,\begin{cases} \min \lbrace 2\leq j \leq m : \mathbf{R}_{m,j}^< \rbrace -2 & \text{if this set is non-empty;}\\
\infty & \text{otherwise.}
\end{cases}\right)$
\end{definition}

\begin{example}\label{ex:beta_max_postmer}
Considering the minimizers $w_1=\ch{ACACAA}$ and $w_2=\ch{ACACAC}$ from Example~\ref{ex:autocorrelation_matrix}, a symbol $<$ can only be found for $w_1$, with $(i,j)$ equal to either $(6,3)$ or $(6,5)$. We derive from it that $P_m(\beta+m)=0$ whenever $\beta\geq 2$, hence $\beta_{\max}(w_1)= \min(k-m,1)$.
\end{example}

Since $\alpha+\beta=k-m$, it follows that \eqref{eq:pi_calcul} can be rewritten as

\begin{equation}\label{eq:pi_calcul_beta_max}
\pi_k(w) = \sum_{\beta=0}^{\beta_{\max}(w)} A(k-m-\beta) \cdot P_m(\beta+m).
\end{equation}

We immediately derive the following upper bound:
\begin{lemma}\label{lemma:trivial_upper_bound}
$\pi_k(w) \leq (\beta_{\max}(w)+1) \cdot |\Sigma|^{k-m}.$
\end{lemma}
\begin{proof}
Trivially, $A(k-m-\beta)\leq |\Sigma|^{k-m-\beta}$ and $P_m(\beta+m) \leq |\Sigma|^{\beta}$.
\end{proof}

From Proposition~\ref{prop:key_result_matrix}, we can also rewrite \eqref{eq:antemer_def} as

\begin{equation}\label{eq:antemer_final}
A(\alpha) = \sum_{i=0}^{i_{\max}(w)-1} A_i(\alpha)
\end{equation}
where $i_{\max}(w)$ is defined as the following.

\begin{definition}\label{def:prefix_max_size}
$i_{\max}(w) = \begin{cases} \min \lbrace 2\leq i \leq m : \exists\, 2\leq j \leq i, \mathbf{R}_{i,j}^< \rbrace & \text{if this set is non-empty;}\\
m & \text{otherwise.}
\end{cases}$
\end{definition}

\begin{example}
From Example~\ref{ex:beta_max_postmer}, we deduce that $i_{\max}(\ch{ACACAA}) = 6$ since the minimal $i$ for which there exists $j$ so that $\mathbf{R}_{i,j}^<$ is $6$. In this case, it is equal to $m$ and does not further restrict the size of the common prefix that an antemer can share with $\ch{ACACAA}$.
\end{example}

Both $\beta_{\max}(w)$ and $i_{\max}(w)$ can be computed on the fly while computing the autocorrelation matrix $\mathbf{R}$, see Appendix~\ref{app:computation_autocorrelation_matrix}.

\section{Computing $A(\alpha)$}\label{sec:antemers}

Recall that 
\begin{equation}
A(\alpha) = \sum_{i=0}^{i_{\max}(w)-1} A_i(\alpha) \tag{\ref{eq:antemer_final}}
\end{equation}
where $i_{\max}(w)$ has been defined in Definition~\ref{def:prefix_max_size}. In this section, we are interested in finding recurrence relationships between the values of $A(\alpha)$ and $A_i(\alpha)$, where $i < i_{\max}(w)$. It follows that we have, for all $1\leq j\leq i$, $\mathbf{R}_{i,j} \in \lbrace =,>\rbrace$.

\subsection{Edge cases}

\paragraph{Case $i=0$.} In this case, we count the $\alpha$-antemers $x=b_1\cdots b_\alpha$ that share no prefix with $w$. Necessarily, $b_1>a_1$ --- therefore there are $\varphi_>(a_1)$ choices for $b_1$. Then, there are $A(\alpha-1)$ ways of completing $b_2\cdots b_\alpha$.

\begin{proposition}\label{prop:antemer_0}
$A_0(\alpha) = \varphi_>(a_1) \cdot A(\alpha-1)$.
\end{proposition}

\paragraph{Case $i=\alpha$.} Here, the antemer is fixed and is equal to $a_1\cdots a_i$. Therefore, $A_i(i)\in \lbrace 0,1\rbrace$, whether the word $a_1\cdots a_i a_1\cdots a_m$ is acceptable or not. Indeed, we must ensure that all $m$-mers (but the last one) of the word $a_1\cdots a_i a_1\cdots a_m$ are $> w$; that is, for all $1\leq j\leq i$,
$$a_j\cdots a_ia_1 \cdots a_{m-i+j-1} > a_1\cdots a_m$$
This equation trivially holds for any $j$ so that $\textbf{R}_{i,j}^>$. In the case where $\textbf{R}_{i,j}^=$, we have
\begin{alignat*}{3}
  &&\underbrace{a_j\cdots a_i}_{= a_1\cdots a_{i-j+1}}a_1 \cdots a_{m-i+j-1} &> a_1\cdots a_m\\
   &\iff & a_1\cdots a_{m-i+j-1} &> a_{i-j+2}\cdots a_m\\
   &\iff & \textbf{R}_{m,i-j+2}^<\\
\end{alignat*}

We showed the following result.

\begin{proposition}\label{prop:antemer_alpha}
For $1\leq i\leq i_{\max}(w)-1$, we have $A_i(i)\in \lbrace 0,1\rbrace$ with $A_i(i) = \displaystyle\prod_{j=1}^i \big(\textbf{R}_{i,j}^>+\textbf{R}_{i,j}^= \cdot\textbf{R}_{m,i-j+2}^<\big).$
\end{proposition}

\begin{example}\label{ex:edge_cases}
Let us look first at $w_1=\ch{ACACAA}$. Below left, the equations for $i=5$ are illustrated, and one can observe that $A_5(5)=0$ for $w_1$ since there exists $m$-mers that are $<w_1$ in the associated antemer. Below right, one can find all values of $A_i(i)$ for $1\leq i\leq 5$, as well as the associated word $a_1\cdots a_i a_1\cdots a_m$.

\noindent
\begin{minipage}[c]{0.475\textwidth}
\centering
\begin{tikzpicture}

\tikzstyle{kmer}=[rectangle,rectangle split, rectangle split horizontal,rectangle split parts=11,ultra thick,draw=lightgray,rounded corners,rectangle split ignore empty parts,rectangle split part fill={white,white,white,white,white,lgreen!30}]

\tikzstyle{mmer}=[rectangle,rectangle split, rectangle split horizontal,rectangle split parts=6,ultra thick,draw=lightgray,rounded corners]

\node[kmer] (s) at (0,0) {\nodepart{one} \ch{A} \nodepart{two}\ch{C} \nodepart{three}\ch{A}  \nodepart{four}  \ch{C}\nodepart{five} \ch{A}\nodepart{six} \ch{A}\nodepart{seven} \ch{C}\nodepart{eight}\ch{A} \nodepart{nine} \ch{C}\nodepart{ten} \ch{A}\nodepart{eleven}\ch{A}};

 \draw [decorate,decoration={brace,amplitude=5pt,raise=0.5ex}]
  (s.five split north) -- (s.north east) node[midway,yshift=1.5em]{$w_1$};
   \draw [decorate,decoration={brace,amplitude=5pt,raise=0.5ex}]
  (s.north west) -- (s.five split north) node[midway,yshift=1.5em]{$a_1\cdots a_i$};

 \node[mmer,rectangle split part fill={white,white,white,white,white,lgreen!30}] (m1) at ([yshift=-0.8cm]s.three split) {\nodepart{one} \ch{A} \nodepart{two}\ch{C} \nodepart{three}\ch{A}  \nodepart{four}  \ch{C}\nodepart{five} \ch{A}\nodepart{six} \ch{A}};

  \node[mmer,rectangle split part fill={white,white,white,white,lgreen!30}] (m2) at ([yshift=-1.4cm]s.four split) {\nodepart{one} \ch{C} \nodepart{two}\ch{A} \nodepart{three}\ch{C}  \nodepart{four}  \ch{A}\nodepart{five} \ch{A}\nodepart{six} \ch{C}};
  
  \node[mmer,rectangle split part fill={white,white,white,lgreen!30}] (m3) at ([yshift=-2cm]s.five split) {\nodepart{one} \ch{A} \nodepart{two}\ch{C} \nodepart{three}\ch{A}  \nodepart{four}  \ch{A}\nodepart{five} \ch{C}\nodepart{six} \ch{A}};

  \node[mmer,rectangle split part fill={white,white,lgreen!30}] (m4) at ([yshift=-2.6cm]s.six split) {\nodepart{one} \ch{C} \nodepart{two}\ch{A} \nodepart{three}\ch{A}  \nodepart{four}  \ch{C}\nodepart{five} \ch{A}\nodepart{six} \ch{C}};
  
  \node[mmer,rectangle split part fill={white,lgreen!30}] (m5) at ([yshift=-3.2cm]s.seven split) {\nodepart{one} \ch{A} \nodepart{two}\ch{A} \nodepart{three}\ch{C}  \nodepart{four}  \ch{A}\nodepart{five} \ch{C}\nodepart{six} \ch{A}};

\node (t) at ([yshift=-0.8cm]s.east) {$\boldsymbol{=}\; w_1$};
\node at ([yshift=-1.4cm]s.east) {$\boldsymbol{>}\; w_1$};
\node at ([yshift=-2cm]s.east) {$\boldsymbol{<}\; w_1$};
\node at ([yshift=-2.6cm]s.east) {$\boldsymbol{>}\; w_1$};
\node (b) at ([yshift=-3.2cm]s.east) {$\boldsymbol{<}\; w_1$};

\node (m) at ($(s)!0.5!(m1)$) {};

\draw[lightgray,ultra thick] (s.west|-,|-m)--(t.east|-,|-m);

 \draw [decorate,decoration={brace,amplitude=5pt,raise=0.5ex}]
  (t.north east) -- (b.south east) node[midway,xshift=1.5em]{\xmark};

\end{tikzpicture}
\end{minipage}\hfill
\begin{minipage}[c]{0.475\textwidth}
\centering
\begin{tikzpicture}

\tikzstyle{kmer}=[rectangle,rectangle split, rectangle split horizontal,rectangle split parts=11,ultra thick,draw=lightgray,rounded corners,rectangle split ignore empty parts]

\node[kmer,rectangle split part fill={white,white,white,white,white,lgreen!30}] (s) at (0,0) {\nodepart{one} \ch{A} \nodepart{two}\ch{C} \nodepart{three}\ch{A}  \nodepart{four}  \ch{C}\nodepart{five} \ch{A}\nodepart{six} \ch{A}\nodepart{seven} \ch{C}\nodepart{eight}\ch{A} \nodepart{nine} \ch{C}\nodepart{ten} \ch{A}\nodepart{eleven}\ch{A}};
\node (i5) at (-3.5,0) {$5$};
\node (j5) at (3.5,0) {$0$};

\node[kmer,rectangle split parts=10,rectangle split part fill={white,white,white,white,lgreen!30}] at ([yshift=0.3cm]s.six split north) {\nodepart{one} \ch{A} \nodepart{two}\ch{C} \nodepart{three}\ch{A}  \nodepart{four}  \ch{C}\nodepart{five} \ch{A} \nodepart{six} \ch{C}\nodepart{seven} \ch{A}\nodepart{eight}\ch{C} \nodepart{nine} \ch{A}\nodepart{ten} \ch{A}};
\node at ([yshift=0.3cm]i5.north) {$4$};
\node at ([yshift=0.3cm]j5.north) {$1$};

\node[kmer,rectangle split parts=9,rectangle split part fill={white,white,white,lgreen!30}] at ([yshift=0.9cm]s.seven north) {\nodepart{one} \ch{A} \nodepart{two}\ch{C} \nodepart{three}\ch{A}  \nodepart{four}  \ch{A}\nodepart{five} \ch{C} \nodepart{six} \ch{A}\nodepart{seven} \ch{C}\nodepart{eight}\ch{A} \nodepart{nine} \ch{A}};
\node at ([yshift=0.9cm]i5.north) {$3$};
\node at ([yshift=0.9cm]j5.north) {$0$};

\node[kmer,rectangle split parts=8,rectangle split part fill={white,white,lgreen!30}] at ([yshift=1.5cm]s.seven split north) {\nodepart{one} \ch{A} \nodepart{two}\ch{C} \nodepart{three}\ch{A}  \nodepart{four}  \ch{C}\nodepart{five} \ch{A} \nodepart{six} \ch{C}\nodepart{seven} \ch{A}\nodepart{eight}\ch{A}};
\node at ([yshift=1.5cm]i5.north) {$2$};
\node at ([yshift=1.5cm]j5.north) {$1$};

\node[kmer,rectangle split parts=7,rectangle split part fill={white,lgreen!30}] at ([yshift=2.1cm]s.eight north) {\nodepart{one} \ch{A} \nodepart{two}\ch{A} \nodepart{three}\ch{C}  \nodepart{four}  \ch{A}\nodepart{five} \ch{C} \nodepart{six} \ch{A}\nodepart{seven} \ch{A}};
\node at ([yshift=2.1cm]i5.north) {$1$};
\node at ([yshift=2.1cm]j5.north) {$0$};

\node at ([yshift=2.7cm]i5.north) {$i$};
\node at ([yshift=2.7cm]s.north) {$a_1\cdots a_i a_1\cdots a_m$};
\node at ([yshift=2.7cm]j5.north) {$A_i(i)$};

\end{tikzpicture}
\end{minipage}

\medskip
\noindent
For $w_2=\ch{ACACAC}$, we leave it to the interested reader to convince themself that $A_i(i)=0$ for all $1\leq i\leq 5$.
\end{example}

\subsection{General case}\label{ss:general_case}

Let us consider an $\alpha$-antemer $x$ sharing with $w$ a prefix of size exactly $0< i <\min(\alpha,i_{\max}(w))$, that is, $x=a_1\cdots a_i b_{i+1} \cdots b_\alpha$, with $b_{i+1}>a_{i+1}$.\footnote{Otherwise, the common prefix between $w$ and $x$ would be of size at least $i+1$.} We must ensure that all $m$-mers of $xw$ (but the last one) are $>w$. Our strategy is to consider the constraints on the possible letters for $b_{i+1}$, and then to consider the recursive cases arising from these alternatives. We are therefore only interested in the $m$-mers that contain $b_{i+1}$. Those $m$-mers are, with $1\leq j\leq i$ :

\begin{equation}
a_j\cdots a_i b_{i+1}\cdots b_{m-1+j} > a_1\cdots a_m\tag{$i,j$}  \label{eq:i_j}
\end{equation}
and
\begin{equation}
b_{i+1}\cdots b_{i+m} > a_1\cdots a_m. \tag{$i$}\label{eq:i}
\end{equation}

where, if necessary, $b_l = a_{l-\alpha}$ for $l>\alpha$.\footnote{Indeed, as we iterate over the $m$-mers of the word $xw$, and depending on $\alpha$ and $i$, some $m$-mers containing $b_{i+1}$ might overlap over $w$, and therefore some $b_l$'s may be contained in $w$ rather than in $x$.}

Recall that, first, since $i$ is the size of the common prefix between $x$ and $w$, we must also have
\begin{equation}
b_{i+1} > a_{i+1} \tag{prefix}\label{eq:prefix}
\end{equation}

Recall also that $\mathbf{R}_{i,j}\in \lbrace =,>\rbrace$. For any $i,j$ so that $\mathbf{R}_{i,j}^>$, it is clear that \eqref{eq:i_j} trivially holds.  For the remainder of this section, we focus on the case $\mathbf{R}_{i,j}^=$. We have
\begin{alignat}{3}
   & &\underbrace{a_j\cdots a_i}_{=a_1\cdots a_{i-j+1
    }}b_{i+1}\cdots b_{m-1+j} &> a_1\cdots a_m\nonumber\\
   & \iff & b_{i+1}\cdots b_{m-1+j} &> a_{i-j+2}\cdots a_m. \tag{$i,j,=$}\label{eq:i_j_equal}
\end{alignat}

\begin{example}\label{example:equations_antemer_general}
We illustrate the preceding equations with $w_1= \ch{ACACAA}$ and $i=5$. We consider the word $x$ of the form --- simplifying $b_{i+1}$ into $b$:
\begin{center}
\begin{tikzpicture}
\tikzstyle{mmer}=[rectangle,rectangle split, rectangle split horizontal,rectangle split parts=8,ultra thick,draw=lightgray,rounded corners,fill=white,,rectangle split part align=base]

\node[mmer] (s) at (0,0) {\nodepart{one}\ch{A}\nodepart{two}\ch{C}\nodepart{three}\ch{A}\nodepart{four}\ch{C}\nodepart{five}\ch{A}\nodepart{six} $b$ \nodepart{seven} $\cdots$};
\end{tikzpicture}
\end{center}
Equation~\eqref{eq:prefix} translates into $b > \ch{A}$; while equations~\eqref{eq:i_j} and~\eqref{eq:i} lead to the following system:
\vspace{-0.5\baselineskip}
\begin{center}
\begin{tikzpicture}
\tikzstyle{mmer}=[rectangle,rectangle split, rectangle split horizontal,rectangle split parts=6,ultra thick,draw=lightgray,rounded corners,fill=white,rectangle split part align=base]

\node[mmer] (s) at (0,0) {\nodepart{one}\ch{A}\nodepart{two}\ch{C}\nodepart{three}\ch{A}\nodepart{four}\ch{C}\nodepart{five}\ch{A}\nodepart{six} $b$};

\node[mmer,anchor=four split] (s1) at ([yshift=-0.5cm]s.five split south) {\nodepart{one}\ch{C}\nodepart{two}\ch{A}\nodepart{three}\ch{C}\nodepart{four}\ch{A}\nodepart{five}$b$\nodepart{six}};

\node[mmer,anchor=three split] (s2) at ([yshift=-0.5cm]s1.four split south) {\nodepart{one}\ch{A}\nodepart{two}\ch{C}\nodepart{three}\ch{A}\nodepart{four}$b$\nodepart{five}\nodepart{six} };

\node[mmer,anchor=two split] (s3) at ([yshift=-0.5cm]s2.three split south) {\nodepart{one}\ch{C}\nodepart{two}\ch{A} \nodepart{three} $b$};

\node[mmer,anchor=one split] (s4) at ([yshift=-0.5cm]s3.two split south) {\nodepart{one}\ch{A}\nodepart{two} $b$};

\node[mmer,anchor=west] (s5) at ([yshift=-0.5cm]s4.one split south) {\nodepart{one} $b$};

\node (e5) at ([xshift=0.25cm]s5.east) {$\boldsymbol{>}$};
\node at (e5|-,|-s4) {$\boldsymbol{>}$};
\node at (e5|-,|-s3) {$\boldsymbol{>}$};
\node at (e5|-,|-s2) {$\boldsymbol{>}$};
\node at (e5|-,|-s1) {$\boldsymbol{>}$};
\node at (e5|-,|-s) {$\boldsymbol{>}$};

\node[mmer,anchor=west] (w5) at ([xshift=0.25cm]e5) {\nodepart{one}\ch{A}\nodepart{two}\ch{C}\nodepart{three}\ch{A}\nodepart{four}\ch{C}\nodepart{five}\ch{A}\nodepart{six} \ch{A}};

\node[mmer] (w4) at (w5|-,|-s4) {\nodepart{one}\ch{A}\nodepart{two}\ch{C}\nodepart{three}\ch{A}\nodepart{four}\ch{C}\nodepart{five}\ch{A}\nodepart{six} \ch{A}};

\node[mmer] (w3) at (w5|-,|-s3) {\nodepart{one}\ch{A}\nodepart{two}\ch{C}\nodepart{three}\ch{A}\nodepart{four}\ch{C}\nodepart{five}\ch{A}\nodepart{six} \ch{A}};

\node[mmer] (w2) at (w5|-,|-s2) {\nodepart{one}\ch{A}\nodepart{two}\ch{C}\nodepart{three}\ch{A}\nodepart{four}\ch{C}\nodepart{five}\ch{A}\nodepart{six} \ch{A}};

\node[mmer] (w1) at (w5|-,|-s1) {\nodepart{one}\ch{A}\nodepart{two}\ch{C}\nodepart{three}\ch{A}\nodepart{four}\ch{C}\nodepart{five}\ch{A}\nodepart{six} \ch{A}};

\node[mmer] (w) at (w5|-,|-s) {\nodepart{one}\ch{A}\nodepart{two}\ch{C}\nodepart{three}\ch{A}\nodepart{four}\ch{C}\nodepart{five}\ch{A}\nodepart{six} \ch{A}};

\node (rij) at ([xshift=0.5cm,yshift=0.5cm]w.east) {$\mathbf{R}_{i,j}^>$};

\node at ([xshift=0.5cm]w1.east) {\checkmark};
\node at ([xshift=0.5cm]w3.east) {\checkmark};
\node at ([xshift=0.5cm]w2.east) {\xmark};
\node at ([xshift=0.5cm]w4.east) {\xmark};
\node at ([xshift=0.5cm]w.east) {\xmark};

\node (iff) at ([xshift=3cm]w2) {$\iff$};

\node[mmer,anchor=west,rectangle split parts=1] (z1) at ([xshift=0.75cm]iff|-,|-s) {\nodepart{one}$b$};

\node[mmer,anchor=west,rectangle split parts=3] (z2) at ([xshift=0.75cm]iff|-,|-s2) {\nodepart{one}$b$};

\node[mmer,anchor=west,rectangle split parts=5] (z3) at ([xshift=0.75cm]iff|-,|-s4) {\nodepart{one}$b$};

 \draw [decorate,decoration={brace,amplitude=5pt,raise=0.5ex,mirror}]
  (z1.north west) -- (z3.south west);

\node (f1) at ([xshift=0.25cm]z1.east) {$\boldsymbol{>}$};

\node[mmer,anchor=west,rectangle split parts=1] (a1) at ([xshift=0.25cm]f1) {\nodepart{one}\ch{A}};

\node (f2) at ([xshift=0.25cm]z2.east) {$\boldsymbol{>}$};

\node[mmer,anchor=west,rectangle split parts=3] (a2) at ([xshift=0.25cm]f2) {\nodepart{one}\ch{C}\nodepart{two}\ch{A}\nodepart{three}\ch{A}};

\node (f3) at ([xshift=0.25cm]z3.east) {$\boldsymbol{>}$};

\node[mmer,anchor=west,rectangle split parts=5] (a3) at ([xshift=0.25cm]f3) {\nodepart{one}\ch{C}\nodepart{two}\ch{A}\nodepart{three}\ch{C}\nodepart{four}\ch{A}\nodepart{five}\ch{A}};

\node at (f3|-,|-rij) {\eqref{eq:i_j_equal}};

\node (m) at ($(s4)!0.5!(s5)$) {};

\draw[lightgray,ultra thick] (s.west|-,|-m)--(a3.east|-,|-m);
\end{tikzpicture}
\end{center}
This system can be solved manually, leading to $b\geq \ch{C}$. Note that choosing $b=\ch{C}$ leads to the word $x=\ch{ACACAC}\cdots$, which in turn contains prefixes of $w$ (resp. $\ch{ACAC}\cdots$ and $\ch{AC}\cdots$) that must be taken into account recursively.
\end{example}

From \eqref{eq:i_j_equal}, we get that  $b_{i+1}\geq a_{i-j+2}$. Remember that $\mathbf{R}_{i,j}^{=}$, thus $a_j\cdots a_i$ is here a prefix of $w$, namely $a_1\cdots a_{i-j+1}$. Therefore, either $b_{i+1}>a_{i-j+2}$ and the prefix stops here, either $b_{i+1}=a_{i-j+2}$ and the prefix continues and must be taken care of recursively --- as shown in Example~\ref{example:equations_antemer_general}. The continued prefix is of size \emph{at least} $i-j+2$, therefore leading to the following fact:

\begin{lemma}\label{lem:prefix_recursive_relation}
For an $\alpha$-antemer $x$, sharing with $w$ a prefix of size exactly $0<i< \min(\alpha,i_{\max}(w))$, i.e. $x = a_1\cdots a_i b_{i+1}\cdots b_\alpha$, if there exists $j$ so that $\mathbf{R}_{i,j}^=$, then choosing $b_{i+1}=a_{i-j+2}$ leads to considering an $(\alpha-j+1)$-antemer $y = a_j\cdots a_i a_{i-j+2} b_{i+2}\cdots b_\alpha$ sharing with $w$ a prefix of size at least $i-j+2$.
\end{lemma}

The next issue to consider is what if choosing $b_{i+1}=a_{i-j+2}$ leads to \emph{simultaneously} extend several prefixes of $w$, as in Example~\ref{example:equations_antemer_general}? Suppose there exist $j'>j\geq 2$, so that (i) both $\mathbf{R}_{i,j}^=$ and $\mathbf{R}_{i,j'}^=$; and (ii) $a_{i-j+2}=a_{i-j'+2}= a \in \Sigma$. Choosing $b_{i+1}=a$ simultaneously extends two prefixes with $w$: $a_j\cdots a_i a = a_1\cdots a_{i-j+1} a$ and $a_{j'}\cdots a_i a = a_1\cdots a_{i-j'+1} a$.  Since $j'>j$, this corresponds to the following word :

$$\overbrace{a_j\cdots a_{j'-1}\underbrace{a_{j'}\cdots a_i a\cdots}_{\substack{\text{prefix of size } \\ \geq i-j'+2}}}^{\text{prefix of size } \geq i-j+2}\cdots.$$

In the light of Lemma~\ref{lem:prefix_recursive_relation}, it is crucial here to remark that recursively treating an $(\alpha-j+1)$-antemer sharing with $w$ a prefix of size at least $i-j+2$ allows us, at the same time, to treat the case of $(\alpha-j'+1)$-antemers sharing with $w$ a prefix of size at least $i-j'+2$ --- i.e. the case related to $j'$ is completely included in the case related to $j$.

In other respects, consider \eqref{eq:i}. As we know nothing about the size of any common prefix between $b_{i+1}\cdots b_{i+m}$ and $w$, we deduce $b_{i+1}\geq a_1$. Either $b_{i+1}>a_1$, and there is no prefix to consider recursively, or $b_{i+1}=a_1$ and a new prefix is created, leading to consider $(\alpha-i)$-antemers sharing with $w$ a prefix of size at least $1$. To mirror the discussion immediately preceding this very paragraph, suppose there exists $j\geq 2$ so that $\mathbf{R}_{i,j}^=$ and $a_{i-j+2}=a_1$, then choosing $b_{i+1}=a_1$ also simultaneously extends several prefixes of $w$ :

$$\overbrace{a_j\cdots a_{i}\underbrace{a_1\cdots}_{\substack{\text{prefix of} \\ \text{size } \geq 1}}}^{\text{prefix of size } \geq i-j+2}\cdots.$$

The conclusion is the same: recursively handling the longest prefix covers the smaller ones.

The whole discussion leads us to define what we call prefix-letter vectors of $w$. Recall that $[P]$ stands for the Iverson bracket of property $P$.

\begin{definition}\label{def:prefix_letter_vector}
For $1\leq i\leq m$, we define the $i$-th prefix-letter vector of $w$, denoted by $\mathbf{T}_i$, as a vector in $\mathbb{N}^{|\Sigma|}$, so that for all $a\in \Sigma$,
\begin{itemize}
    \item $\mathbf{T}_1(a) = 2 \cdot [a=a_1]$,
    \item for all $2\leq i \leq m$,
$\mathbf{T}_i(a)=\begin{cases} \min \lbrace 2\leq j \leq i : \mathbf{R}_{i,j}^= \wedge (a_{i-j+2}=a)\rbrace & \text{if this set is not empty;}\\
(i+1) \cdot [a=a_1] & \text{otherwise.}
\end{cases}$
\end{itemize}
\end{definition}

Following the discussion above, choosing $b_{i+1}=a$ when $\mathbf{T}_i(a)\neq 0$ leads to the consider a $(\alpha - \mathbf{T}_i(a) +1)$-antemer sharing with $w$ a prefix of size at least $i-\mathbf{T}_i(a)+2$, as per Lemma~\ref{lem:prefix_recursive_relation}. In the case where $\mathbf{T}_{i}(a_1)=i+1$, since $\alpha - (i+1) +1 = \alpha -i$ and $i-(i+1)+2=1$, we retrieve the conclusion we draw previously from \eqref{eq:i}.

\begin{remark}
Note that $\mathbf{T}_i$ is defined for $i\geq i_{\max}(w)$, although this is not applicable to the specific case of this section, which deals with $i<i_{\max}(w)$ values. This is because the notion of prefix-letter vector will also be useful for postmers later on, which do not follow this constraint of $i<i_{\max}(w)$.
\end{remark}

It should be noted that the practical computation of the $\mathbf{T}_i$'s can be performed on the fly while building the autocorrelation matrix $\mathbf{R}$, and is covered in Appendix~\ref{app:computation_autocorrelation_matrix}.

\begin{example}\label{ex:prefix-letter-vectors}
We detail here the computation of the prefix-letter vectors of $w_1=\ch{\textnumbering{A,C,A,C,A,A}}$ and $w_2=\ch{\textnumbering{A,C,A,C,A,C}}$, whose autocorrelation matrices are given in Example~\ref{ex:autocorrelation_matrix}. In both cases, $\mathbf{T}_1(a) = 2\cdot [a = \ch{A}]$. For each minimizer, and each $i$, we provide tuples $(j,a)$ so that $\mathbf{R}_{i,j}^=$ and $a_{i-j+2}=a$:
\begin{center}
  \begin{tabular}{c|ccccc}
    $i$ & $2$&$3$&$4$&$5$&$6$\\
    \hline
    $w_1$ &  $\emptyset$ & $(3,\ch{C})$& $(3,\ch{A})$& $(3,\ch{C})$ ; $(5,\ch{C})$ & $(6,\ch{C})$\\
    $w_2$ &  $\emptyset$ & $(3,\ch{C})$& $(3,\ch{A})$& $(3,\ch{C})$ ; $(5,\ch{C})$ & $(3,\ch{A})$ ; $(6,\ch{A})$
\end{tabular}  
\end{center}
Remember that, for a given $i$ and letter $a$, $\mathbf{T}_i(a)$ is computed by taking the smallest $j$ in the tuples $(j,a)$ when applicable --- otherwise it is $(i+1) \cdot [a=a_1]$. We obtain the following prefix-letter vectors (rows correspond to $i$ values):-
\begin{center}
\begin{minipage}{0.45\textwidth}
  \centering
$\begin{blockarray}{ccccc}
w_1& \ch{A}&\ch{C} & \ch{G} & \ch{T}\\
\begin{block}{c(cccc)}
  1 & 2 & 0 & 0 & 0   \\
  2 & 3 & 0 & 0 & 0   \\
  3 & 4 & 3 & 0 & 0  \\
  4 & 3 & 0 & 0 & 0   \\
  5 & 6 & 3 & 0 & 0   \\
  6 & 7 & 6 &0& 0  \\
\end{block}
\end{blockarray}$
 \end{minipage}~
\begin{minipage}{0.45\textwidth}
 \centering
$\begin{blockarray}{ccccc}
w_2& \ch{A}&\ch{C} & \ch{G} & \ch{T}\\
\begin{block}{c(cccc)}
  1 & 2 & 0 & 0 & 0   \\
  2 & 3 & 0 & 0 & 0   \\
  3 & 4 & 3 & 0 & 0  \\
  4 & 3 & 0 & 0 & 0   \\
  5 & 6 & 3 & 0 & 0   \\
  6 & 3 & 0 &0& 0  \\
\end{block}
\end{blockarray}$
\end{minipage}
\end{center}
\end{example}

The next result sums up the discussion so far.
\begin{proposition}\label{prop:conditions_on_next_letter}
 We have the following conditions on $b_{i+1}$:
\begin{enumerate}[label=(\alph*)]
    \item $b_{i+1}> a_{i+1}$ --- since $x$ shares with $w$ a prefix of size exactly $i$, as per \eqref{eq:prefix};
    \item $b_{i+1}\geq a_1$ --- as per \eqref{eq:i};
    \item $b_{i+1}\geq a_{i-j+2}$ for each $j$ so that $\mathbf{R}_{i,j}^=$ (equivalently so that $\mathbf{T}_i(a_{i-j+2})\neq 0$) --- as per \eqref{eq:i_j_equal};
    \item if $\mathbf{T}_i(b_{i+1})=0$, no further constraint is imposed on the subsequent letters $b_{i+2},\dots$;
    \item on the contrary, if $\mathbf{T}_i(b_{i+1})\neq 0$, this leads to recursively consider $(\alpha - \mathbf{T}_i(b_{i+1})+1)$-antemers sharing with $w$ a prefix of size at least $i-\mathbf{T}_i(b_{i+1})+2$ --- as per Lemma~\ref{lem:prefix_recursive_relation} and Definition~\ref{def:prefix_letter_vector}.
\end{enumerate}  
\end{proposition}

Since conditions $(b)$ and $(c)$ must apply simultaneously (when applicable), they should be reduced to the most stringent; let us define
\begin{equation}\label{eq:a_max}
    a_{\max}(i) = \max\lbrace a \in \Sigma : \mathbf{T}_i(a) \neq 0 \rbrace ;
\end{equation}
then $(b)$ and $(c)$ collapse into a single condition, that is : $b_{i+1}\geq a_{\max}(i)$. Note that $a_{\max}(i)$ is well defined since $a_1$ belongs to the set (coming from $(b)$). Besides, $a_{\max}(i)$ can be computed at the same time as computing the autocorrelation matrix, see Appendix~\ref{app:computation_autocorrelation_matrix}.

In the end, it comes down to choosing one of these two options for $b_{i+1}$: 
\begin{itemize}
    \item $b_{i+1} = a_{\max}(i)$ --- available only if $a_{\max}(i)>a_{i+1}$. In this case, we have $\mathbf{T}_i(b_{i+1})\neq 0$ and we fall upon condition $(e)$;
    \item $b_{i+1}>a_{\max}(i)$ and $b_{i+1}>a_{i+1}$. By definition of $a_{\max}(i)$, we must have $\mathbf{T}_i(b_{i+1})=0$ and we fall upon condition $(d)$. There are $\min(\varphi_>(a_{\max}(i)), \varphi_>(a_{i+1}))$ choices for $b_{i+1}$ with this option.
\end{itemize}

\begin{example}\label{ex:a_max}
Building upon Example~\ref{ex:prefix-letter-vectors}, for $w_1=\ch{ACACAA}$ and $w_2 = \ch{ACACAC}$, denoting for convenience $\star$ instead of $\min(\varphi_>(a_{\max}(i)), \varphi_>(a_{i+1}))$ and $\ast$ instead of $a_{\max(i)}>a_{i+1}$; we have :
\begin{center}
 \begin{minipage}[c]{0.45\textwidth}
\centering
\begin{tabular}{c|c|cccccc}
  \multirow{5}{*}{$w_1$} & $i$ & $1$ & $2$ & $3$ & $4$ & $5$ & $6$\\
  &  $a_{\max}(i)$ & \ch{A} & \ch{A} & \ch{C} & \ch{A} &  \ch{C} & \ch{C}\\
  & $a_{i+1}$ & \ch{C} & \ch{A} & \ch{C}&\ch{A}&\ch{A} & $\varepsilon$\\
  & $\ast$& \xmark & \xmark & \xmark & \xmark &\checkmark & \checkmark \\
  & $\star$& $2$ & $3$ & $2$ & $3$ & $2$ & $2$
\end{tabular}
\end{minipage}~
\begin{minipage}[c]{0.45\textwidth}
\centering
\begin{tabular}{c|c|cccccc}
  \multirow{5}{*}{$w_2$} & $i$ & $1$ & $2$ & $3$ & $4$ & $5$ & $6$\\
  &  $a_{\max}(i)$ & \ch{A} & \ch{A} & \ch{C} & \ch{A} &  \ch{C} & \ch{A}\\
    & $a_{i+1}$ & \ch{C} & \ch{A} & \ch{C}&\ch{A}&\ch{C} & $\varepsilon$\\
      & $\ast$& \xmark & \xmark & \xmark & \xmark &\xmark & \checkmark \\
      & $\star$& $2$ & $3$ & $2$ & $3$ & $2$ & $3$
\end{tabular}
\end{minipage}  
\end{center}
\end{example}

We can now derive our main result for this section.
 
\begin{theorem}\label{th:antemer_general}
For all $0<i<\min(i^\ast,\alpha)$,
\begin{align*}
    A_i(\alpha)&= \min(\varphi_>(a_{\max}(i)), \varphi_>(a_{i+1})) \cdot A(\alpha-(i+1))\\
    &+ [a_{\max}(i)>a_{i+1}] \cdot \sum_{i' = i-\mathbf{T}_i(a_{\max}(i))+2}^{i_{\max}(i) -1} A_{i'}(\alpha-\mathbf{T}_i(a_{\max}(i))+1).
\end{align*}
\end{theorem}

\begin{example}\label{ex:antermer_formulas}
Building upon Examples~\ref{ex:edge_cases},~\ref{ex:prefix-letter-vectors} and~\ref{ex:a_max}; as well as Proposition~\ref{prop:antemer_0}.

\noindent
For $w_1=\ch{ACACAA}$, we have :
\begin{align*}
A_0(\alpha) &= 3 \cdot A(\alpha-1)\\
A_1(\alpha) &= 2 \cdot A(\alpha-2)\\
A_2(\alpha) &= 3 \cdot A(\alpha-3)\\
A_3(\alpha) &= 2 \cdot A(\alpha-4)\\
A_4(\alpha) &= 3 \cdot A(\alpha-5)\\
A_5(\alpha) &= 2 \cdot A(\alpha-6) + A_4(\alpha-2) + A_5(\alpha-2)
\end{align*}
with $A(0)=1$; $A_1(1)=A_3(3)=A_5(5)=0$ and $A_2(2)=A_4(4)=1$.

\medskip
\noindent
For $w_2 = \ch{ACACAC}$, we have :
\begin{align*}
A_0(\alpha) &= 3 \cdot A(\alpha-1)\\
A_1(\alpha) &= 2 \cdot A(\alpha-2)\\
A_2(\alpha) &= 3 \cdot A(\alpha-3)\\
A_3(\alpha) &= 2 \cdot A(\alpha-4)\\
A_4(\alpha) &= 3 \cdot A(\alpha-5)\\
A_5(\alpha) &= 2 \cdot A(\alpha-6)
\end{align*}
with $A(0)=1$; and $A_i(i)=0$ for all $1\leq i\leq 5$.
\end{example}

\subsection{Computation and bounds}\label{subsec:antemer_bounds}

Using Propositions~\ref{prop:antemer_0} and~\ref{prop:antemer_alpha}, as well as Theorem~\ref{th:antemer_general}, and \eqref{eq:antemer_final}, we have.

\begin{proposition}\label{prop:antemer_computation}
$A(\alpha)$ can be computed with dynamic programming, with time complexity $O(\alpha\cdot |w|^2)$ and space complexity $O(\alpha\cdot |w|)$.
\end{proposition}
\begin{proof}
There are at most $\alpha\cdot (i_{\max}(w)+1) = O(\alpha\cdot |w|)$ values to compute, each one requiring at most $O(i_{\max}(w))$ time.
\end{proof}

\begin{example}\label{ex:antemers_values}
Using the formulas of Example~\ref{ex:antermer_formulas}, one can compute the first values of $A(\alpha)$.

\medskip
\noindent
For $w_1=\ch{ACACAA}$, we have:
\begin{center}
\small
\begin{tabular}{cr|ccccccccccc}
\multicolumn{2}{c|}{\multirow{2}{*}{\normalsize$A_i(\alpha)$}} &  \multicolumn{11}{c}{\normalsize$\alpha$} \\
	\multicolumn{2}{c|}{} & $0$ & $1$ & $2$ &$3$ & $4$ &$5$ & $6$&$7$&$8$&$9$&$10$\\
	\hline
\multirow{6}{*}{\normalsize$i $}	&$0$& $\mathbf{1}$& $3$ & $9$&$36$&$135$&$513$&$1\,944$&$7\,371$&$27\,945$&$105\,948$&$401\,679$\\
	&$1$&. &$\mathbf{0}$ & $2$&$6$&$24$&$90$&$342$&$1\,296$&$4\,914$&$18\,630$&$70\,632$\\
 &$2$& .&.& $\mathbf{1}$ &$3$&$9$&$36$&$135$&$513$&$1\,944$&$7\,371$&$27\,945$\\
 &$3$& .&.&.  &$\mathbf{0}$&$2$&$6$&$24$&$90$&$342$&$1\,296$&$4\,914$\\
 &$4$& .&.& .&.&$\mathbf{1}$&$3$&$9$&$36$&$135$&$513$&$1\,944$\\
 &$5$& .&.&  .&.&.&$\mathbf{0}$&$3$&$9$&$36$&$135$&$513$\\
 \hline
 \multicolumn{2}{c|}{\normalsize$A(\alpha)$}& $1$& $3$ & $12$&$45$&$171$&$648$&$2\,457$&$9\,315$&$35\,316$&$133\,893$&$507\,627$
\end{tabular}
\end{center}

\medskip
\noindent
For $w_2=\ch{ACACAC}$, we have:
\begin{center}
\small
\begin{tabular}{cr|ccccccccccc}
\multicolumn{2}{c|}{\multirow{2}{*}{\normalsize$A_i(\alpha)$}} &  \multicolumn{11}{c}{\normalsize$\alpha$} \\
	\multicolumn{2}{c|}{} & $0$ & $1$ & $2$ &$3$ & $4$ &$5$ & $6$&$7$&$8$&$9$&$10$\\
	\hline
\multirow{6}{*}{\normalsize$i $}	&$0$& $\mathbf{1}$& $3$& $9$& $33$& $126$& $477$&$1\,809$& $6\,858$& $25\,992$& $98\,517$& $373\,410$\\
	&$1$& .& $\mathbf{0}$ & $2$ & $6$& $22$& $84$& $318$&$1\,206$&$4\,572$&$17\,328$&$65\,678$\\
 &$2$&.& .& $\mathbf{0}$ & $3$&$9$&$33$&$126$&$477$&$1\,809$&$6\,858$&$25\,992$\\
 &$3$&.& .& .& $\mathbf{0}$ & $2$&$6$&$22$&$84$&$318$&$1\,206$&$4\,572$ \\
 &$4$& .& .& .& .& $\mathbf{0}$ & $3$&$9$&$33$&$126$&$477$&$1\,809$ \\
 &$5$&.& .& .& .& .& $\mathbf{0}$ & $2$&$6$&$22$&$84$&$318$\\
 \hline
 \multicolumn{2}{c|}{\normalsize$A(\alpha)$}& $1$& $3$ & $11$& $42$ & $159$& $603$& $2\,286$& $8\,664$& $32\,839$& $124\,470$& $471\,779$
\end{tabular}
\end{center}
\end{example}

Note that the only recursive formula involving $A_i$'s is that of Theorem~\ref{th:antemer_general}. Remark that

\begin{alignat*}{3}
& 0\, &\leq\sum_{i' = i-\mathbf{T}_i(a_{\max}(i))+2}^{i_{\max}(i) -1} A_{i'}(\alpha-\mathbf{T}_i(a_{\max}(i))+1) &\leq \sum_{i'=1}^{i_{\max}(w)-1} A_{i'}(\alpha-\mathbf{T}_i(a_{\max}(i))+1) \\
&&&= A(\alpha-\mathbf{T}_i(a_{\max}(i))+1) - A_0(\alpha-\mathbf{T}_i(a_{\max}(i))+1)\\
&&&= A(\alpha-\mathbf{T}_i(a_{\max}(i))+1) - \varphi_>(a_1) \cdot A(\alpha-\mathbf{T}_i(a_{\max}(i)));
\end{alignat*}
where the last line is obtained using Proposition~\ref{prop:antemer_0}.  As a result, we can easily calculate upper and lower bounds on $A_i(\alpha)$, by adding $\min(\varphi_>(a_{\max}(i)), \varphi_>(a_{i+1})) \cdot A(\alpha-(i+1))$ to each term of the above inequality. Recall from Equation~\eqref{eq:antemer_final} that $A(\alpha) = \sum_{i=0}^{i_{\max}(w)-1} A_i(\alpha)$; this lead us to define the following two sequences $A^-(\alpha)$ and $A^+(\alpha)$.
\begin{definition}\label{def:antemer_bounds}
Let  $A^-(0) = A^+(0)= 1$ and
\begin{align*}
A^-(\alpha) &= \varphi_>(a_1)\cdot A^-(\alpha-1) + \sum_{i=1}^{i_{\max}(w)-1}  \min(\varphi_>(a_{\max}(i)), \varphi_>(a_{i+1})) \cdot A^-(\alpha-(i+1))\\
A^+(\alpha) &= \varphi_>(a_1)\cdot A^+(\alpha-1) + \sum_{i=1}^{i_{\max}(w)-1} \Big( \min(\varphi_>(a_{\max}(i)), \varphi_>(a_{i+1})) \cdot A^+(\alpha-(i+1))\\&+[a_{\max}(i)>a_{i+1}] \cdot  \big(A^+(\alpha-\mathbf{T}_i(a_{\max}(i))+1) - \varphi_>(a_1) \cdot A^+(\alpha-\mathbf{T}_i(a_{\max}(i)))\big)\Big) \\
\end{align*}
\end{definition}

By construction, we have the following result.

\begin{proposition}\label{prop:antemer_bounds}
$\forall \alpha\geq 0, A^-(\alpha) \leq A(\alpha) \leq A^+(\alpha)$.  Furthermore, both $A^-(\alpha)$ and $A^+(\alpha)$ can be computed with dynamic programming in $O(\alpha)$ space and $O(\alpha\cdot |w|)$ time.
\end{proposition}

One can notice that if $[a_{\max}(i)>a_{i+1}]$ were to be $0$ for all values of $i$, then $A^-(\alpha) = A(\alpha) = A^+(\alpha)$.

\begin{example}
For $w_1=\ch{ACACAA}$, from Example~\ref{ex:antermer_formulas}, we can bound $A_5(\alpha)$ with
$$ 2\cdot A(\alpha-6)\leq A_5(\alpha) \leq 2\cdot A(\alpha-6) + A(\alpha-2) - 3\cdot A(\alpha-3);$$
which leads to the following values of $A^-(\alpha)$ and $A^+(\alpha)$:
\begin{center}
\small
\begin{tabular}{c|ccccccccccc}
$\alpha$ & $0$ & $1$ & $2$ &$3$ & $4$ &$5$ & $6$&$7$&$8$&$9$&$10$\\
	\hline
	\normalsize$A^+(\alpha)$& $1$ & $3$ & $12$ & $45$ & $173$ & $663$ & $2\,543$ & $9\,750$ & $37\,384$ & $143\,337$ & $549\,584$\\
\normalsize$A(\alpha)$&  $1$& $3$ & $12$&$45$&$171$&$648$&$2\,457$&$9\,315$&$35\,316$&$133\,893$&$507\,627$\\
\normalsize$A^-(\alpha)$& $1$ & $3$ & $11$ & $42$ & $159$ & $603$ & $2\,286$ & $8\,664$ & $32\,839$ & $124\,470$ & $471\,779$\\
\end{tabular}
\end{center}
Note that the row of $A^-$ corresponds exactly to the values we have for $A$ with $w_2 = \ch{ACACAC}$. Indeed, by looking at the system in Example~\ref{ex:antermer_formulas}, the system corresponding to $w_2$ is exactly the one used for the lower bound for $w_1$. Speaking of $w_2=\ch{ACACAC}$, notice that $A^-=A=A^+$ for $w_2$ since no $A_i$'s are involved in its equation system.
\end{example}

\section{Computing $P(\beta)$}\label{sec:postmers}

Recall that 
\begin{equation}
P(\beta) = \sum_{i=0}^m P_i(\beta).\tag{\ref{eq:postmer_def}}
\end{equation}
In this section, we are interested in finding recurrence relationships between the values of $P(\beta)$ and $P_i(\beta)$.  With Proposition~\ref{prop:key_result_matrix}, we take the case where $\beta$ is such that $P(\beta)\neq 0$. In other words, for all $1\leq j\leq i$, $\mathbf{R}_{i,j} \in \lbrace =,>\rbrace$.

\begin{example}
In Example~\ref{ex:beta_max_postmer}, we saw that for $w_1=\ch{ACACAA}$, $P_m(\beta+m)=0$ as soon as $\beta \geq 2$; for this reason in this section, we shall consider as an example only the word $w_2=\ch{ACACAC}$ --- for which there are no restrictions on the value of $P_m(\beta+m)$.
\end{example}

\subsection{First values}\label{subsec:postmer_first_values}

We are interested here in the first values of $\beta$, namely when $0<\beta \leq m$. We proceed by disjunction of cases.

\paragraph{Case $0<\beta \leq m-1$.} Here, we build $\beta$-postmers of of length strictly less than $|w|$: the condition that all $m$-mers are $\geq w$ is therefore trivially verified since there are no $m$-mer to consider at all. It follows that $P(\beta) = |\Sigma|^\beta$. 

However, let us take a moment to detail the computation of $P_i(\beta)$.

For $i=0$, the first letter of any $\beta$-postmer must be $\neq a_1$. It follows that $P_0(\beta) = \left(|\Sigma|-1\right) \cdot P(\beta-1)$. For $i=\beta$, there is only one $i$-postmer sharing with $w$ a prefix of size $i$, therefore $P_{\beta}(\beta)=1$.

Now, for $0<i<\beta$. We consider $\beta$-postmers of the form $x=a_1\cdots a_i b_{i+1}\cdots b_\beta$. Since $i$ is the size of the prefix, it follows that $b_{i+1}\neq a_{i+1}$ --- leaving $|\Sigma|-1$ choices for $b_{i+1}$. Let us denote $a$ the chosen letter. If there exists $j\geq 2$ so that $\mathbf{R}_{i,j}^=$, i.e. $a_j\cdots a_i = a_1\cdots a_{i-j+1}$, and $a_{i-j+2}=a$, we reconnect with the discussion in Section~\ref{ss:general_case}, leading us to recursively consider $(\beta-j+1)$-postmers sharing with $w$ a prefix of size at least $i-j+2$. If there exist $j,j'$ satisfying the precedent condition, the conclusion is identical and lead to only consider the recursive case corresponding to $\min(j,j')$: we retrieve $\mathbf{T}_i$ of Definition~\ref{def:prefix_letter_vector}.

The final conditions on $b_{i+1}$ are reminiscent of the ones of Proposition~\ref{prop:conditions_on_next_letter} :

\begin{enumerate}[label=(\alph*')]
    \item $b_{i+1} \neq a_{i+1}$ --- since $x$ share with $w$ a prefix of size exactly $i$;
    \item if $\mathbf{T}_i(b_{i+1})=0$, no further constraint is imposed on the subsequent letters $b_{i+2},\dots$;
    \item on the contrary, if $\mathbf{T}_i(b_{i+1})\neq 0$, this lead to recursively consider $(\beta - \mathbf{T}_i(b_{i+1})+1)$-postmers sharing with $w$ a prefix of size at least $i-\mathbf{T}_i(b_{i+1})+2$ --- as per Lemma~\ref{lem:prefix_recursive_relation} and Definition~\ref{def:prefix_letter_vector}.
\end{enumerate}

Finally, we have the following result.

\begin{proposition}\label{prop:recursion_postmer_first_cases}
For $0 <\beta \leq m-1$, we have
$$P_0(\beta) = (|\Sigma-1)\cdot P(\beta-1);$$
and, for $0<i<\beta$, 
$$P_i(\beta) = |\Sigma_{=0}(i)|\cdot P(\beta-(i+1))+\sum_{a\in \Sigma_{\neq 0}(i)} \sum_{i' = i-\mathbf{T}_i(a)+2}^{\beta-\mathbf{T}_i(a)+1} P_{i'}(\beta-\mathbf{T}_i(a)+1)$$
where $\Sigma_{=0}(i) = \lbrace a\in \Sigma : (a\neq a_{i+1}) \wedge (\mathbf{T}_i(a)=0)\rbrace$ and $\Sigma_{\neq 0}(i) = \lbrace a\in\Sigma : (a\neq a_{i+1})\wedge (\mathbf{T}_i(a) \neq 0)\rbrace$; and finally $P_\beta(\beta)=1$.
\end{proposition}

Note that $\mathbf{T}_i(a)\geq 2$ whenever $\mathbf{T}_i(a)\neq0$, therefore we have $\beta - \mathbf{T}_i(a)+1\leq \beta -1$ --- and the same equation system apply recursively to the $P_i$'s. Note also that $|\Sigma_{=0}(i)|+|\Sigma_{\neq 0}(i)| = |\Sigma|-1$. Both sets can be precomputed --- see Appendix~\ref{app:computation_autocorrelation_matrix}.

\begin{example}
With $w=\ch{ACACAC}$, we have :
\begin{center}
 \begin{tabular}{c|c|cccc|cc}
    $i$ & $a_{i+1}$ & $\mathbf{T}_i(\ch{A})$ & $\mathbf{T}_i(\ch{C})$ & $\mathbf{T}_i(\ch{G})$ & $\mathbf{T}_i(\ch{T})$& $\Sigma_{=0}(i)$ & $\Sigma_{\neq 0}(i)$\\
    \hline
    $1$ & \ch{C}& $2$ &$0$&$0$&$0$&$\lbrace \ch{G},\ch{T}\rbrace$&  $\lbrace \ch{A}\rbrace$\\
    $2$ & \ch{A}& $3$&$0$&$0$&$0$&$\lbrace \ch{C},\ch{G},\ch{T}\rbrace$& $\emptyset$\\
    $3$ & \ch{C}& $4$&$3$&$0$&$0$&$\lbrace \ch{G},\ch{T}\rbrace$& $\lbrace \ch{A}\rbrace$\\
    $4$ & \ch{A}& $3$&$0$&$0$&$0$&$\lbrace \ch{C},\ch{G},\ch{T}\rbrace$& $\emptyset$\\
\end{tabular}   
\end{center}
This leads to the following system of equations, for $1\leq \beta\leq 5$ : $P_i(\beta)=0$ for $i>\beta$ and $P_\beta(\beta)=1$; otherwise 
\begin{align*}
P_{0}(\beta) &= 3\cdot P(\beta-1)\\
P_{1}(\beta) &= 2 \cdot P(\beta-2) + \sum_{i' = 1}^{\beta-1} P_{i'}(\beta-1)\\
P_{2}(\beta) &=3 \cdot P(\beta-3)\\
P_{3}(\beta) &=2 \cdot P(\beta-4) + \sum_{i' = 1}^{\beta-3} P_{i'}(\beta-3)\\
P_{4}(\beta) &=3 \cdot P(\beta-5)
\end{align*}
We obtain the following first values for $P(\beta)$ :
\begin{center}
\begin{tabular}{cr|cccccc:ccccc}
\multicolumn{2}{c|}{\multirow{2}{*}{\normalsize$P_i(\beta)$}} &  \multicolumn{11}{c}{\normalsize$\beta$} \\
	\multicolumn{2}{c|}{} & $0$ & $1$ & $2$ &$3$ & $4$ &$5$ & $6$&$7$&$8$&$9$&$10$\\
	\hline
\multirow{7}{*}{\normalsize$i $}	&$0$& $\mathbf{1}$ & $3$& $12$& $48$ & $192$ &$768$&\\
	&$1$&. &$\mathbf{1}$ & $3$& $12$& $48$ &$192$\\
 &$2$& .&.& $\mathbf{1}$ & $3$ &$12$&$48$\\
 &$3$& .&.&.  &$\mathbf{1}$&$3$&$12$\\
 &$4$& .&.& .&.&$\mathbf{1}$&$3$\\
 &$5$& .&.&  .&.&.&$\mathbf{1}$&\\
 &$6$& .&.&  .&.&.&.&$\mathbf{1}$&\\
 \hline
 \multicolumn{2}{c|}{\normalsize$P(\beta)$}& $1$& $4$ & $16$ & $64$ & $256$ & $1\,024$
\end{tabular}
\end{center}
Note that $P(\beta) = |\Sigma|^\beta$. The dashed bar indicates the limit up to which the above equations are valid.
\end{example}

\paragraph{Case $\beta=m$.} A $m$-postmer $x=b_1\cdots b_m$ contain only one $m$-mer, itself, that must be $\geq w$. Let $0\leq i < m$ be the length of the common prefix between $x$ and $w$. Then $b_{i+1}>a_{i+1}$ and the choice of $b_{i+2},\dots, b_m$ is free; therefore, for all $0\leq i < m$, $P_i(m) = \varphi_>(a_{i+1}) \cdot |\Sigma|^{k-(i+1)}$. Besides, $P_m(m)=1$.

\begin{definition}\label{def:rank}
We denote by $\Phi_>(w)$ the number in base $|\Sigma|$ given by the vector $(\varphi_>(a_1),\dots,\varphi_>(a_m))$.
\end{definition}

From the previous discussion, we have $P(m) = 1 + \Phi_>(w)$.

\begin{remark}\label{rmk:rank}
$\Phi_>(w)$ is also the number of $m$-mer strictly greater than $w$; as well as their rank among all $m$-mers, in decreasing order. $|\Sigma^m|-\Phi_>(w)$ provides the rank in increasing order.
\end{remark}

\begin{example}
With $w_2=\ch{ACACAC}$, we have $\Phi_>(w_2)=3822$, associated to the vector $(3,2,3,2,3,2)$ --- in light of Definition~\ref{def:rank}. We can fill in the column $\beta=6$ of the first values of $P(\beta)$.
\begin{center}
\begin{tabular}{cr|cccccc:c:cccc}
\multicolumn{2}{c|}{\multirow{2}{*}{\normalsize$P_i(\beta)$}} &  \multicolumn{11}{c}{\normalsize$\beta$} \\
	\multicolumn{2}{c|}{} & $0$ & $1$ & $2$ &$3$ & $4$ &$5$ & $6$&$7$&$8$&$9$&$10$\\
	\hline
\multirow{7}{*}{\normalsize$i $}	&$0$& $\mathbf{1}$ & $3$& $12$& $48$ & $192$ &$768$& $3\,072$\\
	&$1$&. &$\mathbf{1}$ & $3$& $12$& $48$ &$192$&$512$\\
 &$2$& .&.& $\mathbf{1}$ & $3$ &$12$&$48$& $192$\\
 &$3$& .&.&.  &$\mathbf{1}$&$3$&$12$&$32$\\
 &$4$& .&.& .&.&$\mathbf{1}$&$3$& $12$\\
 &$5$& .&.&  .&.&.&$\mathbf{1}$& $2$\\
 &$6$& .&.&  .&.&.&.&$\mathbf{1}$\\
 \hline
 \multicolumn{2}{c|}{\normalsize$P(\beta)$}& $1$& $4$ & $16$ & $64$ & $256$ & $1\,024$ & $3\,823$
\end{tabular}
\end{center}
\end{example}

\subsection{Subsequent values}

Now, suppose $\beta >m$. Let $x=b_1\cdots b_\beta$ be a $\beta$-postmer, and let $0\leq i\leq m$ be the size of the common prefix between $x$ and $w$. We proceed by disjunction of cases.

\paragraph{Case $i=0$.} Necessarily, $b_1>a_1$ --- with $\varphi_>(a_1)$ possible choices. Therefore, we have:
\begin{proposition}\label{prop:postmer_0}
$P_0(\beta) = \varphi_>(a_1)\cdot P(\beta-1)$.
\end{proposition}

\paragraph{Case $1\leq i\leq m$.} We have $x=a_1\cdots a_i b_{i+1}\cdots b_\beta$. Since all $m$-mers of $x$ must be $\geq w$, we retrieve a system of equation similar to the one of Section~\ref{ss:general_case}, namely :

\begin{equation}
a_j\cdots a_i b_{i+1}\cdots b_{m-1+j} \geq a_1\cdots a_m\tag{$i,j$; bis}  \label{eq:i_j_bis}
\end{equation}
and
\begin{equation}
b_{i+1}\cdots b_{i+m} \geq a_1\cdots a_m. \tag{$i$; bis}\label{eq:i_bis}
\end{equation}

However, those equations are not necessarily defined for all values of $1\leq j \leq i$. Indeed, one must have $\beta \geq m-1+j$ for \eqref{eq:i_j_bis} to be defined, and $\beta \geq i+m$ for \eqref{eq:i_bis}. Note that since $j\leq i$, $\beta \geq m+i \implies \beta \geq m-1+j$. Conversely, if there exists $j$ so that $\beta < m-1+j$, then surely $\beta <m+i$.\footnote{From $\beta < m-1+j$ we get $\beta < m-1+i$, i.e. $\beta \leq m+i$. But $\beta = m+i$ leads to $j>i+1$ which is absurd.}

Since $i\leq m < \beta$, there exists at least one letter to choose, $b_{i+1}$. We retrieve similar conclusions to the one of Proposition~\ref{prop:conditions_on_next_letter}, namely: 
\begin{enumerate}[label=(\alph*'')]
    \item $b_{i+1} > a_{i+1}$ (if $i=m$, we use $a_{m+1}=\varepsilon$) --- since \eqref{eq:prefix} also holds;
    \item $b_{i+1} \geq a_1$ (if and only if $\beta \geq m+i$);
    \item $b_{i+1}\geq a_{i-j+2}$ for each $j$ so that $\mathbf{R}_{i,j}^=$ (if and only if $\beta \geq m-1+j$).
\end{enumerate}

Normally, the value of $\mathbf{T}_i(b_{i+1})$ would determine which recursive case to consider; but we must take into account that some equations might not be defined. We adapt Definition~\ref{def:prefix_letter_vector} as follows.

\begin{definition}\label{def:prefix_letter_vector_postmer}
For $1\leq i\leq m$ and $a\in \Sigma$, we define $\widetilde{\mathbf{T}_i}(a,\beta)= \mathbf{T}_i(a) \cdot [\beta \geq m-1 +\mathbf{T}_i(a)]$.
\end{definition}

This definition is consistent with the previous discussion, since the we retrieve that $b_{i+1}\geq a$ for some letter $a$ if and only if $\widetilde{\mathbf{T}_i}(a,\beta)\neq 0$. Note that for $a_1$, in the event that $\mathbf{T}_i(a_1)= i+1$, we retrieve $\beta\geq m-1+\mathbf{T}_i(a_1) = m+i$.

 If $\widetilde{\mathbf{T}_i}(b_{i+1},\beta)=0$, then no further constraint is put on $b_{i+2},\dots$; otherwise we must consider the recursive case of $(\beta - \widetilde{\mathbf{T}_i}(b_{i+1},\beta)+1)$-postmers sharing with $w$ a prefix of size at least $i-\widetilde{\mathbf{T}_i}(b_{i+1},\beta)+2$.
 
 We define, similarly to \eqref{eq:a_max}, the most stringent condition for (b'') and (c''); that is $b_{i+1}\geq \widetilde{a_{\max}}(i,\beta)$ where
 
 \begin{equation}\label{eq:a_max_tilde}
 \widetilde{a_{\max}}(i,\beta) =\max \left(\lbrace \varepsilon \rbrace \cup \left\lbrace a \in \Sigma : \widetilde{\mathbf{T}_i}(a,\beta)\neq 0\right\rbrace\right)
 \end{equation}

Contrary to $a_{\max}(i)$, $a_1$ might not belong to the set used to define $\widetilde{a_{\max}}(i,\beta)$ --- hence the presence of $\varepsilon$ to ensure that $\widetilde{a_{\max}}(i,\beta)$ is well defined. Both the values for $\widetilde{\mathbf{T}_i}(a,\beta)$ and $\widetilde{a_{\max}}(i,\beta)$ can be precomputed --- see Appendix~\ref{app:computation_autocorrelation_matrix}.

In the end, it also comes down to choosing one of these two options for $b_{i+1}$: 
\begin{itemize}
    \item $b_{i+1} = \widetilde{a_{\max}}(i,\beta)$ --- possible only if $\widetilde{a_{\max}}(i,\beta)>a_{i+1}$. In this case, we have $\mathbf{T}_i(b_{i+1})\neq 0$ and we get to the usual recursive cases;
    \item $b_{i+1}>\widetilde{a_{\max}}(i,\beta)$ and $b_{i+1}>a_{i+1}$. By definition of $\widetilde{a_{\max}}(i,\beta)$, we have $\mathbf{T}_i(b_{i+1})=0$ and no further conditions are imposed on the letters $b_{i+2},\dots$. There are $\min( \varphi_>(a_{i+1}),\varphi_>(\widetilde{a_{\max}}(i,\beta)))$ choices for $b_{i+1}$ with this option. Recall that $\varphi_>(\varepsilon)=|\Sigma|$.
\end{itemize}

\begin{example}\label{ex:postmer_prefix_vectors}
With $w_2=\ch{ACACAC}$, $a_1=\ch{A}$ and we have :
\begin{center}
 \begin{tabular}{c|c|cccc|cc}
    $i$ & $a_{i+1}$ & $\widetilde{\mathbf{T}_i}(\ch{A},\beta)$ & $\widetilde{\mathbf{T}_i}(\ch{C},\beta)$ & $\widetilde{\mathbf{T}_i}(\ch{G},\beta)$ & $\widetilde{\mathbf{T}_i}(\ch{T},\beta)$ & $\widetilde{a_{\max}}(i,\beta)$ & $\widetilde{a_{\max}}(i,\beta) > a_{i+1}$ ?\\
    \hline
    $1$ & $\ch{C}$& $2 \cdot [\beta\geq 7]$ & $0$ &$0$&$0$& $\begin{cases} \varepsilon & \text{if } \beta < 7;\\
    \ch{A} & \text{otherwise.}\end{cases}$ & \xmark\\
    \hline
    $2$ & $\ch{A}$ &$3 \cdot [\beta\geq 8]$  &$0$&$0$&$0$& $\begin{cases} \varepsilon & \text{if } \beta < 8;\\
    \ch{A} & \text{otherwise.}\end{cases}$& \xmark\\
    \hline
    $3$ & $\ch{C}$ & $4\cdot [\beta\geq 9]$& $3\cdot [\beta \geq 8]$&$0$&$0$&  $\begin{cases} \varepsilon & \text{if } \beta < 8;\\
    \ch{C} & \text{otherwise.}\end{cases}$ & \xmark\\
    \hline
    $4$ & $\ch{A}$& $3\cdot [\beta\geq 8]$&$0$&$0$&$0$& $\begin{cases} \varepsilon & \text{if } \beta < 8;\\
    \ch{A} & \text{otherwise.}\end{cases}$ & \xmark\\
    \hline
    $5$ & $\ch{C}$ & $6 \cdot [\beta\geq 11]$  & $3\cdot [\beta\geq 8]$ &$0$&$0$&  $\begin{cases} \varepsilon & \text{if } \beta < 8;\\
    \ch{C} & \text{otherwise.}\end{cases}$ & \xmark\\
    \hline
    $6$ & $\varepsilon$ & $3\cdot [\beta \geq 8]$ & $0$& $0$& $0$ & $\begin{cases} \varepsilon & \text{if } \beta < 8;\\
    \ch{A} & \text{otherwise.}\end{cases}$ & $\checkmark \iff \beta \geq 8$
\end{tabular}
\end{center}
Note that $a_{i+1}> \widetilde{a_{\max}}(i,\beta)$ implies that $\min( \varphi_>(a_{i+1}),\varphi_>(\widetilde{a_{\max}}(i,\beta)))=\varphi_>(a_{i+1})$ so the equations associated to $w_2$ will not depend at all on $\widetilde{a_{\max}}(i,\beta)$ except for $i=6$.
\end{example}

Finally, we have proven the following.

\begin{theorem}\label{th:postmer_general}
For $\beta >m$ and all $1\leq i\leq m$,
\begin{align*}
P_i(\beta) &= \min(\varphi_>(a_{i+1}),\varphi_>(\widetilde{a_{\max}}(i,\beta))) \cdot P(\beta-(i+1)) \\
&+ [\widetilde{a_{\max}}(i,\beta)>a_{i+1}] \cdot \sum_{i' = i-\widetilde{\mathbf{T}_i}(\widetilde{a_{\max}}(i,\beta),\beta)+2}^{m} P_{i'}(\beta-\widetilde{\mathbf{T}_i}(\widetilde{a_{\max}}(i,\beta),\beta)+1).
\end{align*}

\end{theorem}

\begin{example}
With $w_2=\ch{ACACAC}$, we have the following system of equations, with $\beta \geq m+1=7$:
\begin{align*}
P_0(\beta) &= 3\cdot P(\beta-1)\\
P_1(\beta)  &= 2 \cdot P(\beta-2) \\
P_2(\beta)  &=3 \cdot P(\beta-3)\\
P_3(\beta)  &=2 \cdot P(\beta-4)\\
P_4(\beta)  &=3 \cdot P(\beta-5)\\
P_5(\beta)  &=2\cdot P(\beta-6)
\end{align*}
As well as $P_6(7)=4\cdot P(0)$; and, for $\beta \geq 8$,
\begin{align*}
P_6(\beta) &= 3\cdot P(\beta-7) + P_{5}(\beta-2)+ P_{6}(\beta-2).
\end{align*}
\end{example}

\subsection{Computation and bounds}\label{ss:bounds_final}

Using the results of Section~\ref{subsec:postmer_first_values}, Proposition~\ref{prop:postmer_0}, as well as Theorem~\ref{th:postmer_general}, and \eqref{eq:postmer_def}, we have the following result.

\begin{proposition}\label{prop:postmer_computation}
$P(\beta)$ can be computed with dynamic programming, with time complexity $O(\beta\cdot |w|^2)$ and space complexity $O(\beta\cdot |w|)$.

\medskip
\noindent
$P_m(\beta+m)$ --- the value we are interested in, remember \eqref{eq:pi_calcul_beta_max}, p.~\pageref{eq:pi_calcul_beta_max} --- is computed in $O(\beta\cdot |w|^2 + |w|^3)$ time and $O(\beta \cdot |w| + |w|^2)$ space.
\end{proposition}
\begin{proof}
There are at most $\beta\cdot (m+1) = O(\beta\cdot |w|)$  values to compute, each one requiring at most $O(|w|)$ time. Plugging $\beta+m$ in place of $\beta$ give the second result.
\end{proof}

\begin{example}
With $w_2=\ch{ACACAC}$, we have the following values:
\begin{center}
\begin{tabular}{cr|cccccc:c:cccc}
\multicolumn{2}{c|}{\multirow{2}{*}{\normalsize$P_i(\beta)$}} &  \multicolumn{11}{c}{\normalsize$\beta$} \\
	\multicolumn{2}{c|}{} & $0$ & $1$ & $2$ &$3$ & $4$ &$5$ & $6$&$7$&$8$&$9$&$10$\\
	\hline
\multirow{7}{*}{\normalsize$i $}	&$0$& $\mathbf{1}$ & $3$& $12$& $48$ & $192$ &$768$& $3\,072$ & $11\,469$&$43\,419$& $164\,664$ & $624\,249$\\
	&$1$&. &$\mathbf{1}$ & $3$& $12$& $48$ &$192$&$512$& $2\,048$&$7\,646$& $28\,946$ & $109\,776$\\
 &$2$& .&.& $\mathbf{1}$ & $3$ &$12$&$48$& $192$ &$768$&$3\,072$ & $11\,469$ & $43\,419$\\
 &$3$& .&.&.  &$\mathbf{1}$&$3$&$12$&$32$ & $128$&$512$& $2\,048$ & $7\,646$\\
 &$4$& .&.& .&.&$\mathbf{1}$&$3$& $12$&$48$&$192$& $768$ & $3\,072$\\
 &$5$& .&.&  .&.&.&$\mathbf{1}$& $2$&$8$& $32$& $128$ & $512$\\
 &$6$& .&.&  .&.&.&.&$\mathbf{1}$&$4$&$15$ & $60$& $239$\\
 \hline
 \multicolumn{2}{c|}{\normalsize$P_\beta$}& $1$& $4$ & $16$ & $64$ & $256$ & $1\,024$ & $3\,823$ & $14\,473$ & $54\,888$ & $208\, 083$ & $788\,913$
\end{tabular}
\end{center}

\end{example}
To obtain bounds on $P(\beta)$, we follow a similar reasoning to that of the Section~\ref{subsec:antemer_bounds}. We have:
\begin{align*}
 0\, & \leq \sum_{i' = i-\widetilde{\mathbf{T}_i}(\widetilde{a_{\max}}(i,\beta),\beta)+2}^{m} P_{i'}(\beta-\widetilde{\mathbf{T}_i}(\widetilde{a_{\max}}(i,\beta),\beta)+1)    \\
 &\leq  \sum_{i' = 1}^{m} P_{i'}(\beta-\widetilde{\mathbf{T}_i}(\widetilde{a_{\max}}(i,\beta),\beta)+1)\\
&=P(\beta-\widetilde{\mathbf{T}_i}(\widetilde{a_{\max}}(i,\beta),\beta)+1) - P_0(\beta-\widetilde{\mathbf{T}_i}(\widetilde{a_{\max}}(i,\beta),\beta)+1)\\
&=P(\beta-\widetilde{\mathbf{T}_i}(\widetilde{a_{\max}}(i,\beta),\beta)+1) - \varphi_>(a_1)\cdot P(\beta-\widetilde{\mathbf{T}_i}(\widetilde{a_{\max}}(i,\beta),\beta)).
\end{align*}

Similarly to Definition~\ref{def:antemer_bounds}, we define the following sequences $P^-(\beta)$ and $P^+(\beta)$.

\begin{definition}
Let $P^-(\beta) = P^+(\beta) = |\Sigma|^\beta$ for $0\leq \beta  <m$, $P^-(m) = P^+(m)= 1+\Phi_>(w)$ --- remember Definition~\ref{def:rank} --- and
\begin{align*}
P^-(\beta) &= \varphi_>(a_1)\cdot P^-(\beta-1) + \sum_{i=1}^m \min(\varphi_>(a_{i+1}),\varphi_>(\widetilde{a_{\max}}(i,\beta)))\cdot P^-(\beta-(i+1))\\
P^+(\beta) &= \varphi_>(a_1)\cdot P^+(\beta-1) + \sum_{i=1}^m \Big(\min(\varphi_>(a_{i+1}),\varphi_>(\widetilde{a_{\max}}(i,\beta)))\cdot P^+(\beta-(i+1))\\
&+[\widetilde{a_{\max}}(i,\beta)>a_{i+1}] \cdot \big(P^+(\beta-\widetilde{\mathbf{T}_i}(\widetilde{a_{\max}}(i,\beta),\beta)+1) - \varphi_>(a_1)\cdot P^+(\beta-\widetilde{\mathbf{T}_i}(\widetilde{a_{\max}}(i,\beta),\beta))\big)\Big)
\end{align*}
\end{definition}

Since we are actually interested in bounding $P_m(\beta)$, we introduce two secondary sequences $P_m^-(\beta)$, and $P_m^+(\beta)$.

\begin{definition}
Let $P_m^-(m) = P_m^+(m)=1$, and for $\beta \geq m$,
\begin{align*}
P_m^-(\beta) &= \varphi_>(\widetilde{a_{\max}}(m,\beta)) \cdot P^-(\beta-(m+1)) \\
P_m^+(\beta) &= \varphi_>(\widetilde{a_{\max}}(m,\beta)) \cdot P^+(\beta-(m+1)) \\
&+ [\widetilde{a_{\max}}(m,\beta)\neq\varepsilon] \cdot \big(P^+(\beta-\widetilde{\mathbf{T}_i}(\widetilde{a_{\max}}(m,\beta),\beta)+1) - \varphi_>(a_1)\cdot P^+(\beta-\widetilde{\mathbf{T}_i}(\widetilde{a_{\max}}(m,\beta),\beta))\big)
\end{align*}
\end{definition}

We have the following result.

\begin{proposition}\label{prop:postmer_bounds}
\phantom{hey}
\begin{itemize}
    \item $\forall \beta\geq 0, P^-(\beta) \leq P(\beta) \leq P^+(\beta)$;  
    \item $\forall \beta\geq m, P_m^-(\beta) \leq P_m(\beta) \leq P_m^+(\beta)$.
\end{itemize}
Furthermore, both $P^-(\beta)$ and $P^+(\beta)$ can be computed with dynamic programming in $O(\beta)$ space and $O(\beta\cdot|w|)$ time, whereas $P_m^-(\beta)$ and $P^+_m(\beta)$ are computed in $O(1)$ space and $O(1)$ time, once $P^-$ and $P^+$ are computed.
\end{proposition}

\begin{example}\label{ex:valeurs_finales}
With $w_2=\ch{ACACAC}$, we start by bounding the values of $P(\beta)$ as follows:
\begin{center}
\small
\begin{tabular}{c|ccccccccccc}
$\beta$ & $0$ & $1$ & $2$ &$3$ & $4$ &$5$ & $6$&$7$&$8$&$9$&$10$\\
	\hline
	\normalsize$P^+(\beta)$& $1$ & $ 4$ & $ 16$ & $ 64$ & $ 256$ & $ 1\,024$ & $ 3\,823$ & $ 14\,473$ & $ 55\,636$ & $ 213\,319$ & $ 818\,287$\\
\normalsize$P(\beta)$& $1$ & $ 4$ & $ 16$ & $ 64$ & $ 256$ & $ 1\,024$ & $ 3\,823$ & $ 14\,473$ & $ 54\,888$ & $ 208\,083$ & $ 788\,913$\\
\normalsize$P^-(\beta)$& $1$ & $ 4$ & $ 16$ & $ 64$ & $ 256$ & $ 1\,024$ & $ 3\,823$ & $ 14\,473$ & $ 54\,885$ & $ 208\,062$ & $ 788\,797$\\
\end{tabular}
\end{center}
and the values of $P_m(\beta)$ as follows:
\begin{center}
\small
\begin{tabular}{c|ccccccccccc}
$\beta$ & $0$ & $1$ & $2$ &$3$ & $4$ &$5$ & $6$&$7$&$8$&$9$&$10$\\
	\hline
	\normalsize$P_m^+(\beta)$& $1$ & $ 4$ & $ 763$ & $ 3\,052$ & $ 12\,409$ & $ 47\,179$ & $ 181\,402$ & $ 694 \,657$ & $ 2\,663\,689$ & $ 10\,215\,016$ & $ 39\,174\,430$\\
\normalsize$P_m(\beta)$& $1$ & $ 4$ & $ 15$ & $ 60$ & $ 239$ & $ 956$ & $ 3\,823$ & $ 14\,473$ & $ 54\,888$ & $ 208\,083$ & $ 788\,913$\\
\normalsize$P_m^-(\beta)$&$1$ & $ 4$ & $ 12$ & $ 48$ & $ 192$ & $ 768$ & $ 3\,072$ & $ 11\,469$ & $ 43\,419$ & $ 164\,655$ & $ 624\,186$\\
\end{tabular}
\end{center}
\end{example}

Combining this result and Proposition~\ref{prop:antemer_bounds} with 
\begin{equation}
\pi_k(w) = \sum_{\beta=0}^{\beta_{\max}(w)} A(k-m-\beta) \cdot P_m(\beta+m), \tag{\ref{eq:pi_calcul_beta_max}}
\end{equation}
one can define $\pi_k^-(w)$ and $\pi_k^+(w)$ as
\begin{align*}
\pi_k^-(w) &= \max\left(1,\sum_{\beta=0}^{\beta_{\max}(w)} A^-(k-m-\beta) \cdot P_m^-(\beta+m)\right) \\
\pi_k^+(w) &= \min\left((\beta_{\max}(w)+1) \cdot |\Sigma|^{k-m},\sum_{\beta=0}^{\beta_{\max}(w)} A^+(k-m-\beta) \cdot P_m^+(\beta+m) \right)
\end{align*}

\begin{proposition}
For all $k\geq |w|$, $\pi^-_k(w) \leq \pi_k(w)\leq \pi^+_k(w)$.
\end{proposition}

\begin{example}
With $w_1=\ch{ACACAA}$, we obtain the following values :
\begin{center}
\small
\begin{tabular}{c|ccccccccccc}
$k$ & $6$ & $7$ & $8$ &$9$ & $10$ &$11$ & $12$&$13$&$14$&$15$&$16$\\
	\hline
	\normalsize$\pi^+_k(w_1)$& $1$& $7$&$24$&$93$&$353$&$1\,355$&$5\,195$&$19\,922$&$76\,384$&$292\,873$&$1\,122\,932$\\
	
\normalsize$\pi_k(w_1)$&  $1$& $7$&$24$&$93$&$351$&$1\,332$&$5\,049$&$19\,143$&$72\,576$&$275\,157$&$1\,043\,199$\\

\normalsize$\pi^-_k(w_1)$&  $1$& $7$&$23$&$86$&$327$&$1\,239$&$4\,698$&$17\,808$&$67\,495$&$255\,826$&$969\,659$\\
\end{tabular}
\end{center}
And with $w_2=\ch{ACACAC}$, we obtain the following values :
\begin{center}
\small
\begin{tabular}{c|ccccccccccc}
$k$ & $6$ & $7$ & $8$ &$9$ & $10$ &$11$ & $12$&$13$&$14$&$15$&$16$\\
	\hline
	\normalsize$\pi^+_k(w_2)$&$1$& $7$& $48$&$256$&$1\,280$&$6\,144$&$28\,672$&$131\,072$&$589\,824$&$2\,621\,440$&$11\,534\,336$\\
\normalsize$\pi_k(w_2)$& $1$&$7$&$38$&$191$&$911$&$4\,202$&$18\,923$&$82\,889$&$356\,478$&$1\,511\,583$&$6\,337\,559$\\
\normalsize$\pi^-_k(w_2)$&$1$&$7$&$35$&$170$&$795$&$3\,615$&$16\,110$&$69\,873$&$298\,273$&$1\,257\,505$&$5\,247\,521$\\
\end{tabular}
\end{center}
\end{example}

Finally, we have the following result.

\begin{theorem}\label{th:complexity_pik}
For all $k\geq m$, with $|w|=m$, $\pi_k(w)$ can be computed in $O(km)$ space and $O(km^2)$ time.

\medskip
\noindent
Furthermore, both $\pi_k^-(w)$ and $\pi_k^+(w)$ can be computed in $O(k)$ space and $O(km)$ time.
\end{theorem}
\begin{proof}
First, computing $\mathbf{R}$ --- see Appendix~\ref{app:computation_autocorrelation_matrix} --- is done in $O(m^2)$. Then, from \eqref{eq:pi_calcul_beta_max} and the definitions of $\pi^+$ and $\pi^-$, it is sufficient to calculate the values of $A(k-m)$ and $P_m(\beta_{\max}(w)+m)$  and store the intermediate values in memory. According to Propositions~\ref{prop:antemer_computation} and~\ref{prop:postmer_computation}, we have:
\begin{center}
\begin{tabular}{c|cc}
     & Space & Time \\
     \hline
 $A(k-m)$   & $O((k-m)\cdot m)$ & $O((k-m)\cdot m^2)$\\
 $P_m(\beta_{\max}(w)+m)$ & $O(\beta_{\max}(w)\cdot m+m^2)$ & $O(\beta_{\max}(w)\cdot m^2 + m^3)$
\end{tabular}
\end{center}
Since $\beta_{\max}(w)\leq k-m$, and since the computation of $\pi_k(w)$ from the $A(\cdot)$'s and $P_m(\cdot)$'s is $O(k-m)$, we have the result.

Concerning $A^+$, $A^-$, $P^+$ and $P^-$, one can notice from Propositions~\ref{prop:antemer_bounds} 
and~\ref{prop:postmer_bounds} that the complexities are simply divided by a factor $m$ compared to the exact values of $A$ and $P$.
\end{proof}

Note that incorporating hard rules on the particular structure of certain minimizers could speed up the calculation of $\pi_k(w)$ for them. For example, if the minimizer starts with $\max(\Sigma)$, we can immediately return $\pi_k(w)=1$ in $O(1)$ --- as can be observed in Figure~\ref{fig:partition_theory}.

\section{Numerical results}\label{sec:numerical}

The methods developed in this article have been implemented in Python and can be found on Github\footnote{At \url{https://github.com/fingels/minimizer_counting_function}.}, as well as the scripts and actual data used to produce the figures in this section.

Note that all results of this section use the DNA alphabet.

We recall that $\pi_k(w)$ count the number of $k$-mers (among all $|\Sigma|^k$ possible $k$-mers) that admits $w$ as their lexicographical minimizer.

\subsection{Theoretical partition of $k$-mers}\label{ss:partition_theorique}

Since we now know how to calculate $\pi_k(w)$ for any word $w$ and any $k\geq |w|$, we can explicitly calculate the theoretical distribution of $k$-mers across the partition formed by their minimizers.

For $m$ and $k\geq m$ fixed, we enumerate the set of $m$-mers and calculate, for each such $m$-mer $w$, the associated value $\pi_k(w)$. Please note that we do not need to enumerate the $k$-mers, but only the $m$-mers, for this computation. The results for $k=21, m=8$ and $k=31, m=10$ are given in Figure~\ref{fig:partition_theory}.

\begin{figure}[h!]
    \centering
    
    \begin{subfigure}[t]{0.49\textwidth}
        \centering
            \includegraphics[width=\textwidth]{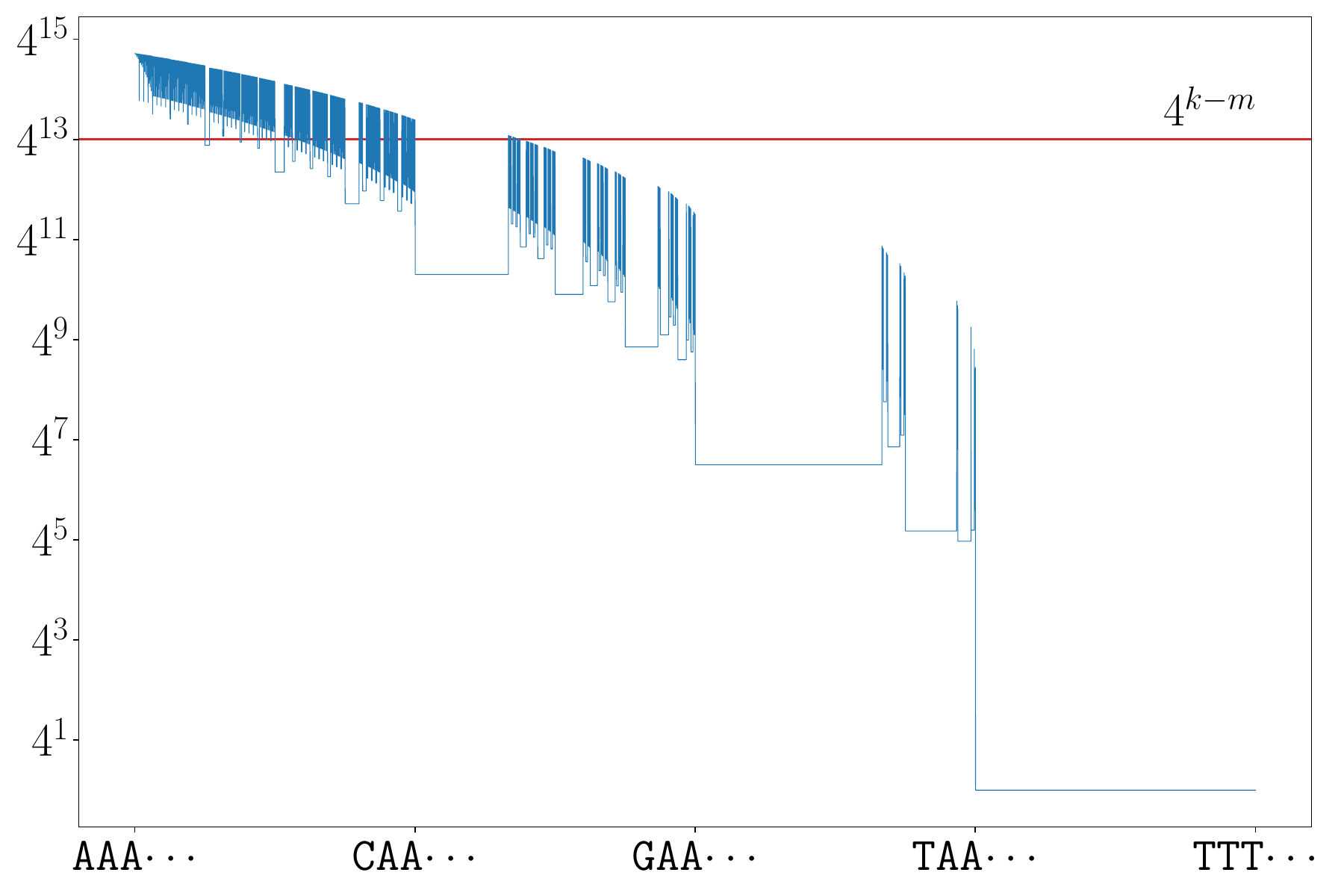}
    \caption{$k=21, m=8$}
    \label{fig:partition_theory:a}
    \end{subfigure}\hfill
       \begin{subfigure}[t]{0.49\textwidth}
        \centering
            \includegraphics[width=\textwidth]{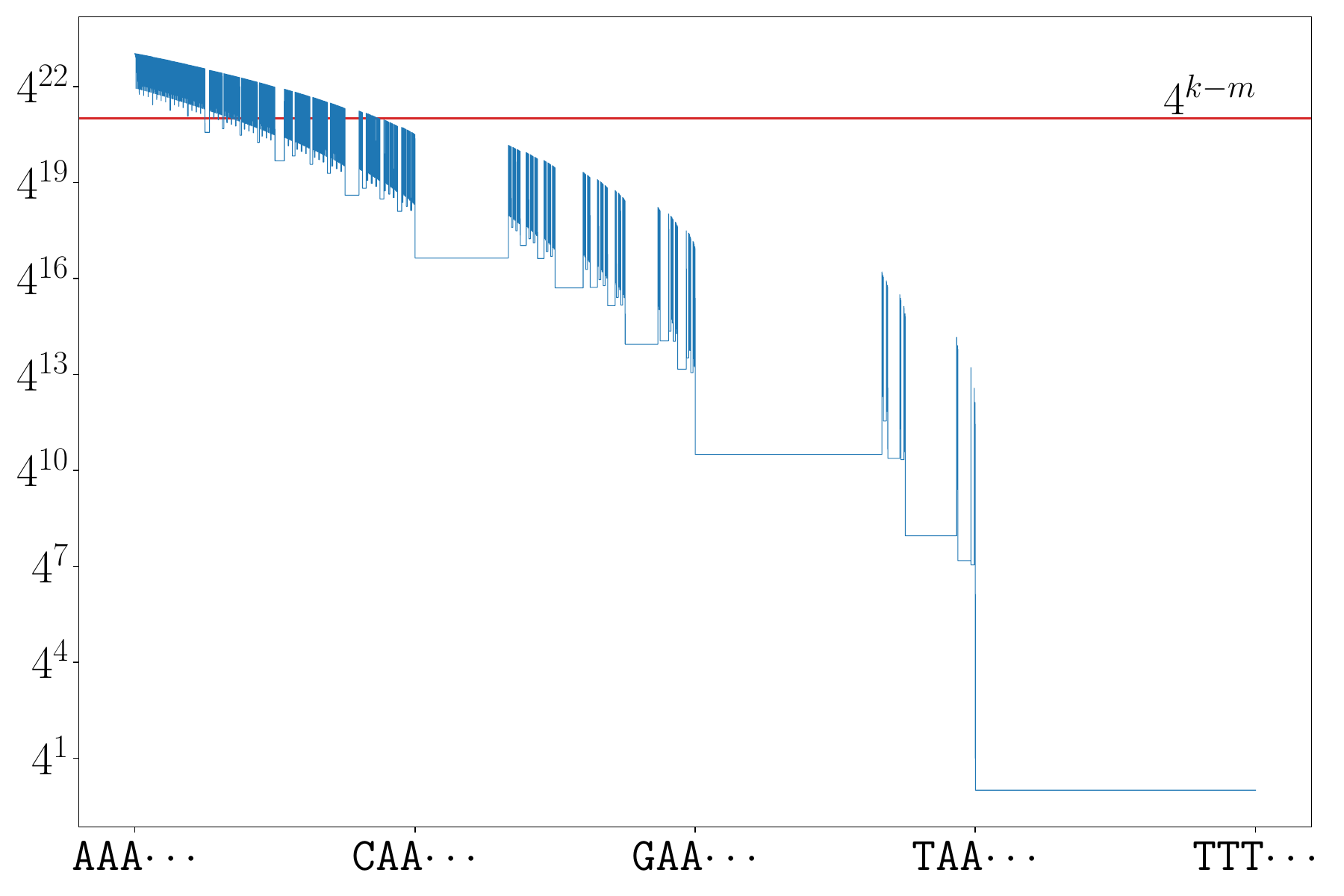}
        \caption{$k=31, m=10$}
    \label{fig:partition_theory:b}
    \end{subfigure}
   \caption{Theoretical partition of all $k$-mers into their lexicographical minimizers of size $m$, for different values of $(k,m)$. Minimizers are sorted lexicographically from left to right, and $\pi_k(\cdot)$'s values are given in log-scale. The horizontal red line stands for $4^{k-m}$, that is, the size of buckets one would expect if the $k$-mers were evenly distributed across all minimizers.}
\label{fig:partition_theory}
\end{figure}

As one can see, the two distributions are very similar, and support folklore: they are indeed extremely unbalanced. One can also observe plateaus (where many successive minimizers share the same value of $\pi_k$) and rapid oscillations (where $\pi_k$ cyclically increases and decreases over consecutive minimizers). These observations can be explained as follows:
\begin{itemize}
    \item The plateaus correspond to certain common prefixes (which are therefore naturally consecutive over a long series of minimizers) which strongly constrain the number of $k$-mers in their partition. For example, consider the minimizer $w=\ch{CA}\cdots \ch{A}$. Since $\ch{A}<\ch{C}$, one extra letter after $w$ is sufficient to change the minimizer, so $w$ is necessarily placed at the end of $k$-mer --- and $\beta_{\max}(w)=0$. The letters that can be placed before $w$, necessarily, can be freely chosen from $\ch{C},\ch{G}$ or $\ch{T}$. It is easy to see that the same constraints apply to all minimizers starting with $\ch{CA}\cdots$, hence the observed plateau, which also occurs for $\ch{GA}\cdots$ and $\ch{TA}\cdots$. In particular, all minimizers starting with a \ch{T} have exactly one $k$-mer in their bucket.
    \item The oscillations can be explained --- at least partially --- by the suffixes of the minimizers. Consider the followed examples $w_1=\ch{ACACAA}$ and $w_2=\ch{ACACAC}$. Because of the suffix $\ch{AA} < \ch{AC}$, we have $\beta_{\max}(w_1)=1$ for $k>m$ whereas $\beta_{\max}(w_2)=k-m$. So, since $w_2$ benefits from more starting positions than $w_1$ within a $k$-mer, it is natural to observe $\pi_k(w_2)>\pi_k(w_1)$ --- even though there are more antemers for $w_1$ than for $w_2$ --- remember Example~\ref{ex:antemers_values} !\footnote{It is actually not that counterintuitive : from Equation~\eqref{eq:pi_calcul_beta_max}, we have $\pi_k(w_1) = A(k-6)\cdot P_6(6) + A(k-7)\cdot P_6(7)$ --- with $P_6(6)=1$ and $P_6(7)\leq 4$ (since there is only one free letter). So $\pi_k(w_1)$ depends mainly on just two terms $A(k-6)$ and $A(k-7)$, while $\pi_k(w_2)$ depends on $k-m$ terms of the form $A(k-m-\beta)\cdot P_m(\beta+m)$, both $A(\cdot)$ and $P_m(\cdot)$ increasing exponentially. Therefore the surplus of antemers for $w_1$ does not compensate for this terms.} A few minimizers later, we move from \ch{ACACTT} to \ch{ACAGAA}, and we expect $\pi_k$ to decrease again because of the same argument.
\end{itemize}

\subsection{Approximating $\pi_k$}\label{ss:approximation}
We have seen in Theorem~\ref{th:complexity_pik} that calculating $\pi_k$ can be relatively costly in terms of time and space, especially when repeated over millions or billions of minimizers. It would be useful to have quick approximations of $\pi_k$, if one is interested in the order of magnitude rather than the exact value. Note that this section is intended to be more of a discussion, opening up avenues of research, than a truly exhaustive study answering the question.

\paragraph{Bounds on $\pi_k$} We start by considering the bounds $\pi_k^-(w)$ and $\pi_k^+(w)$ defined in Section~\ref{ss:bounds_final}, and how closely they match the theoretical value $\pi_k(w)$. In Figure~\ref{fig:bounds:a}, we show the results for $k=31$ and $m=10$, for all $m$-mers.

Despite the sometimes large differences between the bounds and the theoretical value, we can also see that the bounds are relatively close to the real values, and sometimes even equal, for certain minimizers. On this example, we encountered $820\,061$ (over $1\,048\,576$ minimizers) tight bounds --- i.e. at least one bound is equal to the actual value --- representing 78\% of all minimizers; among which they were $262\,144$ equal bounds --- 25\% of all minimizers. Those equal bounds correspond to all minimizers starting with a \ch{T}.

\begin{figure}[h!]
    \centering
    
    \begin{subfigure}[t]{0.49\textwidth}
        \centering
            \includegraphics[width=\textwidth]{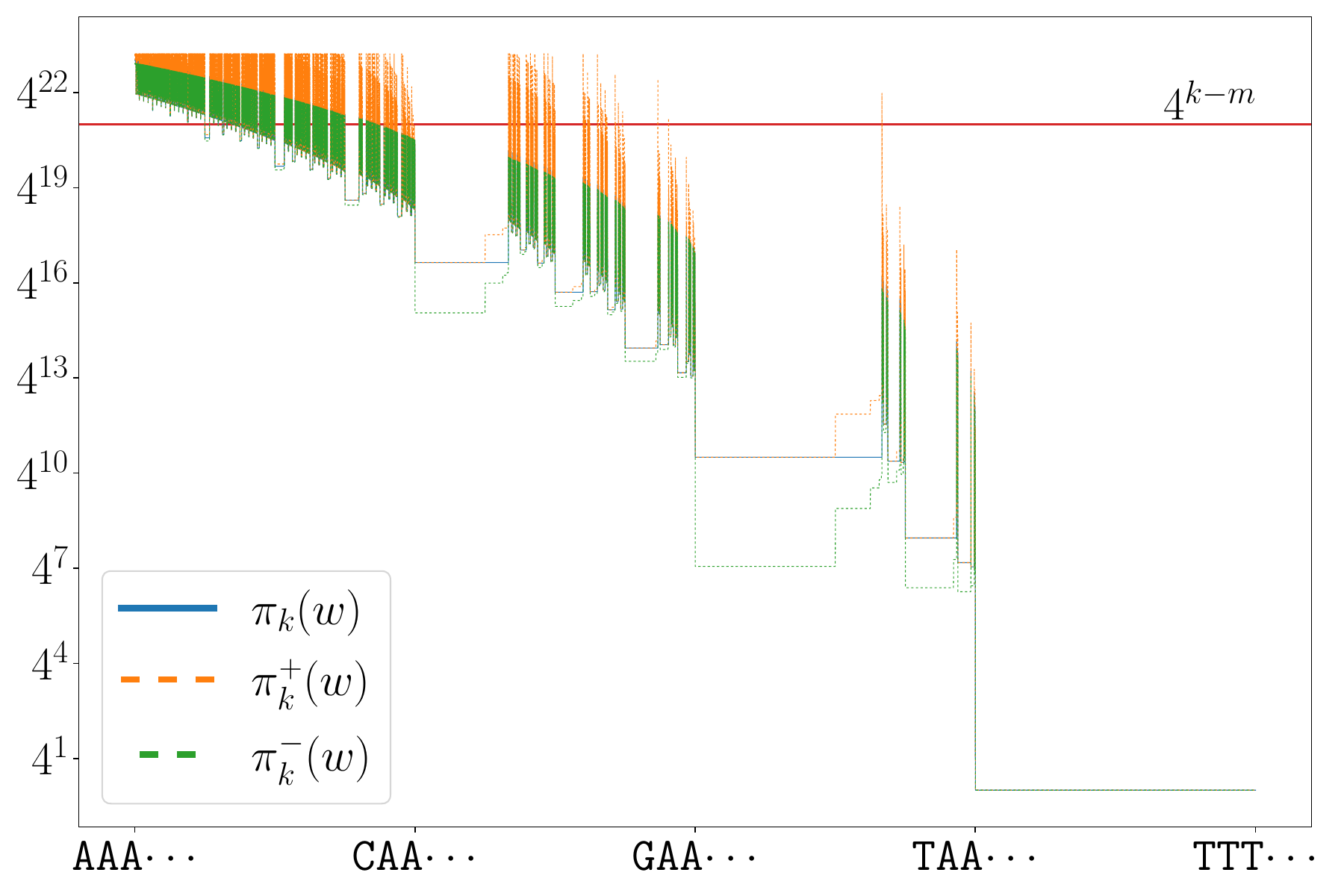}
    \caption{Upper and lower bounds on $\pi_k(w)$.}
    \label{fig:bounds:a}
    \end{subfigure}\hfill
       \begin{subfigure}[t]{0.49\textwidth}
        \centering
            \includegraphics[width=\textwidth]{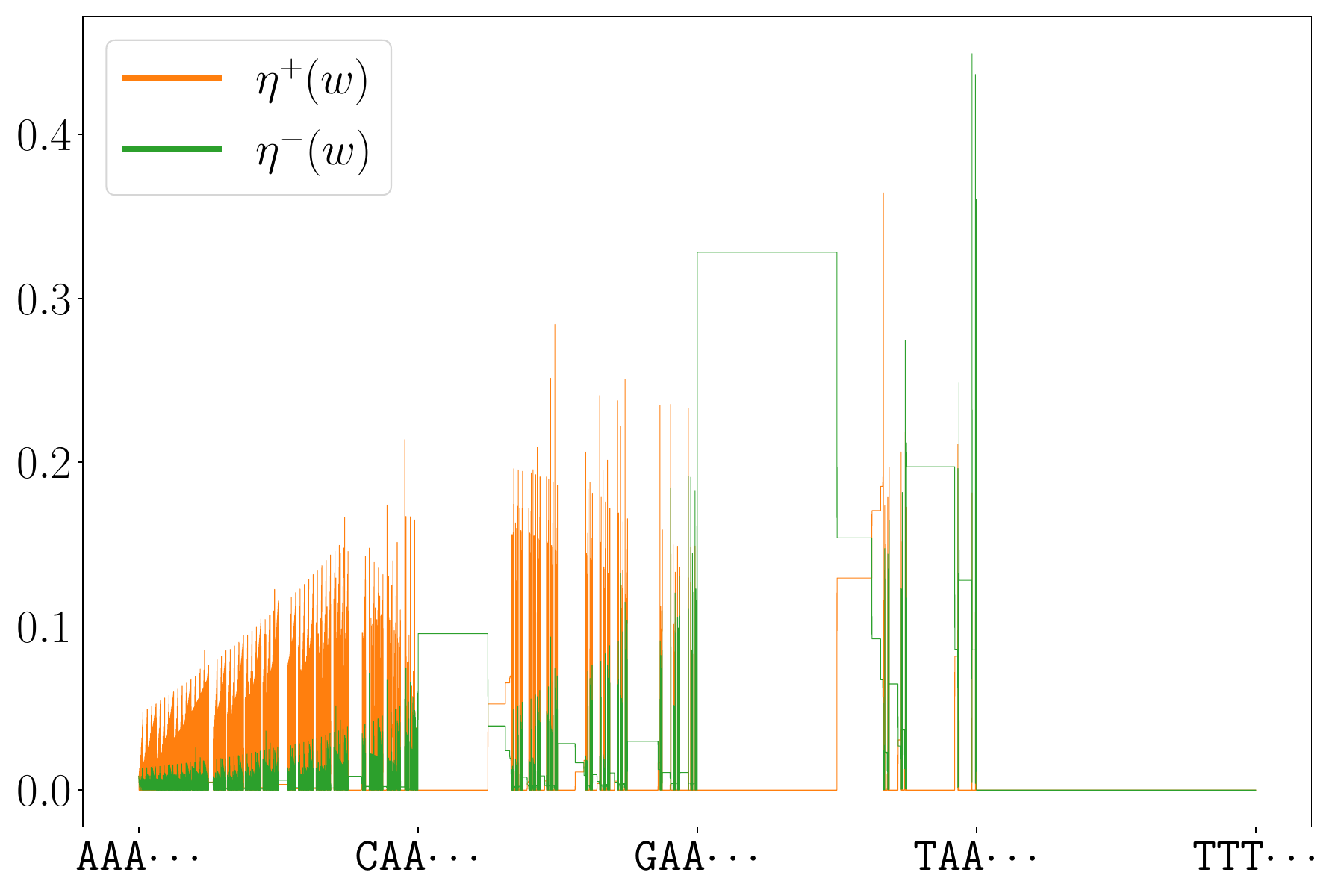}
        \caption{Relative errors $\eta^+(w)$ and $\eta^-(w)$}
    \label{fig:bounds:b}
    \end{subfigure}
   \caption{Relations between $\pi_k(w)$, $\pi_k^-(w)$ and $\pi_k^+(w)$ for $k=31, m=10$.}
\label{fig:bounds}
\end{figure}

Figure~\ref{fig:bounds:b} shows the relative errors in terms of order of magnitude, i.e. for each minimizer we computed 
$$\eta^+(w) = \frac{\log_{|\Sigma|}\pi_k^+(w)-\log_{|\Sigma|}\pi_k(w)}{\log_{|\Sigma|}\pi_k(w)}\quad \text{and}\quad\eta^-(w) = \frac{ \log_{|\Sigma|}\pi_k(w)-\log_{|\Sigma|}\pi_k^-(w)}{\log_{|\Sigma|}\pi_k(w)}.$$

We found that each time only one bound was tight, it was $\pi_k^+(w)$. However, $\pi_k^+(w)$ is not necessarily a good estimate of $\pi_k(w)$, since we can see that for a large number of minimizers, the relative error in order of magnitude is rather high. This is particularly striking if we compare Figures~\ref{fig:bounds:a} and~\ref{fig:partition_theory:b}. On the other hand, $\pi_k^-(w)$ is more consistent, especially on small minimizers.

Using $\pi_k^+(w)$ or $\pi_k^-(w)$ as an approximation of $\pi_k(w)$ saves a factor $m$ in temporal and spatial complexity, but, as we have just seen, at the cost of an approximation that is not necessarily reliable --- despite almost 80\% of minimizers having a tight bound. Are there competing approaches to approximate $\pi_k(w)$ quickly and with better precision?

\paragraph{Normalization} We have already pointed out that the curves of Figure~\ref{fig:partition_theory:a} and~\ref{fig:partition_theory:b} are particularly similar. To what extent?

For $k,m$ fixed, and any $m$-mer $w$, note that $0\leq \Phi_>(w)\leq |\Sigma|^m-1$ --- where $\Phi_>(w)$ is defined in Definition~\ref{def:rank} and correspond to the lexicographical rank of $w$ among all $m$-mers. Note also that 
$$\pi_k(w)\leq (k-m+1)\cdot |\Sigma|^{k-m}$$
by the upper bound of Lemma~\ref{lemma:trivial_upper_bound}, with $\beta_{\max}(w)\leq k-m$.

Therefore, we can represent any $m$-mer by a point $p(w)$ in the unit square $[0,1]\times[0,1]$ of coordinates
\begin{equation}
p(w)=\left(\frac{\Phi_>(w)}{|\Sigma|^m-1},\frac{\log_{|\Sigma|}\pi_k(w)}{(k-m) + \log_{|\Sigma|}(k-m+1)}\right)
\end{equation}
--- using $\log_{|\Sigma|}$ to avoid all the points being squashed on the $x$-axis.

Figure~\ref{fig:normalized} shows the results we obtained with severals values of $k$ and $m$. In Figure~\ref{fig:normalized:a}, we have considered small values of $m$: $m=2,3,4$. It can be seen that oscillations and plateaus appear as $m$ increases. In Figure~\ref{fig:normalized:b}, by setting $m=3$ --- to keep the curve simple --- we can see that by increasing $k$ some parts of the curve increase, others decrease, but locally the increase/decrease is consistent, and seems to diminish as $k$ increases. The same phenomenon can be seen in Figure~\ref{fig:normalized:c}, where this time we chose a larger value $m=10$. Restricting ourselves to high values of $k$, we compare the curves obtained for $m=8$ and $m=10$ in Figure~\ref{fig:normalized:d}. Although the values of $m$ are distinct, we can see that the curves seem to align almost perfectly.

\begin{figure}[h!]
    \centering
    
    \begin{subfigure}[t]{0.49\textwidth}
        \centering
            \includegraphics[width=\textwidth]{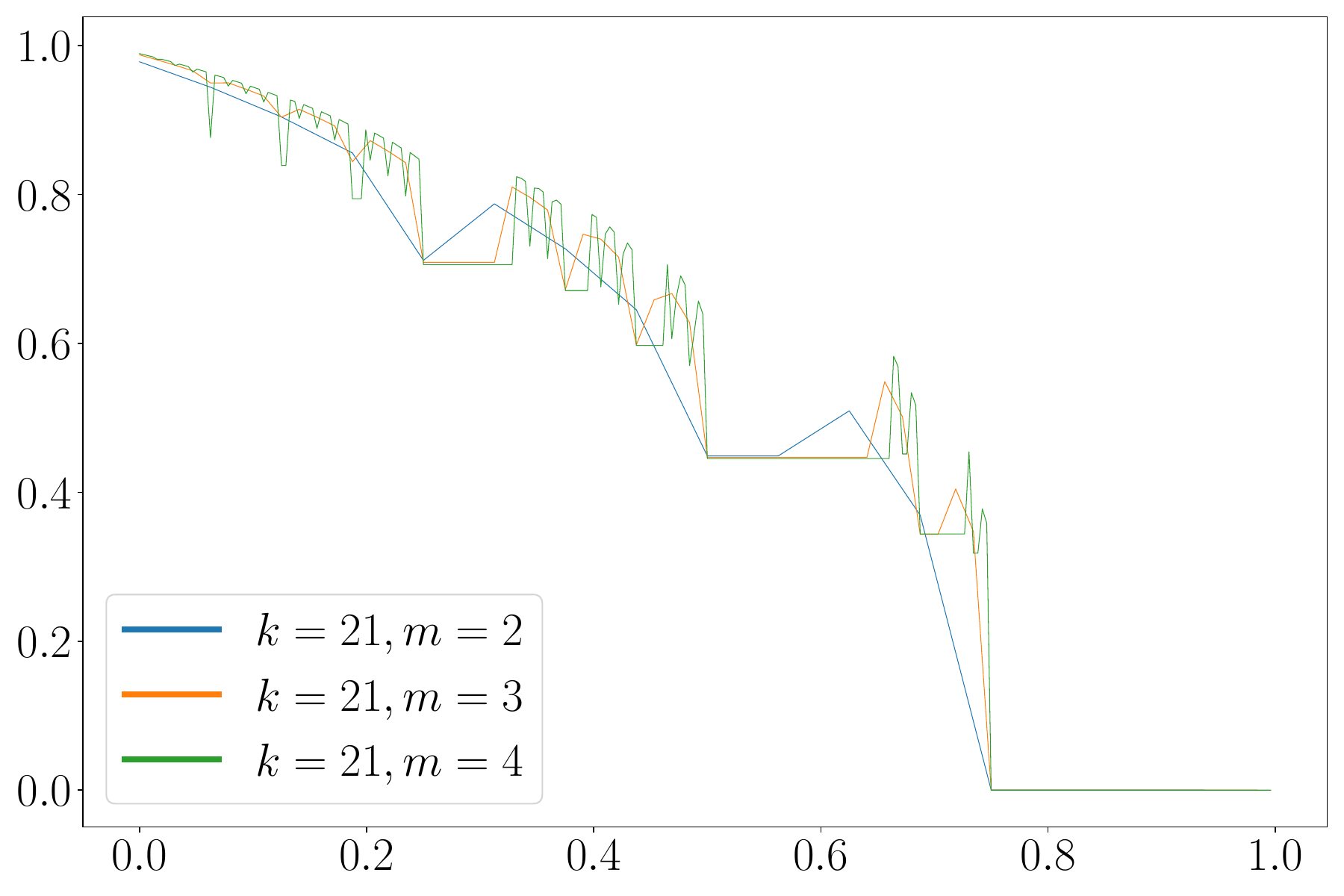}
    \caption{$k=21$, $m=2,3,4$}
    \label{fig:normalized:a}
    \end{subfigure}\hfill
       \begin{subfigure}[t]{0.49\textwidth}
        \centering
            \includegraphics[width=\textwidth]{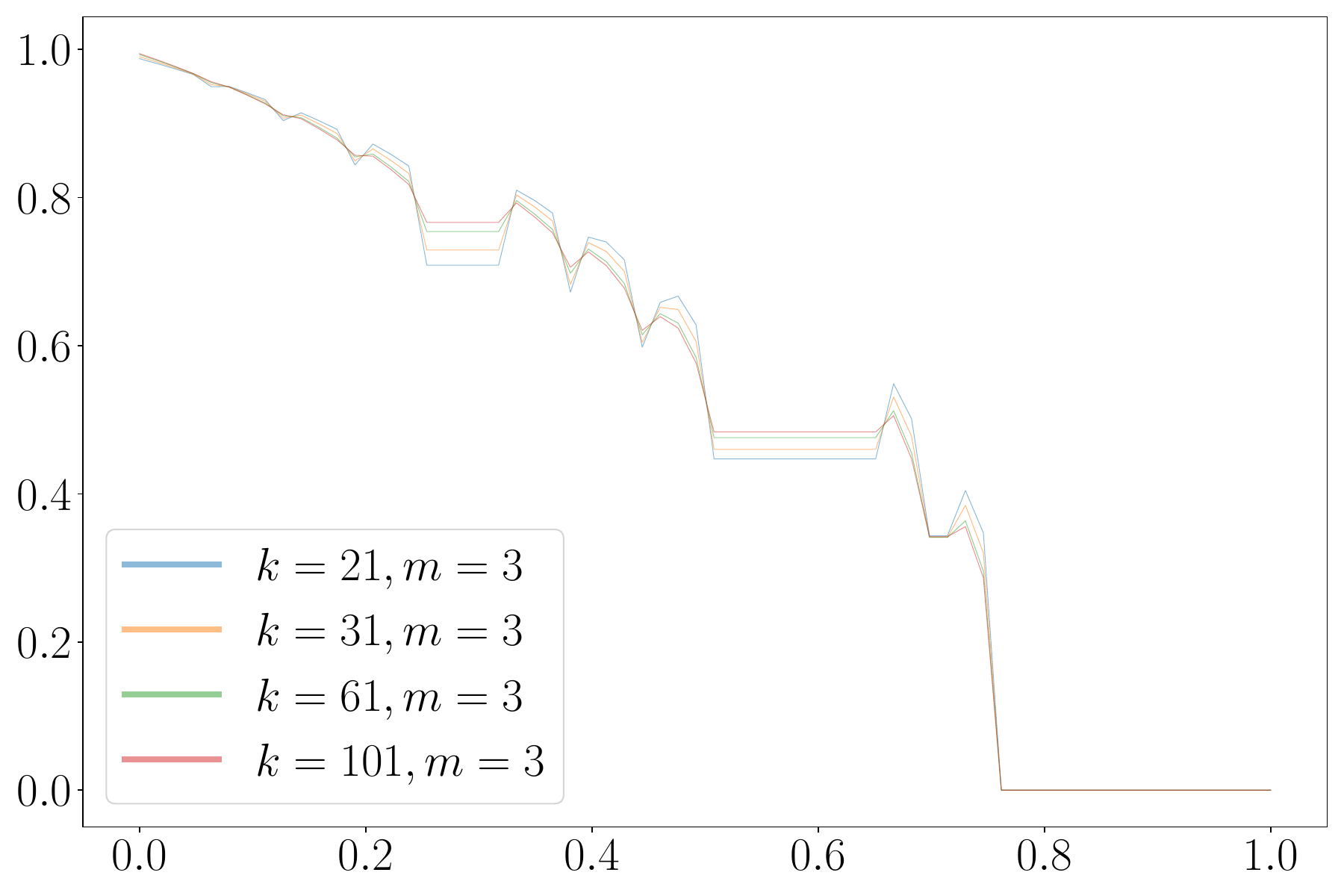}
        \caption{$m=3$, $k=21,31,61,101$}
    \label{fig:normalized:b}
    \end{subfigure}
    
        \begin{subfigure}[t]{0.49\textwidth}
        \centering
            \includegraphics[width=\textwidth]{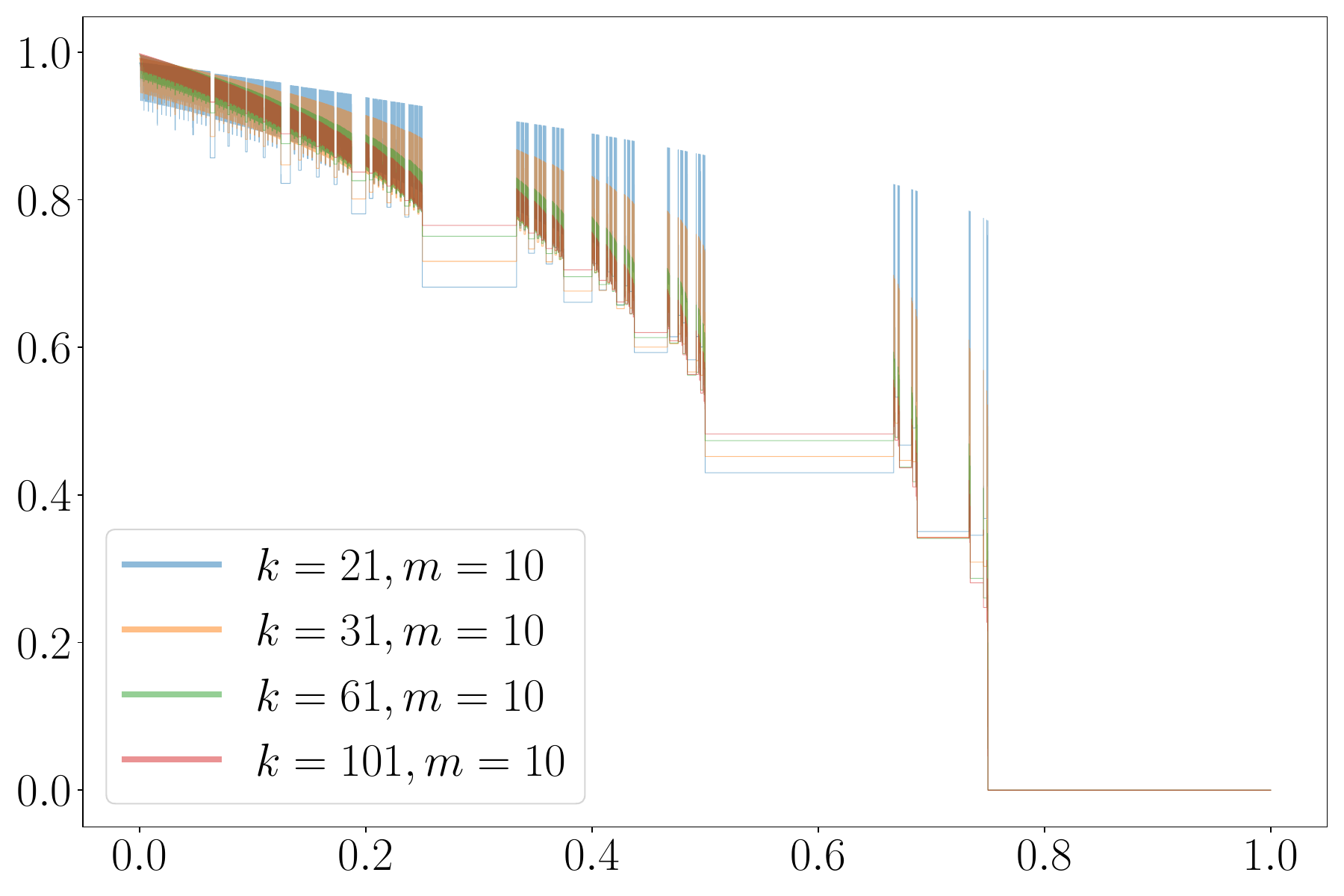}
    \caption{$m=10$, $k=21,31,61,101$}
    \label{fig:normalized:c}
    \end{subfigure}\hfill
      \begin{subfigure}[t]{0.49\textwidth}
        \centering
            \includegraphics[width=\textwidth]{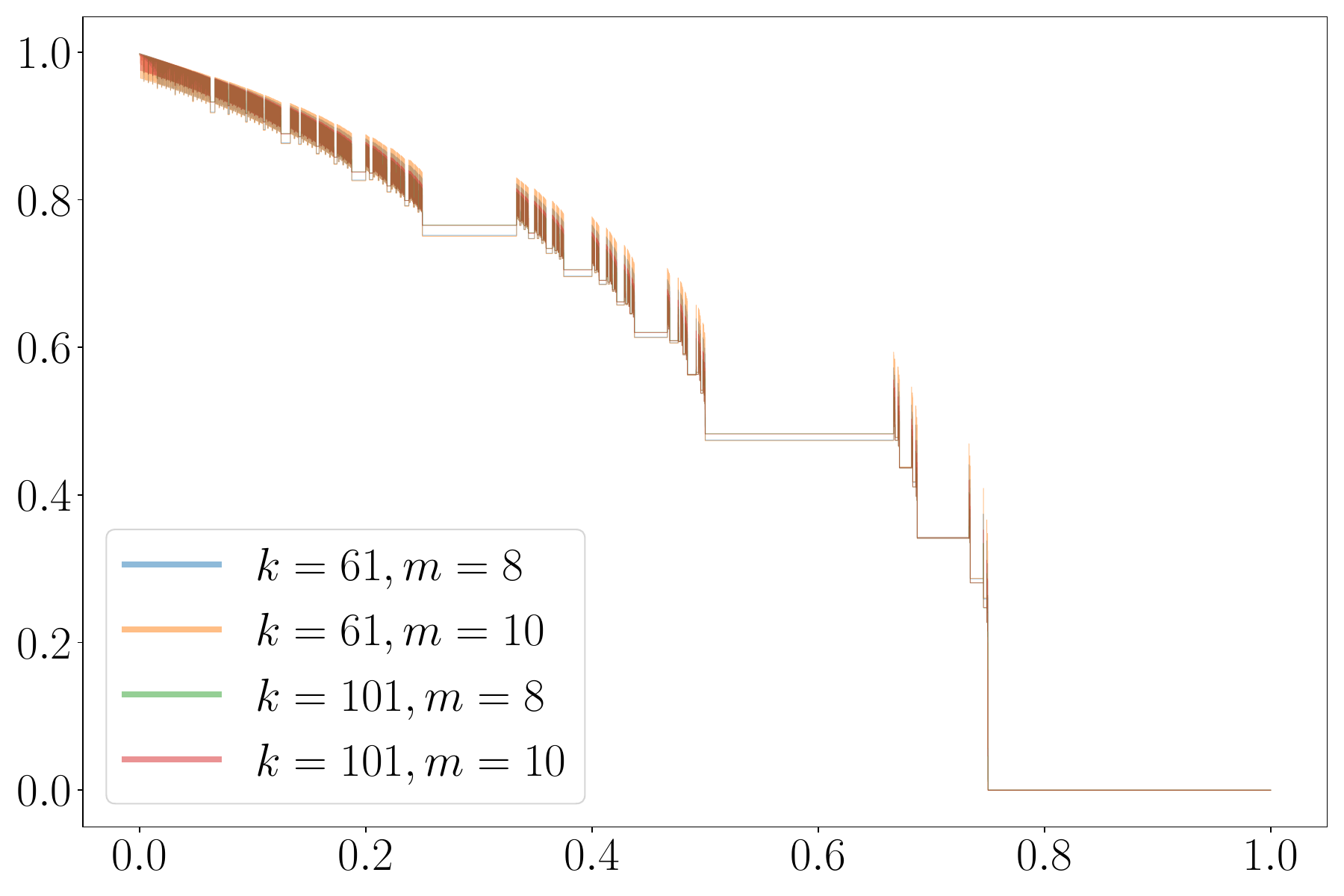}
        \caption{$m=8,10$, $k=61,101$}
    \label{fig:normalized:d}
    \end{subfigure}
   \caption{$p(w)$ for several values of $k,m$.}
\label{fig:normalized}
\end{figure}

These observations lead us to formulate the following conjecture. Let us define $f_{k,m}: [0,1]\to [0,1]$ the piecewise linear function such that for all $w\in\Sigma^m$, denoting $p(w)=(x,y)$, we have $f_{k,m}(x)=y$, as well as $f_{k,m}(0)=1$ and $f_{k,m}(1)=0$.

\begin{conjecture}\label{conj:courbe_limite}
There exist a function $f_\infty : [0,1]\to [0,1]$ so that $(f_{k,m})$ converges to $f_\infty$ in a sense to be determined as $k,m\to \infty$.
\end{conjecture}

If this is the case, and provided that the function $f_\infty$ is easy to calculate, then we can approximate $\pi_k(w)$ by
$$\log_{|\Sigma|}\pi_k(w)\approx f_\infty\left(\frac{\Phi_>(w)}{|\Sigma|^m-1}\right) \cdot \Big((k-m) + \log_{|\Sigma|} (k-m+1)\Big).$$
Proving (or disproving) such a conjecture is beyond the scope of this article. However, it is conceivable to employ computing power to calculate this curve for values of $k,m$ as high as possible, and then, for a given minimizer, to look for the point in this curve closest to $p(w)$ --- which we do not intend to do in this article either.

\paragraph{Asymptotics of $\pi_k(w)$} The values obtained for $\pi_k(\ch{ACACAA})$ and $\pi_k(\ch{ACACAC})$ in Example~\ref{ex:valeurs_finales} seem to indicate an exponential growth of $\pi_k(w)$ as $k$ increases, at least for these two examples. This is visually confirmed in log-scale in Figure~\ref{fig:asymptotics:examples}, alongside some other examples.

Consider the following linear regression model: 
$$\log_{|\Sigma|}\pi_k(w) = A_w\times k + B_w + \varepsilon_w.$$

Then Figure~\ref{fig:asymptotics:regression} shows the estimated values for $\widehat{A_w}$ using the least squares method, together with the associated coefficient of determination $R^2$, for all $m$-mers with $m=6$, computed on the values obtained with $m\leq k \leq m+100$. As can be seen, with the exception of minimizers starting with \ch{T} (for which we have $\pi_k(w)=1$), we obtain $R^2$ values extremely close to $1$, underlining the relevance of the model. The value of the slope is globally decreasing, with oscillations and plateaus which reproduce what was observed in Section~\ref{ss:partition_theorique}.

\begin{figure}[h!]
    \centering
    
    \begin{subfigure}[t]{0.49\textwidth}
        \centering
            \includegraphics[width=\textwidth]{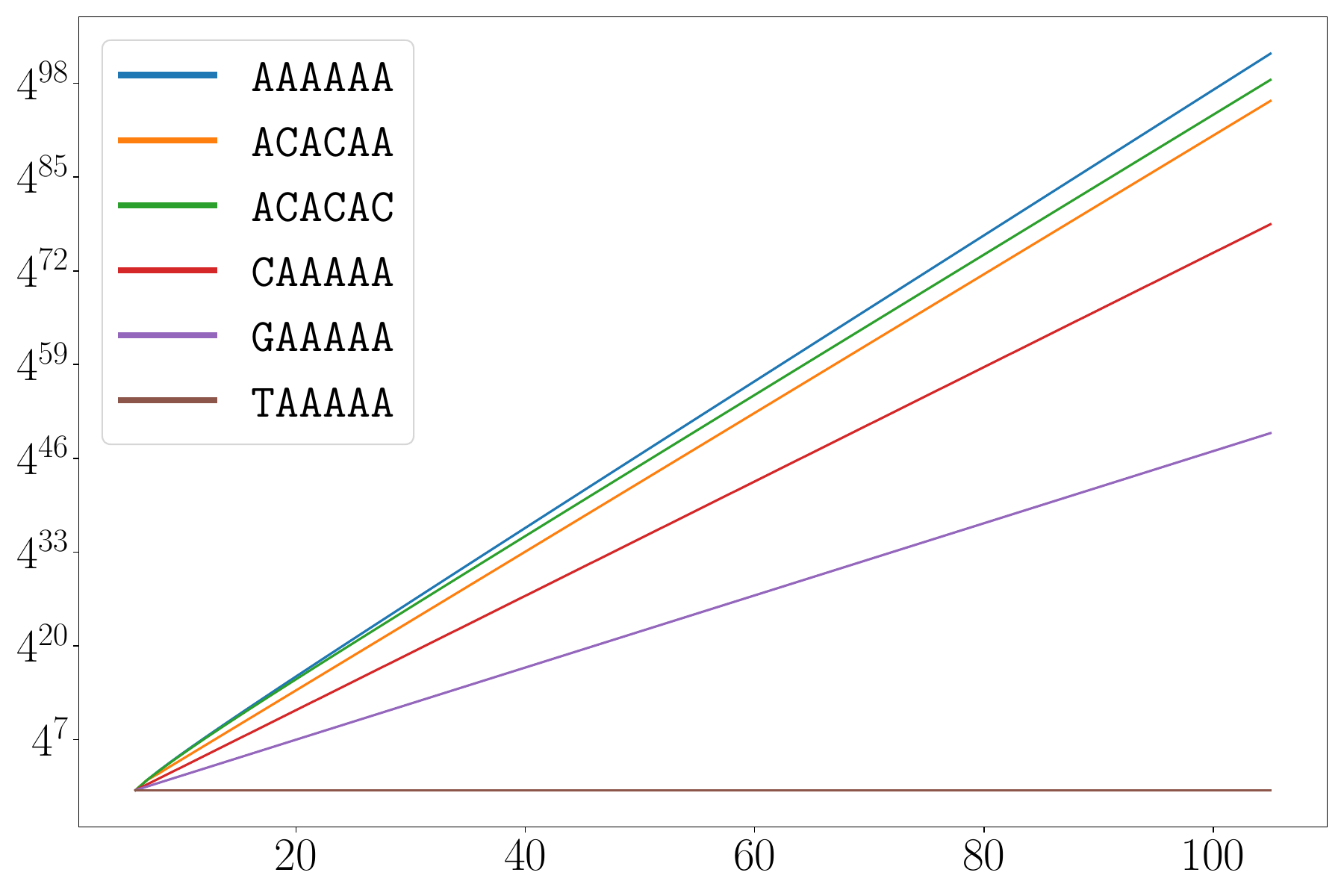}
    \caption{$\pi_k(w)$ as a function of $k$, for several choices of $w$}
    \label{fig:asymptotics:examples}
    \end{subfigure}\hfill
       \begin{subfigure}[t]{0.49\textwidth}
        \centering
            \includegraphics[width=\textwidth]{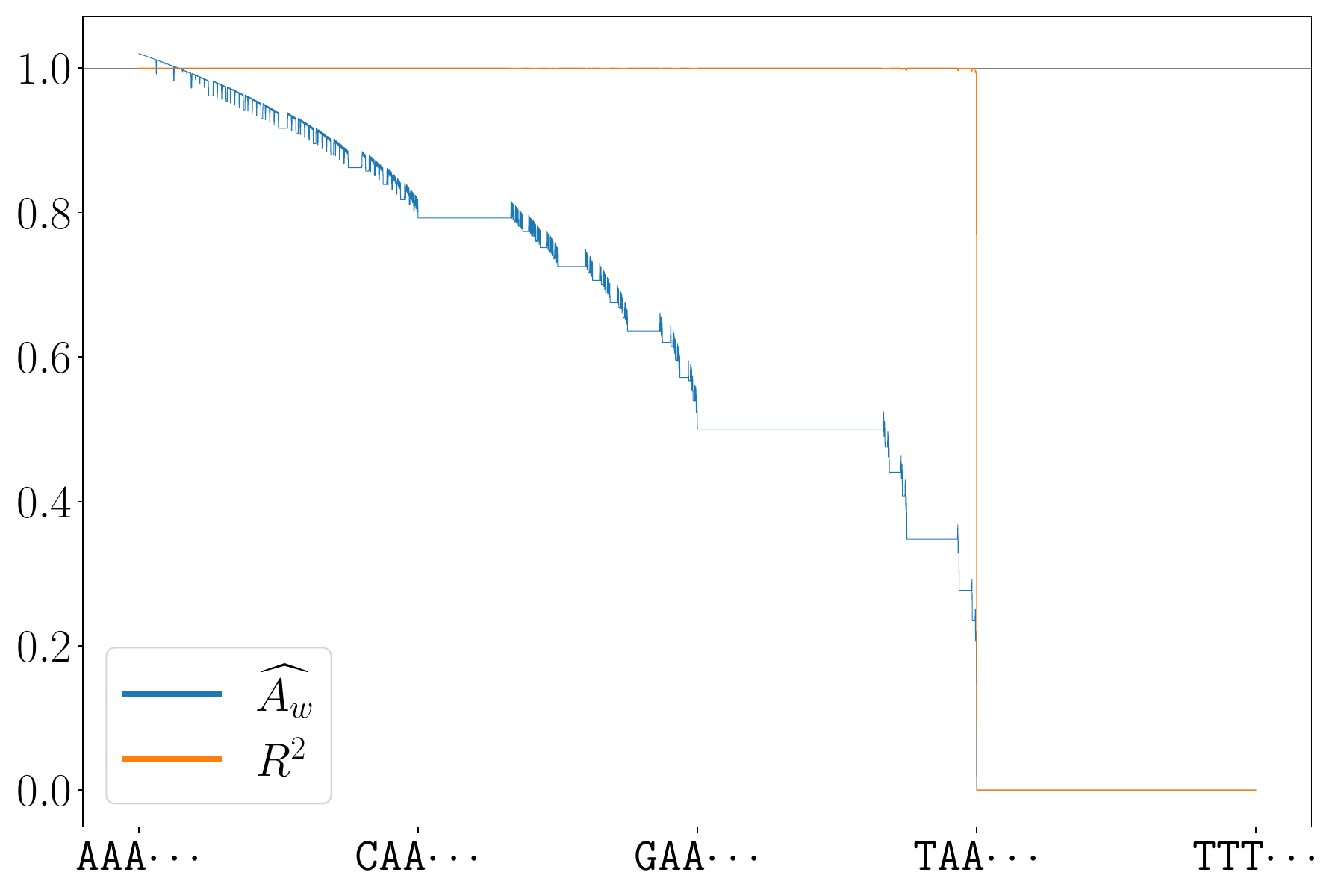}
 
        \caption{Estimation of $\widehat{A_w}$ and associated $R^2$ for all possible $6$-mer $w$, sorted lexicographically}
    \label{fig:asymptotics:regression}
    \end{subfigure}
   \caption{Asymptotics of $\pi_k(w)$.}
\label{fig:asymptotics}
\end{figure}

Extrapolating from these results, we formulate the following conjecture.

\begin{conjecture}\label{conj:asymptotics}
For any $m$-mer $w\in \Sigma^m$ with $m\geq 1$, if $w$ does not start by $\max(\Sigma)$, then
$$\pi_k(w) \sim \alpha(w)\cdot |\Sigma|^{\beta(w) k}$$
as $k\to\infty$, for some constants $\alpha(w),\beta(w)>0$.
\end{conjecture}

This conjecture, if true, would provide yet another way of approximating $\pi_k(w)$, in addition to Conjecture~\ref{conj:courbe_limite} --- and assuming that $\alpha(w)$ and $\beta(w)$ are easy to compute. Again, we shall not pursue this question further in this article.

\subsection{Imbalance in practice}\label{ss:imabalance}

In this section, we consider typical bioinformatics datasets and study the relationship between the empirical partitions and what the theory predicts. This section is also intended as a discussion rather than an exhaustive study, the aim of which is to raise directions for future research that would use our theoretical approach to improve in practice the partition methods used in bioinformatics.

\paragraph{Datasets}

Table~\ref{tab:datasets}  gives, for each dataset, the number of $k$-mers, as well as the number of minimizers encountered for different values of $m$.

\begin{table}[h!]
    \centering
        \resizebox{\textwidth}{!}{%
\begin{tabular}{c|cccccccc}
      \multirow{2}{*}{Dataset} & \multirow{2}{*}{$k$}& \multirow{2}{*}{Number of $k$-mers} &  \multicolumn{6}{c}{Number of minimizers}\\
       &&&$m=8$ & $m=9$& $m=10$& $m=11$& $m=12$& $m=21$\\
      \hline
    Human chromosome 1 (GRCh38) & $31$ & $204\,155\,258$& \xmark & \xmark & $332\,999$& $1\,001\,678$&$2\,740\,352$ & $37\,716\,694$\\
    Chromosome Y (GRCh38) & $31$ & $17\,680\,259$& $20\,612$& $62\,197$& $173\,673$&  $426\,104$& $862\,743$&\xmark\\
    \emph{Escherichia Coli (\href{https://www.ncbi.nlm.nih.gov/datasets/genome/GCF_000273425.1/}{MG12655\_V1}}) & $21$ & $4\,543\,786$ & $26\,344$&  $84\,431$ & $232\,361$ & $491\,574$& $761\,243$& \xmark \\
    RNA human lung dataset (\href{https://www.ncbi.nlm.nih.gov/sra/SRR8616107}{SRR8616107}) & $61$ & $190\,917\,566$& $17\,410$& $55\,118$ & $175\,405$& $532\,489$&\xmark & \xmark\\
    \end{tabular}
    }
    \caption{Statistics of the datasets}
    \label{tab:datasets}
\end{table}

To avoid overloading the paper with too many figures, we are not going to show the results for all possible values of $m$ for each dataset. However, the data and scripts for reproducing all the figures are available on Github\footnote{At \url{https://github.com/fingels/minimizer_counting_function}.}. Note that, overall, for close values of $m$ (i.e. all of them except $21$), the observations are similar --- therefore the conclusions as well.

Please note that in this section, unlike the results presented in Sections~\ref{ss:partition_theorique} and~\ref{ss:approximation}, the theoretical values $\pi_k(w)$ have only been calculated for the minimizers actually encountered in the data, and not for all the $|\Sigma|^m$ possible minimizers. In particular, please pay attention to the $x$-axis on the figures: unlike, say, Figure~\ref{fig:partition_theory}, the figures in this section --- e.g., upcoming Figure~\ref{fig:empirical_partition_vs_theory}, it is clear that the entire space of minimizers is not explored; with a bias in favour of minimizers starting with an \ch{A}. 

\paragraph{Comparing partitions} What we did for each minimizer $w$ we encountered in the data was to compute the theoretical value $\pi_k(w)$ and compare this with the observed empirical value $\widehat{\pi_k}(w)$ --- where $\widehat{\pi_k}(w)$ counts the number of $k$-mers from the data that admit $w$ as a minimizer. The results are shown in Figure~\ref{fig:empirical_partition_vs_theory}.

\begin{figure}[h!]
    \centering
\begin{subfigure}{0.49\textwidth}
\centering
\includegraphics[width=\textwidth]{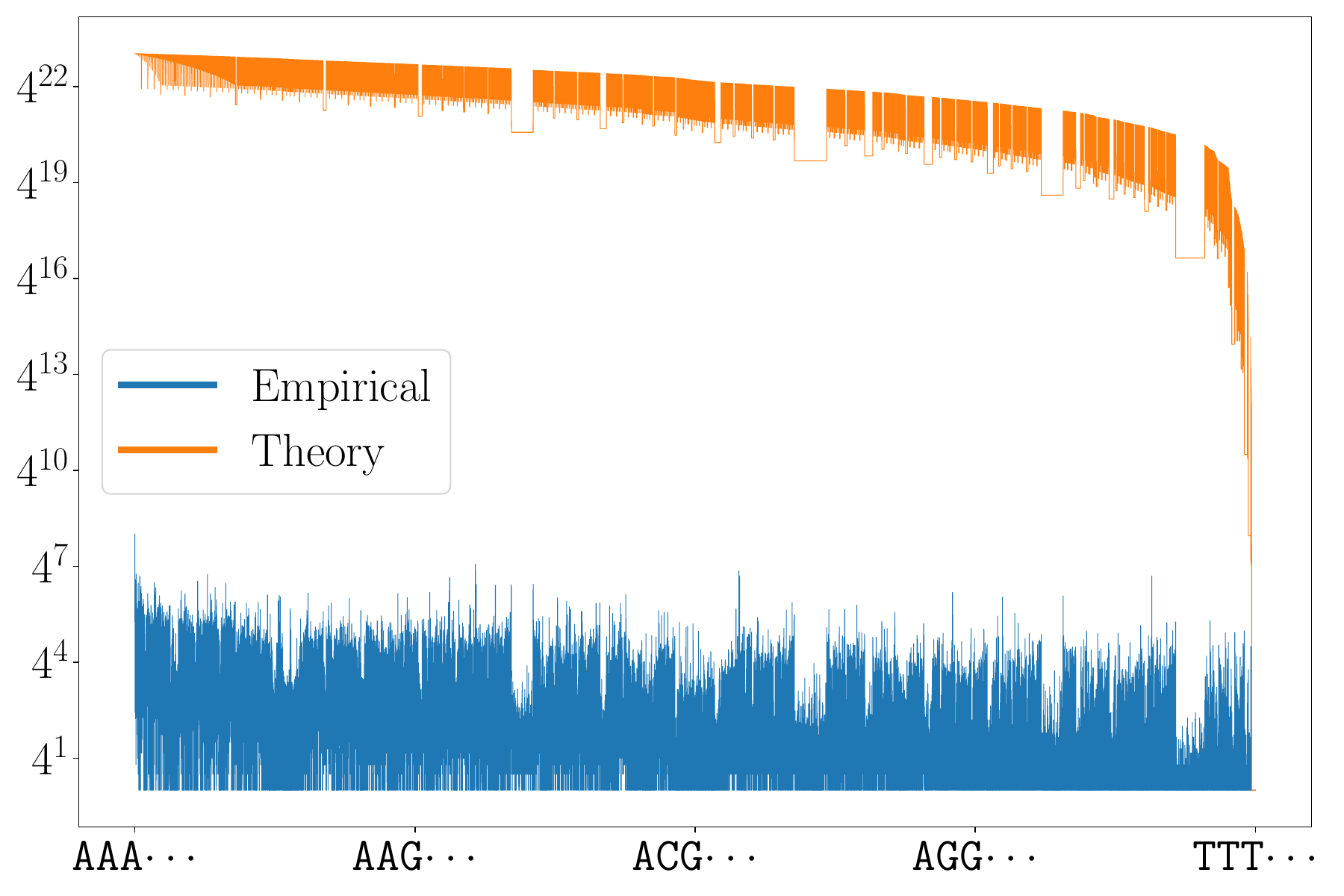}
    \caption{Chromosome Y, $k=31, m=10$}
    \label{fig:empirical_partition_vs_theory:a} 
\end{subfigure}\hfill
\begin{subfigure}{0.49\textwidth}
\centering
\begin{tikzpicture}
\node at (0,0) {\includegraphics[width=\textwidth]{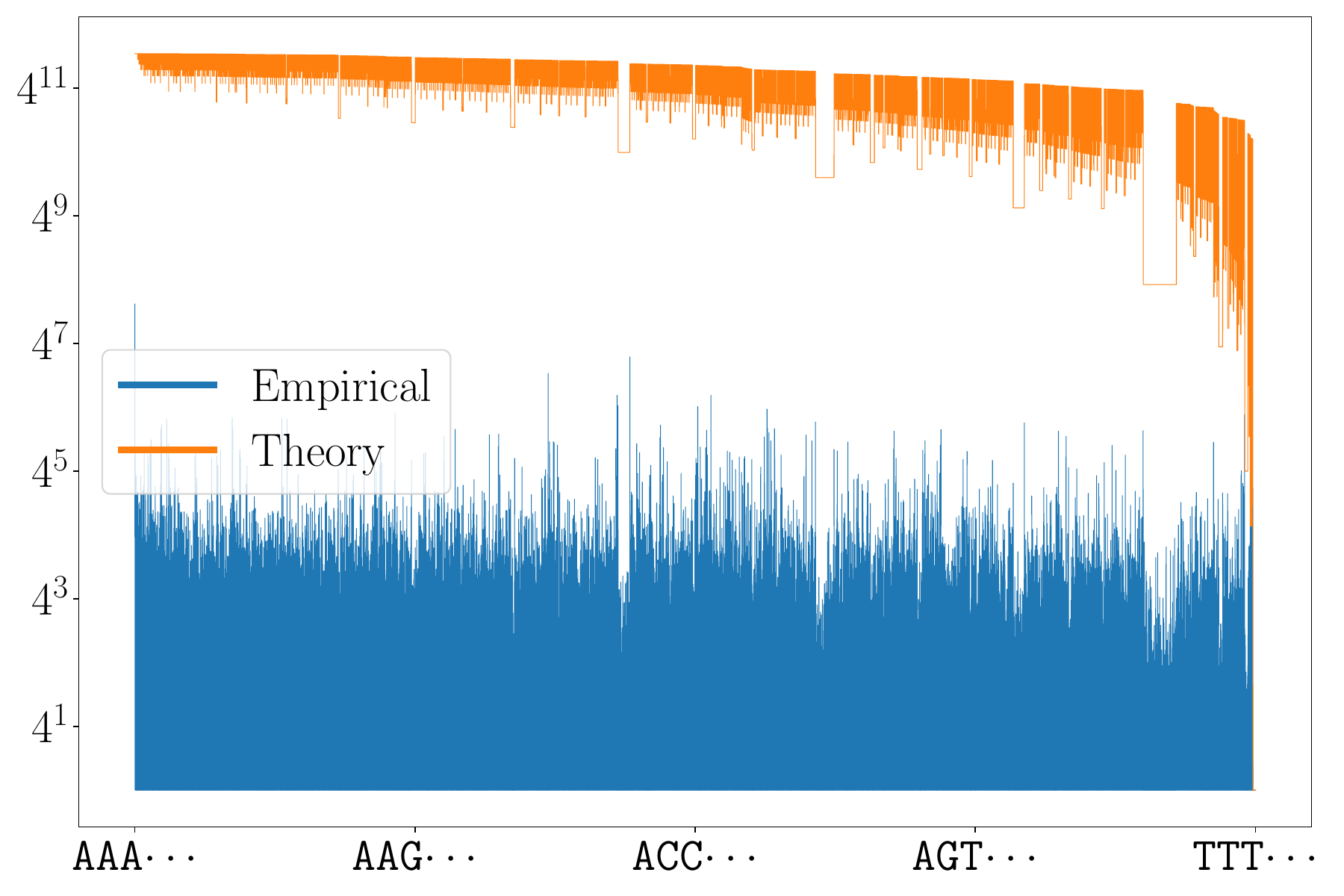}};
\draw[->,>=latex,ultra thick, red] (2,1)--(2.75,1);
\end{tikzpicture}
    \caption{Chromosome 1, $k=31, m=21$}
    \label{fig:empirical_partition_vs_theory:b} 
\end{subfigure}

\begin{subfigure}{0.49\textwidth}
\centering
\includegraphics[width=\textwidth]{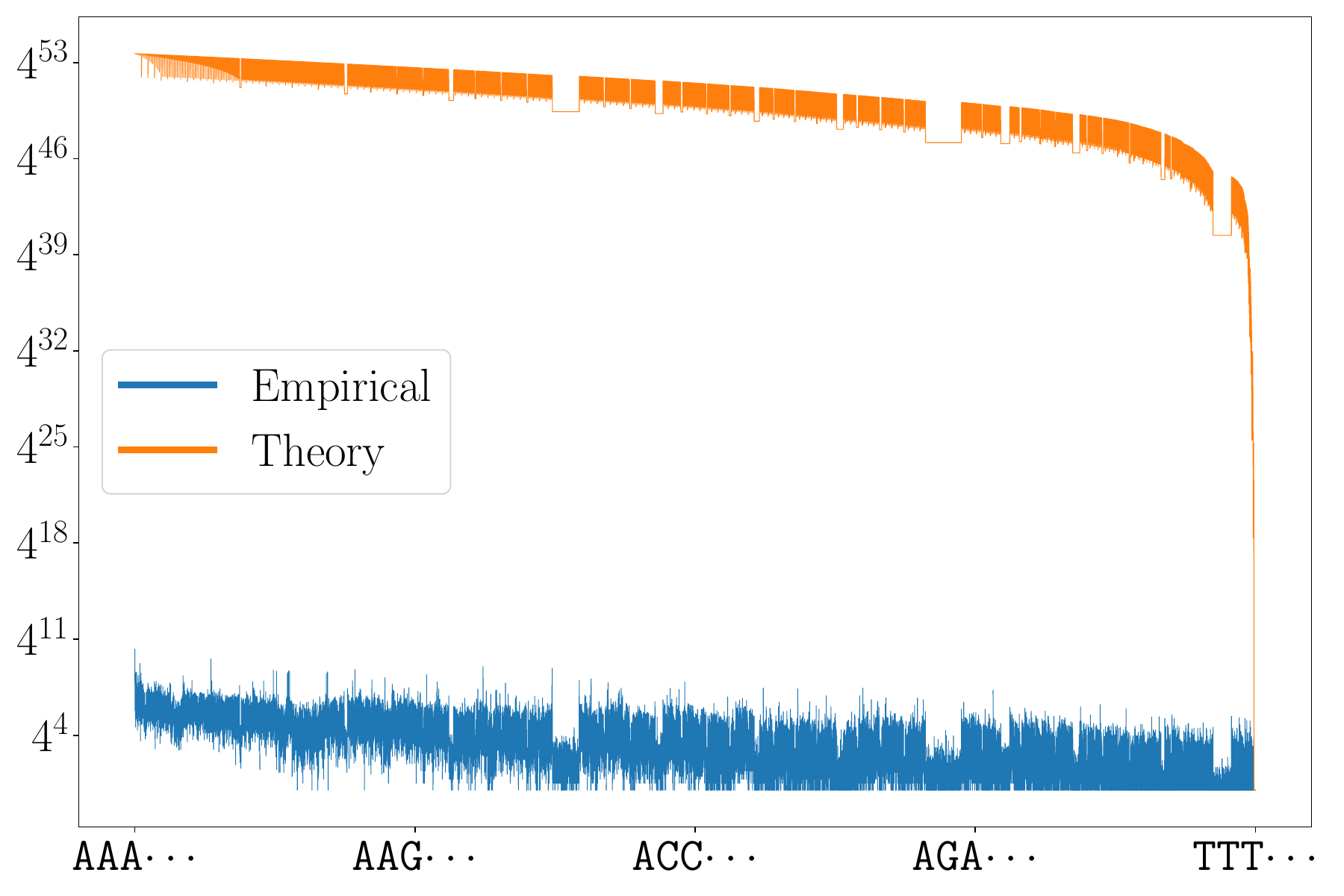}
    \caption{RNA fusion, $k=61, m=10$}
    \label{fig:empirical_partition_vs_theory:c} 
\end{subfigure}\hfill
\begin{subfigure}{0.49\textwidth}
\centering
\includegraphics[width=\textwidth]{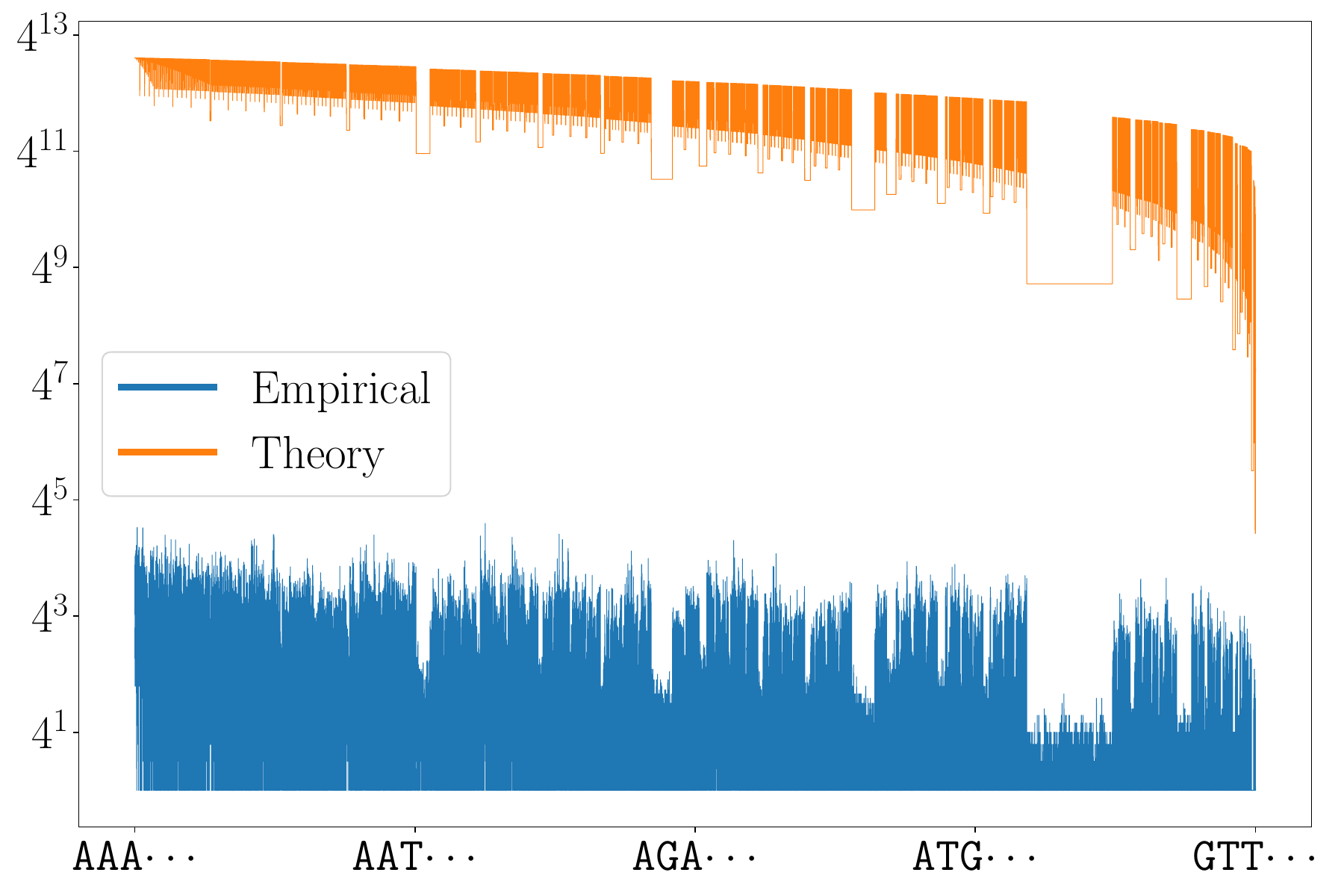}
    \caption{\emph{Escherichia Coli}, $k=21, m=10$}
    \label{fig:empirical_partition_vs_theory:d} 
\end{subfigure}
    \caption{Comparing the empirical partition $\widehat{\pi_k}(w)$ and the theory $\pi_k(w)$.}
    \label{fig:empirical_partition_vs_theory}
\end{figure}

As one can see, we are a long way from filling the buckets.  Remarkably, local decreases in $\pi_k(w)$ are also visible in $\widehat{\pi_k}(w)$ (for example, the plateau at the end of the distribution in Figure~\ref{fig:empirical_partition_vs_theory:d}). This indicates that the values observed empirically are not totally uncorrelated with what the theory predicts --- even if this is not true in all cases; for example, the local drop at the end of the distribution in Figure~\ref{fig:empirical_partition_vs_theory:b}, indicated by a red arrow, is not as obvious empirically as it is in theory. These alignments between theory and practice become all the more apparent  in upcoming Figure~\ref{fig:frequences}.

\paragraph{Comparing frequencies} Since the empirical and theoretical partitions appear to be correlated, we are interested here in comparing frequencies, i.e.

$$f_k(w) = \frac{\pi_k(w)}{|\Sigma|^k} \quad \text{and} \quad \widehat{f_k}(w)=\frac{\widehat{\pi_k}(w)}{\text{Number of $k$-mers in the data}}.$$

We compare them on a logarithmic scale, which in practice means translating the curves in Figure~\ref{fig:empirical_partition_vs_theory}; this brings us to Figure~\ref{fig:frequences}.

\begin{figure}[h!]
    \centering
\begin{subfigure}{0.49\textwidth}
\centering
\includegraphics[width=\textwidth]{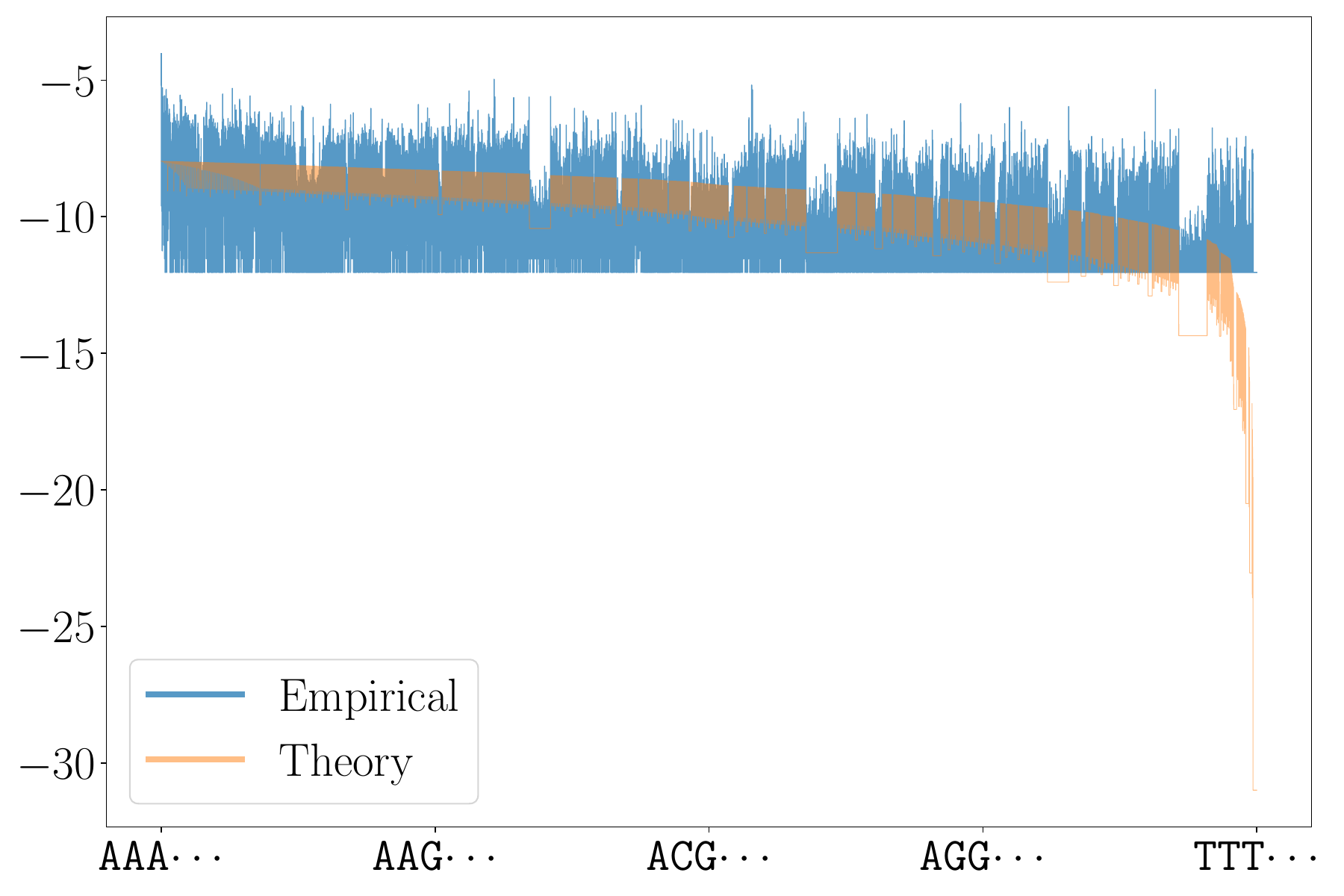}
    \caption{Chromosome Y, $k=31, m=10$}
    \label{fig:frequences:a} 
\end{subfigure}\hfill
\begin{subfigure}{0.49\textwidth}
\centering
\includegraphics[width=\textwidth]{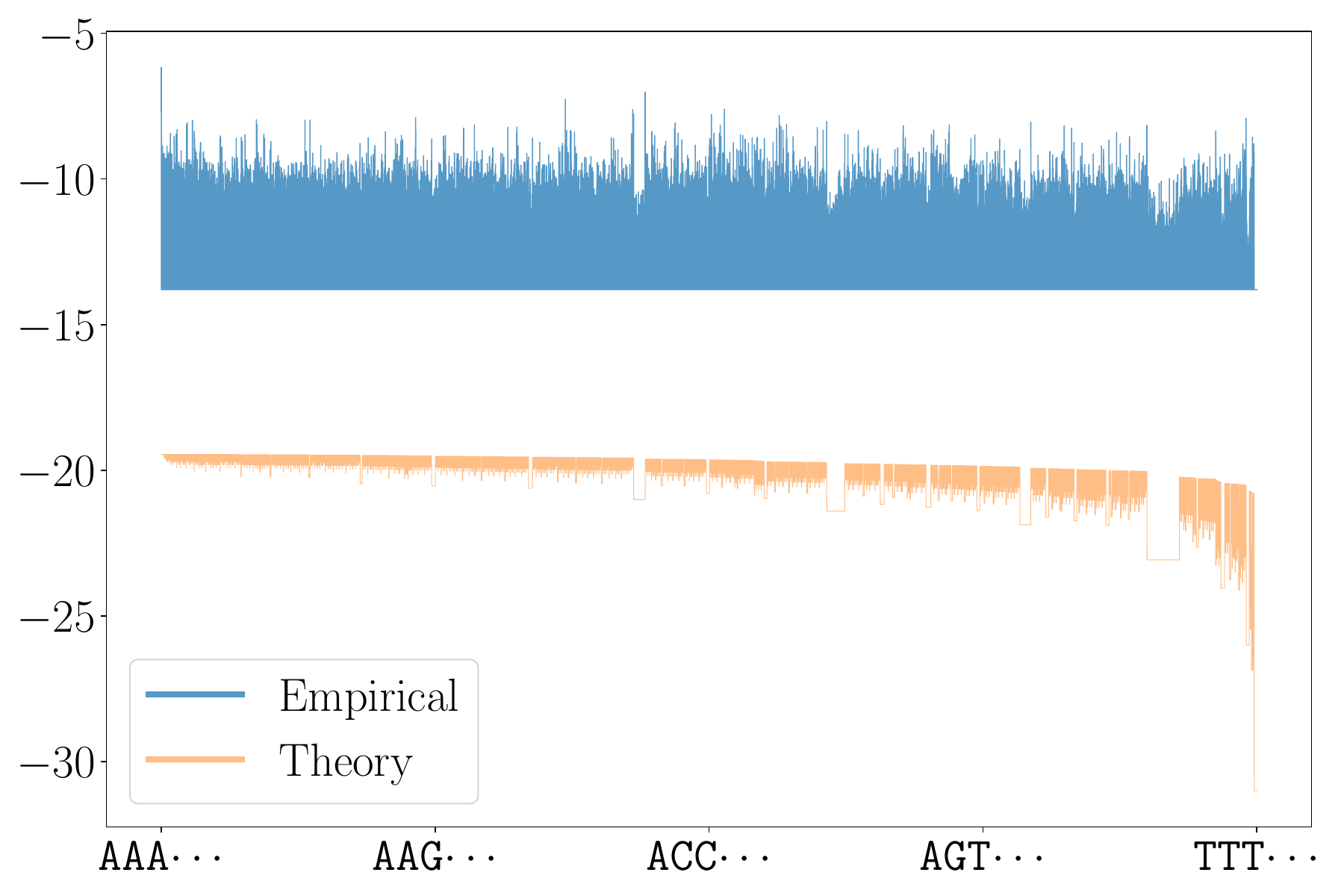}
    \caption{Chromosome 1, $k=31, m=21$}
    \label{fig:frequences:b} 
\end{subfigure}

\begin{subfigure}{0.49\textwidth}
\centering
\includegraphics[width=\textwidth]{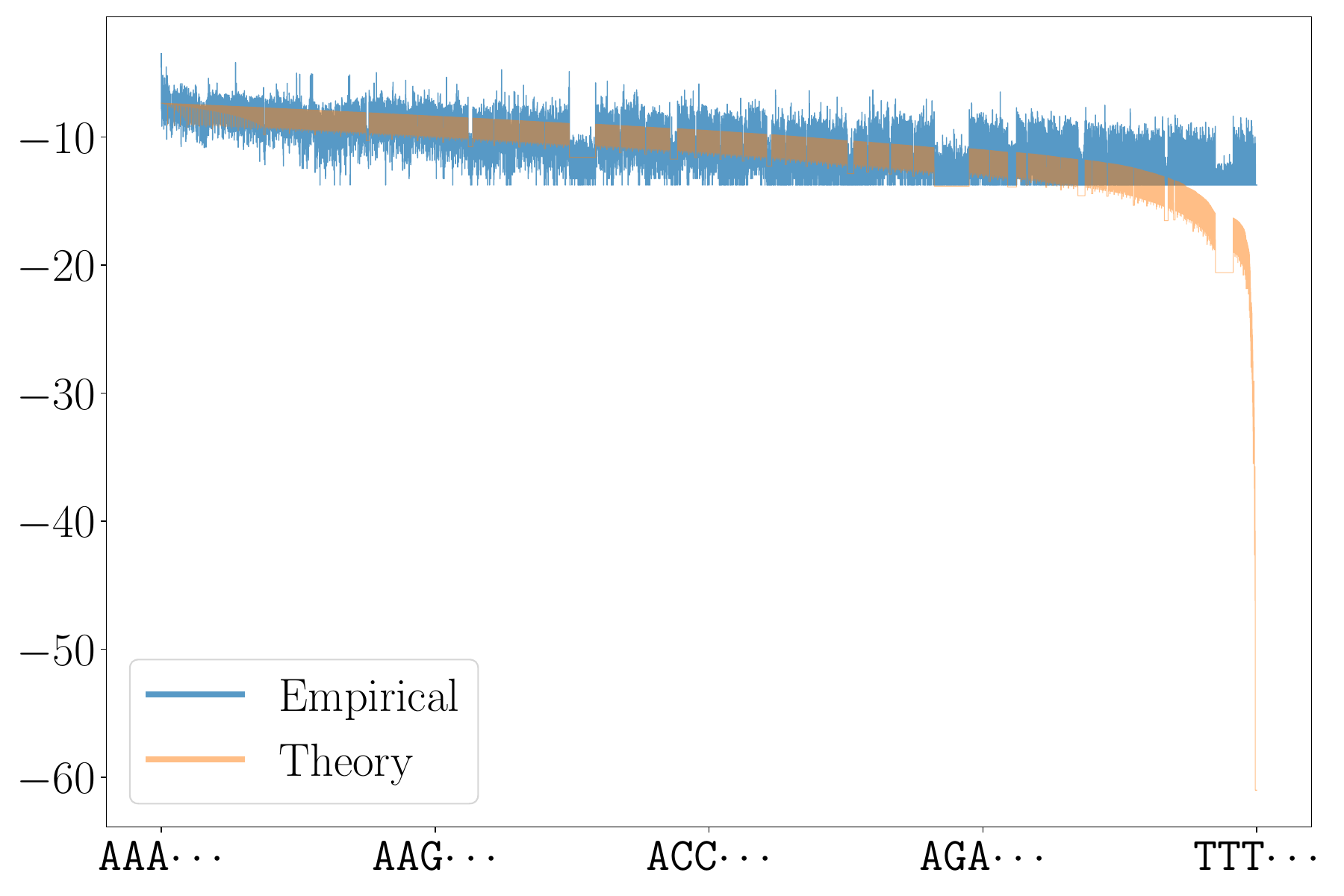}
    \caption{RNA fusion, $k=61, m=10$}
    \label{fig:frequences:c} 
\end{subfigure}\hfill
\begin{subfigure}{0.49\textwidth}
\centering
\includegraphics[width=\textwidth]{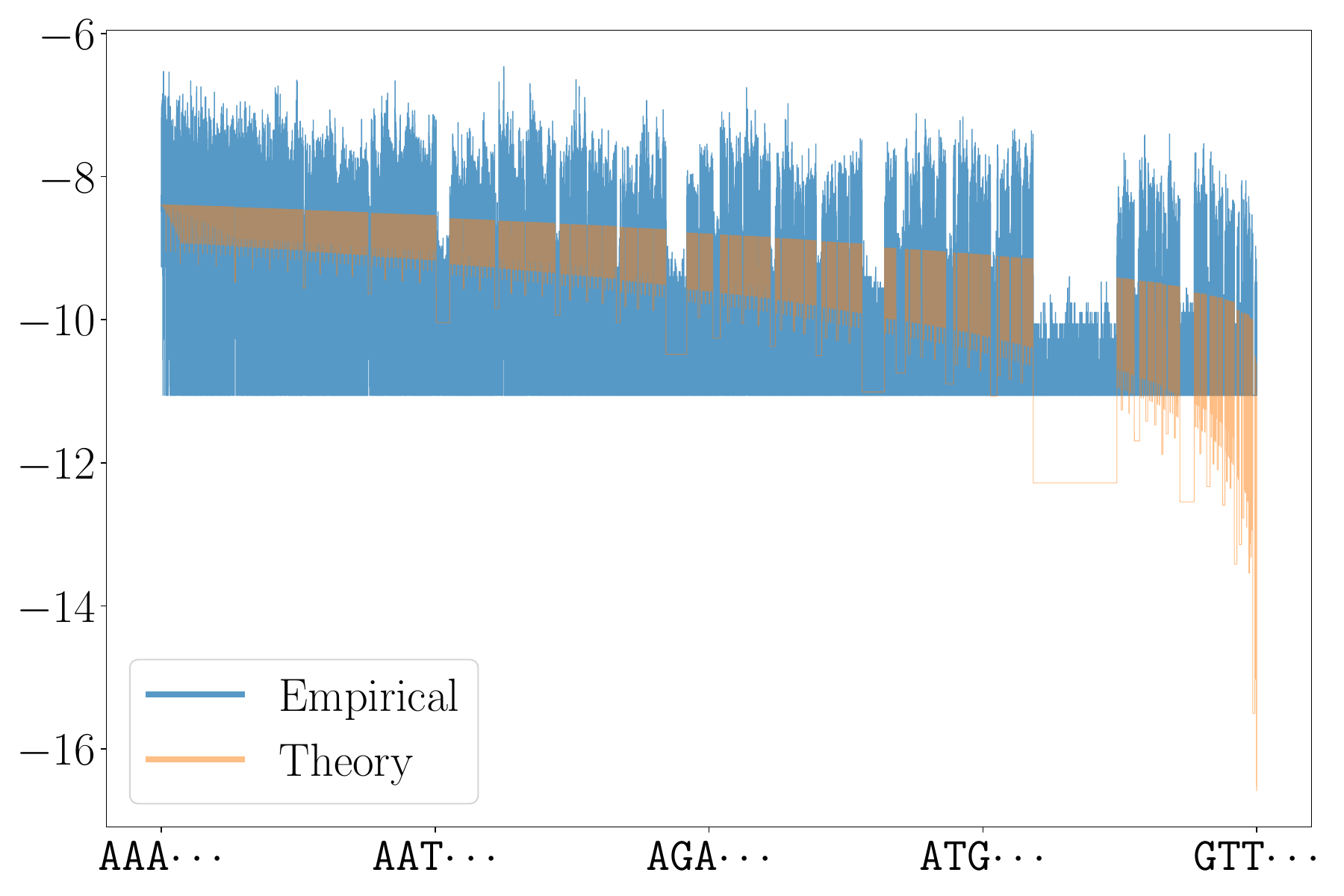}
    \caption{\emph{Escherichia Coli},$k=21, m=10$}
    \label{fig:frequences:d} 
\end{subfigure}
    \caption{Comparing the empirical frequencies $\log_{|\Sigma|}\widehat{f_k}(w)$ and the theory $\log_{|\Sigma|}f_k(w)$.}
    \label{fig:frequences}
\end{figure}

Quite remarkably, with the exception of Figure~\ref{fig:frequences:b}, we see that the theoretical frequency  seems to approximate the empirical frequency, particularly for small minimizers --- and not so much for large ones. Notably the local drops in the number of $k$-mers coincide very well between the empirical and theoretical distribution. It is not our ambition to quantify precisely to what extent $f_k(w)$ can correctly approximate $\widehat{f_k}(w)$ in this article, but it opens up an interesting avenue: using the theoretical frequency (or a transformation of it) to predict the empirical frequency. We can imagine partition heuristics that would use this oracle to guide the choice of minimizers, for example. We leave this question open for future work.

\paragraph{Quantifying the influence of $m$} We have seen that $f_k(w)$ and $\widehat{f_k}(w)$ do not align as well as for the other data in Figure~\ref{fig:frequences:b}. Among the data shown, this is the only one for which we have $m=21$. If we generate the same figure exclusively for this dataset (human chromosome 1) but with varying values of mm, we obtain Figure~\ref{fig:chr1}.

\begin{figure}[h!]
    \centering
\begin{subfigure}{0.49\textwidth}
\centering
\includegraphics[width=\textwidth]{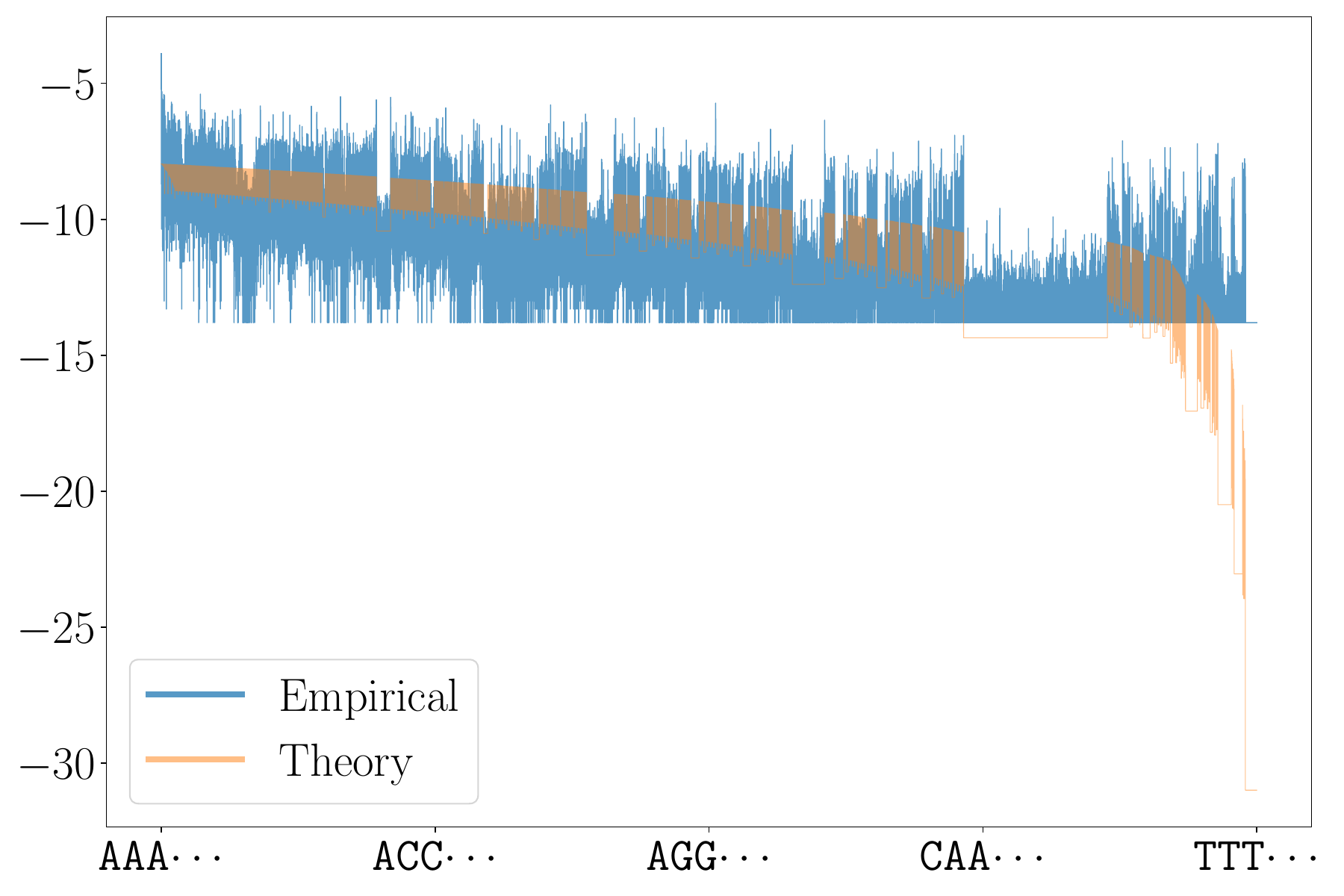}
    \caption{$m=10$}
    \label{fig:chr1:a} 
\end{subfigure}\hfill
\begin{subfigure}{0.49\textwidth}
\centering
\includegraphics[width=\textwidth]{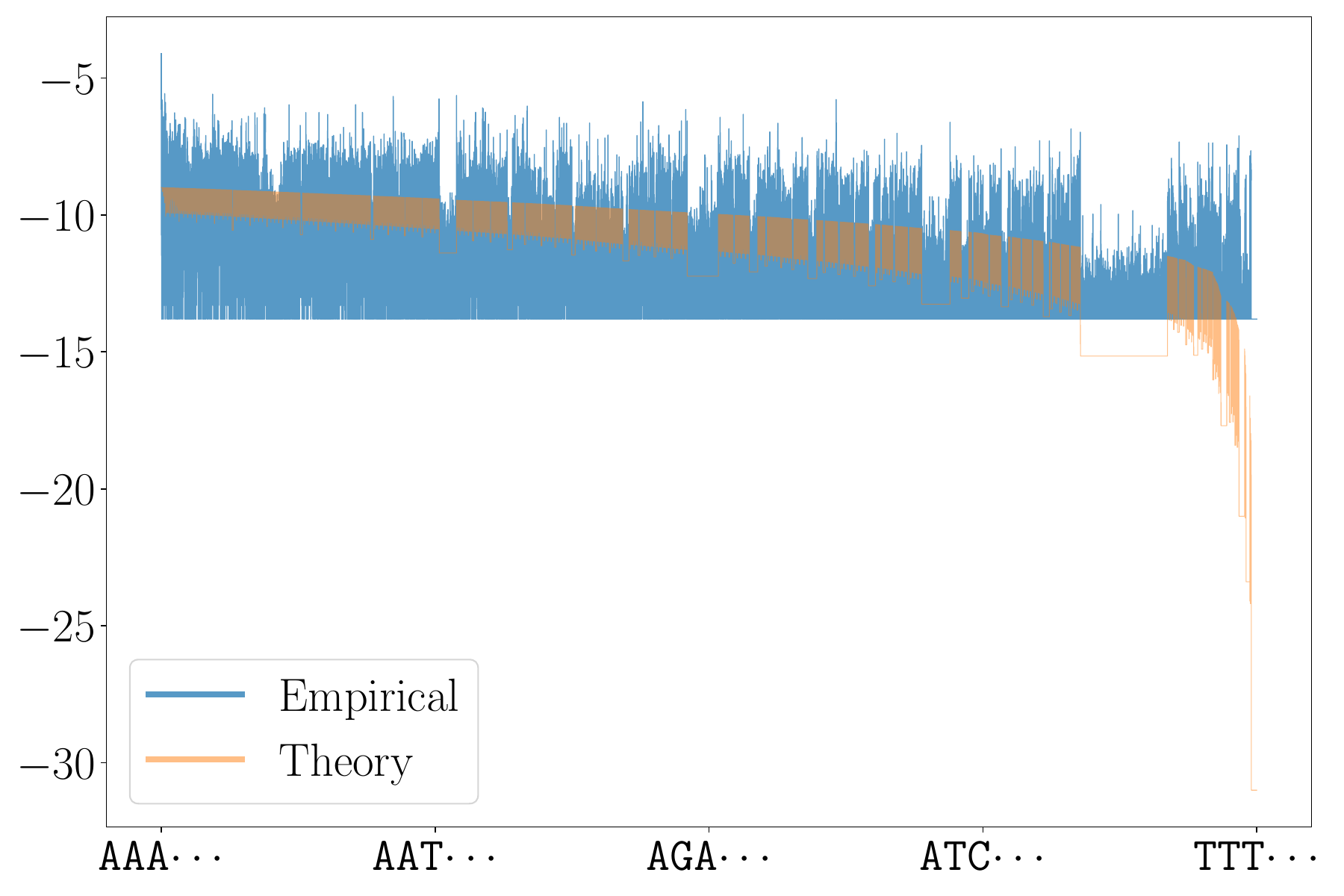}
    \caption{$m=11$}
    \label{fig:chr1:b} 
\end{subfigure}

\begin{subfigure}{0.49\textwidth}
\centering
\includegraphics[width=\textwidth]{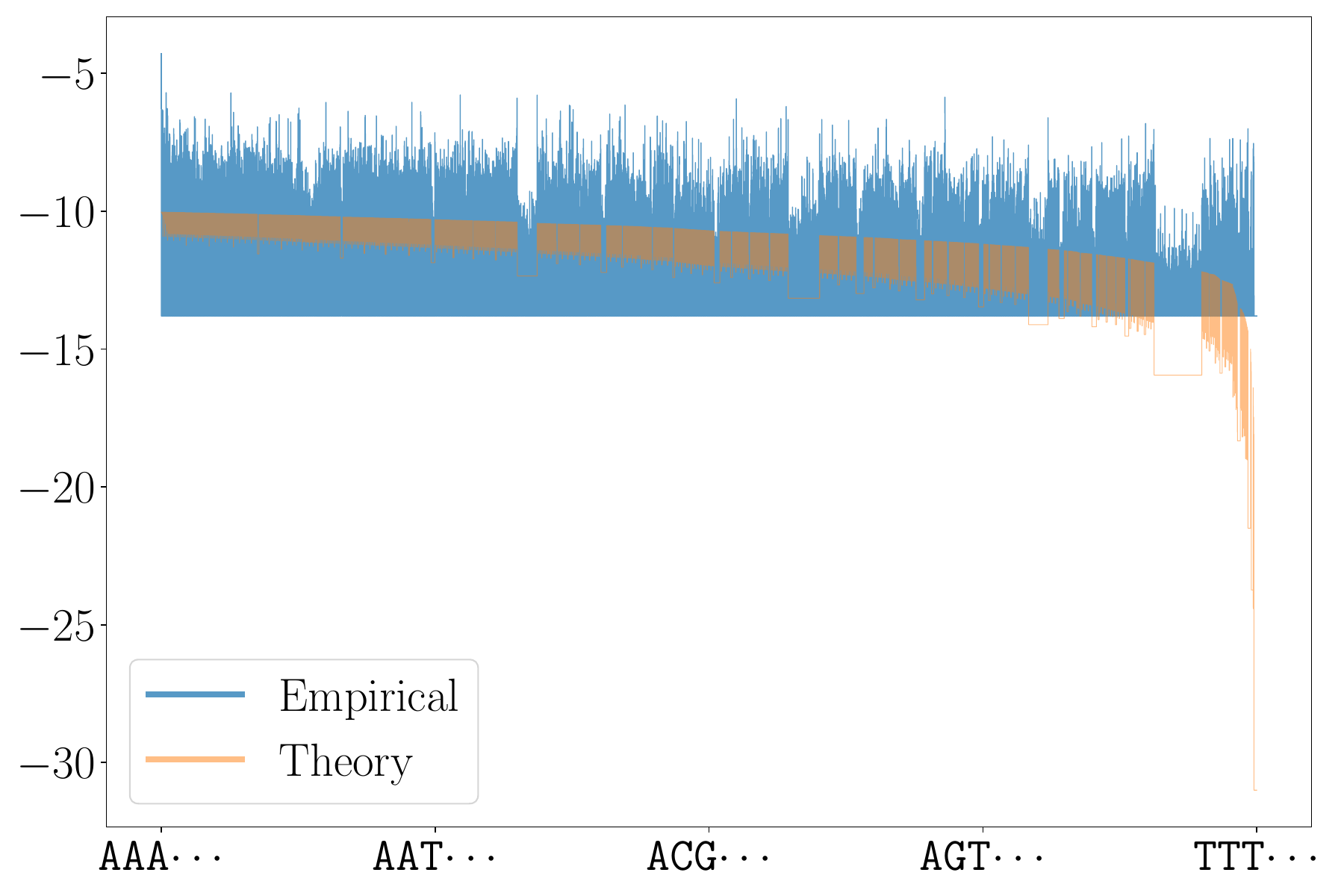}
    \caption{$m=12$}
    \label{fig:chr1:c} 
\end{subfigure}\hfill
\begin{subfigure}{0.49\textwidth}
\centering
\includegraphics[width=\textwidth]{figures_pratique/frequences_Hg_chr1_m21_k=31_m=21.pdf}
    \caption{$m=21$}
    \label{fig:chr1:d} 
\end{subfigure}
    \caption{Comparing the empirical frequencies $\log_{|\Sigma|}\widehat{f_k}(w)$ and the theory $\log_{|\Sigma|}f_k(w)$ for the chromosome 1 dataset, with different values of $m$, and $k=31$.}
    \label{fig:chr1}
\end{figure}

We can already see, simply by looking at the $x$-axis, that as $m$ increases, the minimizers selected become increasingly smaller. Furthermore, the theoretical curve seems to shift downwards as $m$ increases, explaining the poor fit when $m=21$. However, one can imagine looking for a renormalization factor depending on $m$ to ``pull up'' the theoretical curve --- if one wish to use $f_k(w)$ as an oracle for $\widehat{f_k}(w)$, as discussed above.

\section*{Conclusion}

In this article, we focused on quantifying theoretically $\pi_k(w)$, i.e. how many $k$-mers admit a given $m$-mer $w$ as a lexicographic minimizer. In a context of partitioning $k$-mer according to their minimizer, this amounts to calculating the size of the associated bucket in the worst case. Most of the article has been devoted to establishing the recursive equations for calculating these quantities. We showed that $\pi_k(w)$ can be computed in $O(km)$ space and $O(km^2)$ time, and proposed approximations that can be computed in $O(k)$ space and $O(km)$ time.

The results we obtained numerically in Section~\ref{sec:numerical} open up a number of perspectives for further research, which we hope to explore in the future. In particular, from a theoretical point of view:
\begin{description}
\item[Conjecture~\ref{conj:courbe_limite}] We found that the functions $ w\in\Sigma^m \mapsto \pi_k(w)$ are highly similar, and when renormalised correctly, they overlap almost perfectly. Is there a limit function to which those functions converge?
\item[Conjecture~\ref{conj:asymptotics}] We also found that, at a fixed $w$, the function $k\mapsto \pi_k(w)$ was remarkably well approximated by an exponential function. Is it possible to derive theoretically that $\pi_k(w)\sim \alpha(w) \cdot |\Sigma|^{\beta(w)k}$ for some constants $\alpha(w),\beta(w)$?
\end{description}

When considering genomic data, we found that, for small values of $m$, the empirical frequency was susceptible of being approximated by the theoretical frequency --- and that local variations in $\pi_k(w)$ could be retrieved empirically. As detailed in the introduction, it should be remembered that there are arguments in favour of using lexicographic minimizers, such as reduced storage and fine-grained comparisons, if only we could counterbalance the empirical skewness of the partitions. In the future, we intend to pursue our work in this direction, using the theoretical frequency as an oracle for the empirical frequency, and devise new partitions heuristics.

More generally, we believe that a better theoretical understanding of the distribution of lexicographic minimizers can lead to the development of new relevant practical methods. One of the theoretical directions we think it would be useful to explore in the future is the link between minimizers of different sizes --- especially in light of the recent article \cite{alanko2024finimizers}. Suppose $w$ and $w'$ are two words of different length where $w'$ is a substring of $w$. Is there any link between $\pi_k(w')$ and $\pi_k(w)$?

We are also interested in the distribution of \emph{canonical} $k$-mers among lexicographic minimizers. A $k$-mer $x=a_1\cdots a_k$ is said to be canonical if $x\leq\rc(x)$, where $\rc(x) = \overline{a_k}\cdots \overline{a_1}$ is called the reverse complement of $x$, and where $\overline{\cdot}$ is a involution of the DNA alphabet, defined by $\overline{\ch{A}}=\ch{T}$ and $\overline{\ch{C}}=\ch{G}$. Canonical $k$-mers are absolutely central to bioinformatics --- and are specific to this discipline, see \cite{marcais2024k}.  It would be of utmost interest to also be able to compute a minimizer counting function for canonical $k$-mers, that we would define as
$$\pi_k^c(w) = \left|\lbrace x\in \Sigma^k : (\min(x)=w)\wedge (x \text{ is canonical}\rbrace\right|.$$
To give a brief overview of this question, we have enumerated, in brute force, all $k$-mers for $k=12$, and calculated, for those that were canonical, their lexicographic minimizer (with $m=4$), and compared the distribution obtained with that of ``regular'' $k$-mers. We obtained Figure~\ref{fig:canonical}.

\begin{figure}[h!]
    \centering

\begin{minipage}[c]{0.49\textwidth}
    \includegraphics[width=\textwidth]{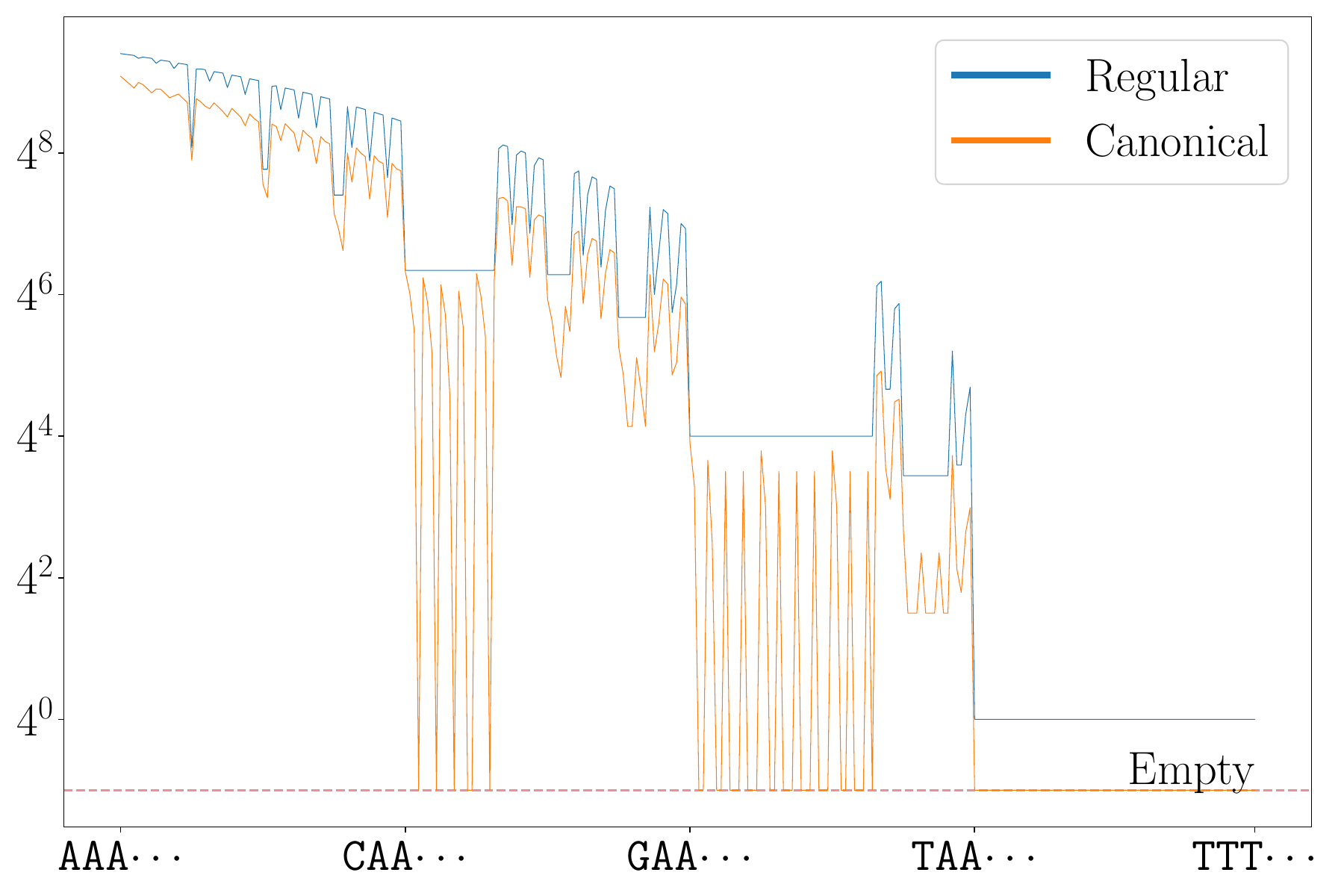}
\end{minipage}\hfill
\begin{minipage}[c]{0.49\textwidth}
    \caption{Distribution of $12$-mers among minimizers of size $m=4$ depending on whether they are canonical or not. The ``Empty'' line corresponds to minimizer that were not seen in the context of canonical $k$-mers but do exist nonetheless for regular $k$-mers.}
    \label{fig:canonical}
\end{minipage}

\end{figure}

Note that there are roughly half the number of canonical $k$-mers compared to regular one (see \cite{wittler2023general} for a precise count), which may explain why the canonical curve is globally below the regular curve. Observe that the plateaus disappear, as many minimizers are not found in the canonical $k$-mers. Nevertheless, we note that the curves are sufficiently different to justify our interest in canonical $k$-mers, and sufficiently close to each other to hope to use the same kind of methods developed in this article.

\paragraph{Closing remark} Through this work, we hope to convince the bioinformatics community that there are still things to explore in terms of lexicographic minimizers, with the hope of bringing them back into practical applications, in view of the perspectives mentioned above.

\section*{Acknowledgements}

This work is funded by a grant from the French ANR: Full-RNA ANR-22-CE45-0007.

\printbibliography

\appendix

\section{Preprocessing algorithm}\label{app:computation_autocorrelation_matrix}

In this section, we propose an algorithm for computing all the important quantities mentioned in the paper, i.e. 
\begin{multicols}{2}
\begin{itemize}
    \item $\mathbf{R}$ --- Definition~\ref{def:autocorrelation_matrix};
        \item $\beta_{\max}(w)$ --- Definition~\ref{def:postmer_beta_max};
    \item  $i_{\max}(w)$ --- Definition~\ref{def:prefix_max_size};
    \item  $\mathbf{T}_i(a)$ --- Definition~\ref{def:prefix_letter_vector};
    \item  $a_{\max}(i)$ --- Equation~\eqref{eq:a_max}; 
    \item $\Sigma_{=0}(i)$ and $\Sigma_{\neq0}(i)$ --- Proposition~\ref{prop:recursion_postmer_first_cases};
    \item $\widetilde{\mathbf{T}_i}(a,\beta)$ --- Definition~\ref{def:prefix_letter_vector_postmer};
    \item $\widetilde{a_{\max}}(i,\beta)$ --- Equation~\eqref{eq:a_max_tilde}.
\end{itemize}
\end{multicols}
While $\mathbf{R}$ can be computed in $O(m^2)$, it verifies the following properties that may speed up its computation.
\begin{lemma}
Let $1\leq i\leq m$. 
\begin{itemize}
    \item $\mathbf{R}_{i,1}^=$;
    \item if there exists $1\leq j\leq i$ so that $\mathbf{R}_{i,j}^<$ (resp. $\mathbf{R}_{i,j}^>$), then for all $i'>i$, $\mathbf{R}_{i',j}^<$ (resp. $\mathbf{R}_{i',j}^>$).
\end{itemize}
\end{lemma}
\begin{proof}
The first item is trivial. For the second one, it suffices to notice that, since $i'>i$, $a_j\cdots a_i$ is a prefix of $a_j\cdots a_{i'}$ and $a_1\cdots a_{i-j+1}$ is a prefix of $a_1\cdots  a_{i'-j+1}$. Therefore, as soon as either $a_j\cdots a_i$  or $a_1\cdots a_{i-j+1}$ becomes larger than the other, this propagates to the larger words of which they are prefixes.
\end{proof}

Algorithm~\ref{algo:preprocessing} allows to compute $\mathbf{R}$, $\beta_{\max}(w)$, $i_{\max}(w)$, $\mathbf{T}_i(a)$, and $a_{\max}(i)$ at the same time. Assuming $O(|\Sigma|)=O(1)$, the algorithm runs in $O(m^2)$. Note that the value of $\beta_{\max}(w)$ returned by the algorithm is not exactly $\beta_{\max(w)}$ as defined in Definition~\ref{def:postmer_beta_max}, to avoid the dependency on $k$. The actual value of $\beta_{\max}(w)$ is the minimum between $k-m$ and the value returned by the algorithm. We leave it to the reader to convince themselves of the correctness of the algorithm.

\begin{algorithm}[h!]
\caption{\textsc{Preprocessing}}\label{algo:preprocessing}
\KwIn{$w=a_1\cdots a_m$}
Initialize $\mathbf{R}$ as a lower triangular matrix of size $m$\;
$i_{\max}(w) \gets m$ and $\beta_{\max}(w)\gets \infty$\;
\For{$1\leq i\leq m$}{
$\mathbf{R}_{i,1}\gets ``="$\;
Initialize $\mathbf{T}_i$ as a vector of size $|\Sigma|$ filled with the value $\infty$\;
Initialize $a_{\max}(i)$ as the set $\lbrace a_1\rbrace$\;
}
\For{$a\in \Sigma$}{
$\mathbf{T}_1(a)\gets 2\cdot [a=a_1]$
}
\For{$2\leq j\leq m$}{
$b\gets \bot$ and $s\gets \emptyset$\;
\For{$j\leq i \leq m$}{
\If{$b=\bot$}{
$x \gets a_j\cdots a_i$ and $y \gets a_1\cdots a_{i-j+1}$\;
\uIf{$x=y$}{
$s\gets ``="$\;
$\mathbf{T}_i(a_{i-j+2}) \gets \min (\mathbf{T}_i(a_{i-j+2}),j+1)$\;
$a_{\max}(i)\gets a_{\max}(i)\cup \lbrace a_{i-j+2}\rbrace$\;}
\uElseIf{$x<y$}
{
$s\gets ``<"$\;
$b\gets \top$\;
$i_{\max}(w)\gets \min (i_{\max}(w), i)$\;
$\beta_{\max}(w)\gets \min (\beta_{\max}(w), j-2)$\;}
\Else{
$s\gets ``>"$\;
$b\gets \top$\;
}
}
$\mathbf{R}_{i,j}\gets s$
}
\For{$a\in\Sigma$}{
\If{$\mathbf{T}_j(a)=\infty$}{
$\mathbf{T}_j(a)\gets (j+1) \cdot [a=a_1]$\;
}
}
}
\For{$1\leq i\leq m$}{
$a_{\max}(i)\gets \max(a_{\max}(i))$\;
}
\Return $\mathbf{R}, i_{\max}(w),\beta_{\max}(w),\big(\mathbf{T}_i\big)_{1\leq i\leq m},\big(a_{\max}(i)\big)_{1\leq i\leq m}$
\end{algorithm}

\begin{algorithm}[h!]
\caption{\textsc{PreprocessingBis}}\label{algo:preprocessingbis}
\KwIn{$w=a_1\cdots a_m$, $(\mathbf{T}_i)_{1\leq i \leq m}$}
\For{$1\leq i \leq m$}{
$\Sigma_{=0}(i)\gets \emptyset$ and $\Sigma_{\neq0}(i) \gets \emptyset$\\
Initialize $\widetilde{\mathbf{T}_i}$ as a vector of size $|\Sigma|$ filled with the zero function $\beta \mapsto 0$\\
Let $f : 0\mapsto \varepsilon$ be a dictionary\\
\For{$a\in\Sigma$}{
\eIf{$\mathbf{T}_i(a)\neq 0$}{
$\Sigma_{\neq0}(i)\gets \Sigma_{\neq0}(i)\cup \lbrace a\rbrace$\\
$\widetilde{\mathbf{T}_i}(a) \gets \big(\beta \mapsto \mathbf{T}_i(a) \cdot [\beta \geq m-1 +\mathbf{T}_i(a)]\big)$\tcp*{Currying}
Define $f(m-1 +\mathbf{T}_i(a)) = a$\\
}{
$\Sigma_{=0}(i)\gets \Sigma_{=0}(i)\cup \lbrace a\rbrace$
}}
$\Sigma_{\neq0}(i)\gets \Sigma_{\neq0}(i) \setminus \lbrace a_i\rbrace$\\
$\Sigma_{=0}(i)\gets \Sigma_{=0}(i) \setminus \lbrace a_i\rbrace$\\
Let $x_1,\dots,x_n$ be the values for which $f$ is defined, sorted in increasing order\\ 
Let $I$ be a table of size $n$ with $I[j] = \max(\lbrace f(x_{j'}) : j'\leq j \rbrace)$\\
$\widetilde{a_{\max}}(i) \gets \left(\beta \mapsto \left(\displaystyle\sum_{j=1}^{n-1} I[j] \cdot [x_j \leq \beta < x_{j+1}] \right) + I[n]\cdot [\beta \geq x_n]\right)$\tcp*{Currying}}
\Return $\big(\Sigma_{\neq0}(i)\big)_{1\leq i \leq m-2},\big(\Sigma_{=0}(i)\big)_{1\leq i \leq m-2}$, $\big(\widetilde{\mathbf{T}_i}\big)_{1\leq i \leq m},\big(\widetilde{a_{\max}}(i)\big)_{1\leq i \leq m}$
\end{algorithm}

The remaining quantities to be calculated, $\Sigma_{\neq0}(i)$, $\Sigma_{=0}(i)$, $\widetilde{\mathbf{T}_i}(a,\beta)$, and $\widetilde{a_{\max}}(i)$, are developed in Algorithm~\ref{algo:preprocessingbis}, running in $O(m)$ --- once again assuming $O(|\Sigma|)=O(1)$. $\Sigma_{\neq0}(i)$ and $\Sigma_{=0}(i)$ can be computed in a straightforward fashion from their definition --- see Proposition~\ref{prop:recursion_postmer_first_cases}. Although we only need the values for $1\leq i \leq m-2$, we decided to include the calculation for $m-1$ and $m$ in the pseudocode to avoid burdening it with extra conditions. 

Concerning $\widetilde{\mathbf{T}_i}(a,\beta)$, recall from Definition~\ref{def:prefix_letter_vector_postmer} that
$$\widetilde{\mathbf{T}_i}(a,\beta) = \mathbf{T}_i(a) \cdot [\beta \geq m-1 +\mathbf{T}_i(a)]$$
To avoid conditional loops every time we need to know whether $\widetilde{\mathbf{T}_i}(a,\beta)$ is 0 or not depending on the value of $\beta$, we prefer to use currying \cite{curry1980some} and define the following functions:
$$\widetilde{\mathbf{T}_i}(a) : \big(\beta \mapsto \mathbf{T}_i(a) \cdot [\beta \geq m-1 +\mathbf{T}_i(a)]\big),$$
where $\widetilde{\mathbf{T}_i}(a,\beta)$ in the text of the paper is replaced by calling $\widetilde{\mathbf{T}_i}(a)(\beta)$ in the implementation. Obviously, if $\mathbf{T}_i(a)=0$, then $\widetilde{\mathbf{T}_i}(a,\beta)=0$ for any value of $\beta$, so in this case we directly use $\widetilde{\mathbf{T}_i}(a) : (\beta \mapsto 0)$ to accelerate evaluation.

For $\widetilde{a_{\max}}(i,\beta)$, from Equation~\eqref{eq:a_max_tilde} and the previous discussion, we have
\begin{align*}
 \widetilde{a_{\max}}(i,\beta) &=\max \left(\lbrace \varepsilon \rbrace \cup \left\lbrace a \in \Sigma : \widetilde{\mathbf{T}_i}(a,\beta)\neq 0\right\rbrace\right)\\
 &= \max\left(\lbrace \varepsilon\rbrace \cup \big\lbrace a\in \Sigma: (\mathbf{T}_i(a)\neq 0) \wedge (\beta\geq m-1+\mathbf{T}_i(a)) \big\rbrace\right)
\end{align*}
Consider the list of pairs $(a,m-1+\mathbf{T}_i(a))$ for which $\mathbf{T}_i(a)\neq 0$, also containing the pair $(\varepsilon,0)$. Sorting the second item of each pair in increasing order, we get a list $x_1,\dots,x_n$ (with $x_1=0$). Note that since there cannot be two letters $a\neq b$ with $\mathbf{T}_i(a)=\mathbf{T}_i(b)$ (as soon as these values are not $0$, of course), the $x_j$'s are unique. Note also that $n\leq |\Sigma|+1$.

Let us denote $b_j$ the letter associated to $x_j$ (hence $b_1=\varepsilon$). $b_j$ belongs to the set used to evaluate $\widetilde{a_{\max}}(i,\beta)$ if and only if $\beta \geq x_j$. Since the $x_j$'s are sorted, if $x_{j+1}>\beta \geq x_j$, then $\widetilde{a_{\max}}(i,\beta)$ is exactly the maximum letter among $\lbrace b_1,\dots,b_j\rbrace$, as illustrated below.
\begin{center}
 \begin{tikzpicture}[>=latex]
\draw (0,0)--(6.5,0);
\draw[dashed] (6.5,0)--(8.5,0);
\draw[->] (8.5,0)--(15,0);
\node[right=8pt] at (15,0) {$\beta$};
\node[] at (0,0) {$\Big[$};
\node[below=8pt] at (0,0) {$x_1=0$};
\node[above=8pt] at (1.5,0) {$\varepsilon$};
\node[] at (3,0) {$\Big[$};
\node[below=8pt] at (3,0) {$x_2$};
\node[above=8pt] at (4.5,0) {$\varepsilon, b_2$};
\node[] at (6,0) {$\Big[$};
\node[below=8pt] at (6,0) {$x_3$};
\node[] at (9,0) {$\Big[$};
\node[below=8pt] at (9,0) {$x_{n-1}$};
\node[above=8pt] at (10.5,0) {$\varepsilon, b_2,\dots,b_{n-1}$};
\node[] at (12,0) {$\Big[$};
\node[below=8pt] at (12,0) {$x_{n}$};
\node[above=8pt] at (13.5,0) {$\varepsilon, b_2,\dots, b_n$};
\end{tikzpicture}  
\end{center}
To match the notations of the pseudocode, let us denote $I[j] = \max \lbrace b_1\cdots b_j\rbrace$. We now have a piecewise reformulation of $\widetilde{a_{\max}}(i,\beta)$:
$$\widetilde{a_{\max}}(i,\beta) = I[1] \cdot [x_1 \leq \beta < x_2] + \cdots + I[n-1] \cdot [x_{n-1}\leq \beta < x_n] + I[n] \cdot [\beta \geq x_n]$$
where addition is understood as character concatenation, and where $a\cdot[P] = a$ if $P$ is true, and $\varepsilon$ otherwise.

Once again, currying is used in practice, where $\widetilde{a_{\max}}(i) : \big(\beta \mapsto \widetilde{a_{\max}}(i,\beta)\big)$ (using the above reformulation) and occurrences of $\widetilde{a_{\max}}(i,\beta)$ in the equations are transformed into evaluations of $\widetilde{a_{\max}}(i)(\beta)$ in the implementation.

\begin{example}
From Example~\ref{ex:postmer_prefix_vectors}, for $w_2=\ch{ACACAC}$ and $i=5$, we have
$$\widetilde{\mathbf{T}_i}(\ch{A},\beta) = 6 \cdot [\beta \geq 11] \quad \text{and} \quad \widetilde{\mathbf{T}_i}(\ch{C},\beta) = 3 \cdot [\beta \geq 8];$$
the other values being $0$.

Following the above procedure, we constitute the following list: $(\varepsilon,0)$, $(\ch{A}, 11)$ and $(\ch{C},8)$. Sorting the list by the second item, we have $x_1=0$, $x_1=8$ and $x_2=11$, as well as $b_1=\varepsilon$, $b_2=\ch{C}$ and $b_3=\ch{A}$. We then compute $I[1] = \varepsilon$ and $I[2]=I[3]=\ch{C}$, leading to 
$$\widetilde{a_{\max}}(i,\beta) = \varepsilon\cdot [0\leq \beta < 8] + \ch{C}\cdot [8\leq \beta < 11] + \ch{C} \cdot[\beta\geq 11].$$
\end{example}

\end{document}